\documentclass[runningheads]{llncs}
\usepackage[T1]{fontenc}
\usepackage{graphicx}
\usepackage{color}
\usepackage[font=small,labelfont=bf]{caption}
\usepackage{amsmath,amssymb,amsfonts}
\usepackage{graphicx}
\usepackage{textcomp}
\usepackage[binary-units=true]{siunitx}
\usepackage{comment}
\usepackage{amsmath}
\usepackage{hyperref}
\usepackage{etoolbox}
\usepackage{color,colortbl}
\usepackage{mathtools}
\usepackage{enumerate}
\usepackage{algorithm}
\usepackage[noend]{algpseudocode}
\usepackage{booktabs}
\usepackage[normalem]{ulem}
\usepackage{pgf}
\usepackage{pgfplotstable}
\usepackage{tikz,pgfplots}
\usepackage{thm-restate}
\usetikzlibrary{trees,decorations,arrows,arrows.meta,automata,shadows,positioning,plotmarks,backgrounds,shapes,shapes.misc}
\usetikzlibrary{calc,matrix,fit,petri,decorations.pathmorphing,patterns}
\usetikzlibrary{decorations.pathreplacing,decorations.markings,shapes.geometric,calc}
\usetikzlibrary{tikzmark}
\usepackage{stmaryrd}
\usepackage{xspace}
\usepackage{graphicx}
\usepackage{float}
\usepackage[utf8]{inputenc} 
  \usepackage{csquotes} 
\usepackage{multirow}
\usepackage{array}
\usepackage{dsfont}
 \usepackage{array}
\usepackage{xifthen}
\usepackage{relsize}
\usepackage{xfrac}
\usepackage{wrapfig}
\usepackage{rotating}
\usepackage{paralist}

\usepackage{longtable}
\usepackage{caption}
\usepackage[position=b]{subcaption}

\usepackage{expl3}[2012-07-08]
\ExplSyntaxOn
\cs_new_eq:NN \fpeval \fp_eval:n
\ExplSyntaxOff

\colorlet{darkgreen}{green!80!black}
\colorlet{darkred}{red!80!black}
\usetikzlibrary{arrows, automata, shapes}
\tikzset{auto, >= stealth}
\tikzset{every edge/.append style={thick, shorten >= 1pt}}
\tikzset{initial/.style={draw, thick, <-, shorten <=1pt}}
\tikzset{player0/.style = {draw, thick, shape=circle, minimum size=5mm}}
\tikzset{player1/.style = {draw, thick, shape=rectangle, minimum size=5mm}}
\tikzset{bplayer0/.style = {draw, thick, shape=ellipse, minimum size=5mm,text width=1.1cm}}
\tikzset{bplayer1/.style = {draw, thick, shape=rectangle, minimum size=5mm,text width=1.6cm}}
\newcommand\pos{1.4}
\newcommand\hpos{2.5}
\newcommand\ypos{1.85}
\newcommand*\circled[1]{~\tikz[baseline=(char.base)]{\node[shape=circle,draw, text width=3mm,align=center, minimum size = 2mm, inner sep = 0pt, fill = black] (char) {\scriptsize \textcolor{white}{\textbf{#1}}}}}

\usepackage{marginnote}
\usepackage{schemata}

\usepackage[
colorinlistoftodos, 
textwidth=\marginparwidth, 
textsize=scriptsize, 
]{todonotes}

\newcommand{\putbracket}[1]{\ifx&#1&%
	
	\else
	(#1)
	\fi}

\newcommand{\pre}[2]{\textsf{pre}\ensuremath{_{#1}(#2)}}
\newcommand{\tpre}[2]{\textsf{tpre}\ensuremath{_{#1}(#2)}}
\newcommand{\cpre}[3]{\textsf{cpre}\ensuremath{^{#3}_{#1}(#2)}}
\newcommand{\cprea}[2]{\textsf{cpre}\ensuremath{^a_{#1}(#2)}}
\newcommand{\cpreo}[2]{\textsf{cpre}\ensuremath{^1_{#1}(#2)}}
\newcommand{\cprez}[2]{\textsf{cpre}\ensuremath{^0_{#1}(#2)}}
\newcommand{\attr}[3]{\textsf{attr}\ensuremath{^{#3}_{#1}(#2)}}
\newcommand{\attra}[2]{\textsf{attr}\ensuremath{^a_{#1}(#2)}}

\newcommand{\attrz}[2]{\textsf{attr}\ensuremath{^0_{#1}(#2)}}
 
\newcommand{\set}[1]{\left\lbrace #1\right\rbrace}
\newcommand{\Set}[1]{\set{#1}}
\newcommand{\tup}[1]{\left( #1\right)}

\newcommand{\inodd}{\in_{\mathrm{odd}}}
\newcommand{\ineven}{\in_{\mathrm{even}}}

\newcommand{\lang}{\mathcal{L}}

\newcommand{\p}[1]{\mathit{Player}~#1}
\newcommand{\pz}{\p{0}}
\newcommand{\po}{\p{1}}
\newcommand{\game}{\mathcal{G}}
\newcommand{\vertex}{V}
\newcommand{\vertexz}{\vertex^0}
\newcommand{\vertexo}{\vertex^1}
\newcommand{\vertexi}{\vertex^i}

\newcommand{\gamegraph}{G}
\newcommand{\spec}{\Phi}
\newcommand{\edge}{E}

\newcommand{\stratz}{\pi^0}
\newcommand{\strato}{\pi^1}
\newcommand{\strati}{\pi^i}

\newcommand{\play}{\rho}
\newcommand{\playprefix}{\mathtt{p}}

\makeatletter
\newsavebox{\@brx}
\newcommand{\llangle}[1][]{\savebox{\@brx}{\(\m@th{#1\langle}\)}%
  \mathopen{\copy\@brx\kern-0.5\wd\@brx\usebox{\@brx}}}
\newcommand{\rrangle}[1][]{\savebox{\@brx}{\(\m@th{#1\rangle}\)}%
  \mathclose{\copy\@brx\kern-0.5\wd\@brx\usebox{\@brx}}}

\makeatother
\newcommand{\team}[1]{\llangle #1\rrangle}
\newcommand{\src}{\mathit{src}}

\newcommand{\front}{\mathit{front}}

\renewcommand{\game}{\ensuremath{\mathcal{G}}}

\newcommand{\N}{\mathbb{N}}

\newcommand{\bigO}{\mathcal{O}}

\newcommand{\tqed}{\hfill{ $\triangleleft$ }}

\newcommand{\mucal}{$ \mu $-calculus\xspace}
\newcommand{\buchi}{\ifmmode B\ddot{u}chi \else B\"uchi \fi}
\newcommand{\cobuchi}{\ifmmode co\text{-}B\ddot{u}chi \else co-B\"uchi \fi}

\newcommand{\assump}{\ensuremath{\Psi}}
\newcommand{\assumpsafe}{\ensuremath{\assump_{\textsc{unsafe}}}}

\newcommand{\assumpgrlive}{\ensuremath{\assump_{\textsc{live}}}}
\newcommand{\assumpcondlive}{\ensuremath{\assump_{\textsc{cond}}}}
\newcommand{\assumpdep}{\ensuremath{\assump_{\textsc{colive}}}}

\newcommand{\livegroup}{H^\ell}
\newcommand{\livegroupSingle}{H_i}
\newcommand{\livegroupSingleN}{H}
\newcommand{\colivegroup}{D}
\newcommand{\safegroup}{S}
\newcommand{\condlivegroup}{\mathcal{H}^\ell}

\newcommand{\computeLive}{\textsc{LiveA}}
\newcommand{\computeCoLive}{\textsc{CoLiveA}}
\newcommand{\computeSafe}{\textsc{UnsafeA}}
\newcommand{\parityAssump}{\textsc{ParityAssumption}}

\newcommand{\paritygame}{\mathit{Parity}}
\newcommand{\buchigame}{\mathit{B\ddot{u}chi}}
\newcommand{\cobuchigame}{\mathit{co\text{-}B\ddot{u}chi}}
\newcommand{\safetygame}{\mathit{Safety}}

\newcommand{\solveParity}{\textsc{Parity}}
\newcommand{\solveBuchi}{\textsc{B\"uchi}}
\newcommand{\solveCobuchi}{\textsc{CoB\"uchi}}
\newcommand{\TsolveBuchi}{\textsc{TB\"uchi}}
\newcommand{\TsolveCobuchi}{\textsc{TCoB\"uchi}}
\newcommand{\solveSafety}{\textsc{Safety}}

\newcommand{\timeout}{Timeout}
\newcommand{\name}[1]{#1}
\newcommand{\tool}{\textsc{SImPA}\xspace}
\newcommand{\aname}{APA\xspace}
\newcommand{\anames}{APAs\xspace}
\newcommand{\krishtool}{\mathrm{GIST}}
\newcommand{\sepcomma}[1]{\tablenum[group-separator={,},table-format=9.0]{#1}}
\newcommand{\secformat}[1]{\tablenum[table-format=3.2]{#1}}

\begin{document}
\title{Computing Adequately Permissive Assumptions for Synthesis}
%
 \author{Ashwani Anand\inst{1} \and Kaushik~Mallik\inst{2} \and Satya Prakash Nayak\inst{1} \and \\ Anne-Kathrin Schmuck\inst{1}}

 \authorrunning{A. Anand, K. Mallik, S. P. Nayak, and A. Schmuck}
 %
 \institute{Max Planck Institute for Software Systems, Kaiserslautern, Germany\\
 \email{\{ashwani,sanayak,akschmuck\}@mpi-sws.org}\\
 \url{} \and
 Institute of Science and Technology Austria, Klosterneuburg, Austria\\
 \email{kaushik.mallik@ist.ac.at}\\
 \url{}}
%
\maketitle              
\begin{abstract}
We solve the problem of automatically computing a new class of environment assumptions in two-player turn-based finite graph games which characterize an ``adequate cooperation'' needed from the environment to allow the system player to win.
Given an $\omega$-regular winning condition $\Phi$ for the system player, we compute an $\omega$-regular assumption $\Psi$ for the environment player, such that (i) every environment strategy compliant with $\Psi$ allows the system to fulfill $\Phi$ (sufficiency), (ii) $\Psi$ can be fulfilled by the environment for every strategy of the system (implementability), and (iii) $\Psi$ does not prevent any cooperative strategy choice (permissiveness).

For parity games, which are canonical representations of $\omega$-regular games, we present a polynomial-time algorithm for the symbolic computation of \emph{adequately permissive assumptions} and show that our algorithm runs faster and produces better assumptions than existing approaches---both theoretically and empirically. To the best of our knowledge, for \emph{$\omega$-regular} games, we provide the first algorithm to compute sufficient and implementable environment assumptions that are also \emph{permissive}.

\end{abstract}
%
%
%

\section{Introduction}\label{section:intro}

Two-player $\omega$-regular games on finite graphs are the core algorithmic components in many important problems of computer science and cyber-physical system design. 
Examples include the synthesis of programs which react to environment inputs, modal $\mu$-calculus model checking, correct-by-design controller synthesis for cyber-physical systems, and supervisory control of autonomous systems. 

These problems can be ultimately reduced to an abstract two-player game between an \emph{environment player} and a \emph{system player}, 
respectively capturing the external unpredictable influences and the system under design, while the game captures the non-trivial interplay between these two parts.
A \emph{solution of the game} is a set of decisions the system player needs to make to satisfy a given $\omega$-regular temporal property over the states of the game,
 which is then used to design the sought system or its controller.

Traditionally, two-player games over graphs are solved in a zero-sum fashion, i.e., assuming that the environment will behave arbitrarily and possibly adversarially. 
Although this approach results in robust system designs, it usually makes the environment too powerful to allow an implementation for the system to exist.  However in reality, many of the outlined application areas actually account for some cooperation of system components, especially if they are co-designed. In this scenario it is useful to understand how the environment (i.e., other processes) needs to cooperate to allow for an implementation to exist. This can be formalized by environment assumptions, which are $\omega$-regular temporal properties that restrict the moves of the environment player in a synthesis game. 
Such assumptions can then be used as additional specifications in other components' synthesis problems to enforce the necessary cooperation (possibly in addition to other local requirements) or can be used to verify existing implementations. 

For the reasons outlined above, the automatic computation of assumptions has received significant attention in the reactive synthesis community.
It has been used in two-player games~\cite{ChatterjeeHenzingerJobstmann:EnvironmentAssumptionsforSynthesis,cavezza2020minimal}, both in the context of monolithic system design~\cite{chatterjee2010obliging,majumdar2019environmentally} as well as distributed system design~\cite{majumdar2020assume,finkbeiner2022information}.

All these works emphasize two desired properties of assumptions. They should be (i) \emph{sufficient}, i.e., enable the system player to win if the environment obeys its assumption and (ii) \emph{implementable}, i.e., prevent the system player to falsify the assumption and thereby vacuously win the game by not even respecting the original specification. 
In this paper, we claim that there is an important third property---termed \emph{permissiveness}---which is needed when computed assumptions are used for distributed synthesis. An assumption is \emph{permissive} if it retains all cooperatively winning plays in the game. This notion is crucial in the setting of distributed synthesis, as here assumptions are generated \emph{before} the implementation of every component is fixed. Therefore, assumptions need to retain \emph{all} feasible ways of cooperation to allow for a distributed implementation to be discovered in a decentralized manner.

While the class of assumptions considered in this paper is motivated by their use for distributed synthesis, this paper focuses only on their formalization and computation, i.e., given a two-player game over a finite graph and an $\omega$-regular winning condition $\Phi$ for the system player, we automatically compute an \emph{adequately permissive $\omega$-regular assumption} $\Psi$ for the environment player that formalizes the above intuition by being (i) sufficient, (ii) implementable, and (iii) permissive. 
The main observation that we exploit is that such \emph{adequately permissive assumptions} (\aname for short) can be constructed from three simple templates which can be directly extracted from a cooperative synthesis game leading to a polynomial-time algorithm for their computation.  
By observing page constrains, we postpone the very interesting but largely orthogonal problem of contract-based distributed synthesis using \anames to future work.

To appreciate the simplicity of the assumption templates we use, consider the 
game graphs depicted in Fig.~\ref{fig:illustration_intro} where the system and the environment player control the circle and square vertices, respectively. Given the specification $\Phi=\Diamond\Box \{p\}$ (which requires the play to eventually only see vertex $p$), the system player can win the game in Fig.~\ref{fig:illustration_intro} (a) by requiring the environment to fully disable edge $e_1$. This introduces the first template type---a \emph{safety template}---on $e_1$. On the other hand, the game in  Fig.~\ref{fig:illustration_intro} (b) only requires that $e_1$ is taken finitely often. This is captured by our second template type---a \emph{co-liveness template}---on $e_1$. Finally, consider the game in Fig.~\ref{fig:illustration_intro} (c) with the specification $\Phi=\Box\Diamond \{p\}$, i.e.\ vertex $p$ should be seen infinitely often. Here, the system player wins if whenever the source vertices of edges $e_1$ and $e_2$ are seen infinitely often, also one of these edges is taken infinitely often. This is captured by our third template type---a \emph{live group template}---on the edge-group $\{e_1, e_2\}$.

\begin{figure}[t]
	\centering
		\begin{subfigure}{.20\textwidth}
		  \centering
		  \begin{tikzpicture}
		  \node[] (duma) at (-0.3,1.4) {\textbf{(a)}};
		  		\node[player1] (0) at (0, 0) {$p$};
		  		\node[player0] (1) at (\fpeval{\pos}, 0) {$q$};
		  		
		  		\path[->] (0) edge[loop above] () edge node{$e_1$} (1);
		  		\path[->] (1) edge[loop above] ();
		  \end{tikzpicture}
		\end{subfigure}
				\hspace*{0.05\textwidth}
	\begin{subfigure}{.20\textwidth}
		  \centering
		  \begin{tikzpicture}
		  \node[] (dumc) at (-0.3,1.4) {\textbf{(b)}};
		  \node[] (dumd) at (2.3,1.4) {\textbf{(c)}};
		  		\node[player1] (0) at (0, 0) {$p$};
		  		\node[player0] (1) at (\fpeval{\pos}, 0) {$q$};
		  		
		  		\path[->] (0) edge[loop above] () edge[bend left = 20] node{$e_1$} (1);
		  		\path[->] (1) edge[bend left = 20] (0);
		  \end{tikzpicture}
		\end{subfigure}
		\hspace*{0.05\textwidth}
		\begin{subfigure}{.25\textwidth}
		  \centering
		  \begin{tikzpicture}
		  		\node[player0] (0) at (0, 0) {$p$};
		  		\node[player1] (1) at (-\fpeval{0.75*\pos}, -\pos) {$q$};
				\node[player1] (2) at (\fpeval{0.75*\pos}, -\pos) {$r$};
						  		
		  		\path[->] (0) edge[bend left =20] (1) edge[bend right=20] (2);
		  		\path[->] (1) edge[loop left] () edge[bend left =20] node{$e_1$} (0) edge[bend left=20] (2);
				\path[->] (2) 
									edge[loop right] () 
									edge[bend right =20] node[right]{$e_2$} (0) edge[bend left=20] (1);
		  \end{tikzpicture}
		\end{subfigure}
		\caption{Game graphs with environment (squares) and system (circles) vertices.}
		\label{fig:illustration_intro}
		\vspace{-0.5cm}
	\end{figure}
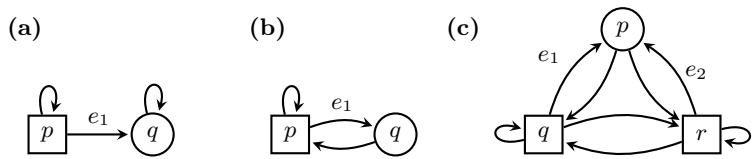
	
\smallskip
\noindent\textbf{Contribution.}
The main contribution of this paper is to show that \anames can always be composed from the three outlined assumption templates and can be computed in polynomial time.
 
Using a set of benchmark examples taken from SYNTCOMP~\cite{syntcomp} and a prototype implementation of our algorithm in our new tool $\tool$, we empirically show that our algorithm is both faster and produces more desirable solutions than existing approaches.
In addition, we apply $\tool$ to the well known 2-client arbiter synthesis benchmark from~\cite{NirGR1}, which is known to only allow for an implementation of the arbiter if the clients' moves are suitably restricted. We show that applying $\tool$ to the unconstrained arbiter synthesis problem yields assumptions on the clients which are less restrictive but conceptually similar to the ones typically used in the literature. 

\smallskip
\noindent\textbf{Related Work.}
The problem of automatically computing environment assumptions for synthesis was already addressed by Chatterjee et al.~\cite{ChatterjeeHenzingerJobstmann:EnvironmentAssumptionsforSynthesis}. However, their class of assumptions does in general not allow to construct \emph{permissive} assumptions. Further, computing their assumptions is an NP-hard problem,
while our algorithm computes \anames in $\mathcal{O}(n^4)$-time for a parity game with $n$ vertices.
The difference in the complexity arises because Chatterjee et al. require minimality of the assumptions. On the other hand, we trade minimality for permissiveness which allows us to utilize cooperative games, which are easier to solve.

When considering cooperative solutions of non-zerosum games, related works either fix strategies for both players~\cite{ChatterjeeHenzinger_2007,RationalSynthesis_2010}, assume a particularly rational behavior of the environment~\cite{BrenguierRaskinSankur_2017} or restrict themselves to safety assumptions~\cite{majumdar2020assume}. In contrast, we do not make any assumption on how the environment chooses its strategy. 
Finally, in the context of specification-repair in zerosum games multiple automated methods for repairing environment models exist, e.g.~\cite{schmelter2017toward,gaaloul2021combining,gaaloul2020mining,maoz2019symbolic,ChatterjeeHenzingerJobstmann:EnvironmentAssumptionsforSynthesis}.
Unfortunately, all of these methods fail to provide permissive repairs.
A recent work by Cavezza et al.~\cite{cavezza2020minimal} computes a minimally restrictive set of assumptions but only for GR(1) specifications, which are a strict subclass of the problem considered in our work.
To the best of our knowledge, we propose the first fully automated algorithm for computing \emph{permissive} assumptions for general $\omega$-regular games.

\section{Preliminaries}\label{section:prelims}

\noindent\textbf{Notation.}
We use $\mathbb{N}$ to denote the set of natural numbers including zero.
Given two natural numbers $a,b\in\mathbb{N}$ with $a<b$, we use $[a;b]$ to denote the set $\set{n\in\mathbb{N} \mid a\leq n\leq b}$.
For any given set $[a;b]$, we write $i\ineven [a;b]$ and $i\inodd [a;b]$ as short hand for $i\in [a;b]\cap \set{0,2,4,\ldots}$ and $i\in [a;b]\cap \set{1,3,5,\ldots}$ respectively.
Given two sets $A$ and $B$, a relation $R\subseteq A\times B$, and an element $a\in A$, we write $R(a)$ to denote the set $\set{b\in B\mid (a,b)\in R}$.

\smallskip
\noindent\textbf{Languages.}
Let $\Sigma$ be a finite alphabet.
The notations $\Sigma^*$ and $\Sigma^\omega$ denote the set of finite and infinite words over $\Sigma$, respectively, and $\Sigma^\infty$ is equal to $\Sigma^*\cup \Sigma^\omega$.
For any word $w\in \Sigma^\infty$, $w_i$ denotes the $i$-th symbol in $w$.
Given two words $u\in \Sigma^*$ and $v\in \Sigma^\infty$, the concatenation of $u$ and $v$ is written as the word $uv$.

\smallskip
\noindent\textbf{Game graphs.}
A \emph{game graph} is a tuple $\gamegraph= \tup{V,V^0,V^1,E}$ where $(V,E)$ is a finite directed graph with \emph{vertices} $ V $ and \emph{edges} $ E $, and 
$ \vertexz, \vertexo\subseteq V$ form a partition of $V$ (i.e.\ $\vertexz\cap\vertexo = \emptyset$ and $\vertexz\cup\vertexo = V$). Without loss of generality, we assume that for every $v\in V$ there exists $v'\in V$ s.t.\ $(v,v')\in E$. For the purpose of this paper, the \emph{system} and the \emph{environment} players will be denoted by $ \p{0} $ and $ \p{1} $, respectively.
A \emph{play} originating at a vertex $v_0$ is a finite or infinite sequence of vertices $\rho=v_0v_1\ldots \in V^\infty$. 
A  \emph{play prefix} $\playprefix = v_0v_1\cdots v_k$ is a finite play.

\smallskip
\noindent\textbf{Winning conditions.}
Given a game graph $\gamegraph$, we consider winning conditions specified using a formula $\Phi$ in \emph{linear temporal logic} (LTL) over the vertex set $V$, that is, we consider LTL formulas whose atomic propositions are sets of vertices $V$. 
In this case the set of desired infinite plays is given by the semantics of $\spec$ over $\gamegraph$, which is an $\omega$-regular language $\lang(\spec)\subseteq V^\omega$. 
Every game graph with an arbitrary $\omega$-regular set of desired infinite plays can be reduced to a game graph (possibly with an extended set of vertices) with an LTL winning condition, as above. 
The standard definitions of $\omega$-regular languages and LTL are omitted for brevity and can be found in standard textbooks~\cite{baier2008principles}.

\smallskip
\noindent\textbf{Games and strategies.}
A \emph{two-player (turn-based) game} is a pair $\game=\tup{\gamegraph,\spec}$ where $G $ is a game graph and 
$ \spec $ is a \emph{winning condition} over $\gamegraph$.
A strategy of $\p{i},~i\in\{0,1\}$, is a partial function $\strati\colon \vertex^*\vertexi\to \vertex$ such that for every $\playprefix v \in \vertex^*\vertexi$ for which $\pi$ is defined, it holds that $\strati(\playprefix v)\in \edge(v)$.
Given a strategy $\strati$, we say that the play $\play=v_0v_1\ldots$ is \emph{compliant} with $\strati$ if $v_{k-1}\in \vertexi$ implies $v_{k} = \strati(v_0\ldots v_{k-1})$ for all $k\in dom(\play)$ and,  $\play$ is a finite play $\playprefix$ only if $\strati(\playprefix)$ is undefined.
We refer to a play compliant with $\strati$ and a play compliant with both $\stratz$ and $\strato$ as a \emph{$ \strati $-play} and a \emph{$ \stratz\strato $-play}, respectively. 
We collect all plays compliant with $\strati$, and compliant with both $\stratz$ and $\strato$ in the sets $\lang(\strati)$ and $\lang(\stratz\strato)$, respectively. 

\smallskip
\noindent\textbf{Winning.}
Given a game $\game=(\gamegraph,\spec)$, a strategy $\strati$ is (surely) \emph{winning for $\p{i}$} if $\lang(\strati)\subseteq\lang(\spec)$, i.e., a $\pz$ strategy $\stratz$ is winning if \emph{for every} $\po$ strategy $\strato$ it holds that $\lang(\stratz\strato)\subseteq\lang(\spec)$. 
Similarly, a fixed strategy profile $(\stratz,\strato)$ is \emph{cooperatively winning} if $\lang(\stratz\strato)\subseteq\lang(\spec)$. We say that a vertex $v\in V$ is  \emph{winning for $\p{i}$} (resp. \emph{cooperatively winning}) if there exists a winning strategy $\strati$  (resp. a cooperatively winning strategy profile $(\stratz,\strato)$) s.t.\ $\strati(v)$ is defined. We collect all winning vertices of $\p{i}$ in the \emph{$\p{i}$ winning region} $\team{i}\spec\subseteq V$ and all cooperatively winning vertices in the \emph{cooperative winning region} $\team{0,1}\spec$. We note that $\team{i}\spec\subseteq\team{0,1}\spec$ for both $i\in\{0,1\}$. 

\section{Adequately Permissive Assumptions for Synthesis}\label{section:maximally permissive assumptions}
Given a two-player game $\game$, the goal of this paper is to compute assumptions on $\po$ (i.e., the environment), such that both players cooperate \emph{just enough} to fulfill $\spec$ while retaining all possible cooperative strategy choices.
Towards a formalization of this intuition, we define winning under assumptions. 

\begin{definition}\label{def:winundera}
 Let $\game=((V,\vertexz,\vertexo,E),\spec)$ be a game and $\assump$ be an LTL formula over $V$. Then a $\pz$ strategy $\stratz$ is winning in $\game$ under assumption $\assump$, if \emph{for every} $\po$ strategy $\strato$ s.t.\ $\lang(\strato)\subseteq\lang(\assump)$ it holds that $\lang(\stratz\strato)\subseteq\lang(\spec)$. We denote by $\team{0}_\assump\spec$ the set of vertices from which such a $\pz$ strategy exists.
\end{definition}

We see that the assumption $\assump$ introduced in  Def.~\ref{def:winundera} \emph{weakens} the strategy choices of the environment player ($\po$). 
We call assumptions \emph{sufficient} if this weakening is strong enough to allow  $\pz$ to win from every vertex in the cooperative winning region. 

\begin{definition}\label{def:sufficient}
 An assumption $\assump$ is \emph{sufficient} for $(\gamegraph,\spec)$ if $\team{0}_\assump\spec \supseteq \team{0,1}\spec$.
\end{definition}

Unfortunately, sufficient assumptions can be abused to change the given synthesis problem in an unintended way. 
Consider for instance the game in Fig.~\ref{fig:moreexamples} (left) with $\spec=\Box\lozenge\{v_0\}$ and $\assump = \square\lozenge e_1$. Here, there is no strategy $\strato$ for $\po$ such that $\lang(\strato)\subseteq\lang(\assump)$ as the system can always falsify the assumption by simply not choosing $e_1$ infinitely often in $v_1$. Therefore, any $\pz$ strategy is winning under assumption even if $\spec$ is violated.
The assumption $\assump$, however, is trivially sufficient, as
$\team{0}_\assump\spec=V$.
In order to prevent sufficient assumptions to be falsifiable and thereby enabling vacuous winning, we define the notion of \emph{implementability}, which ensures that $\assump$ solely restricts $\po$ moves.

\begin{definition}\label{def:implementable}
  An assumption $\assump$ is \emph{implementable} for $(\gamegraph,\spec)$ if $\team{1}\assump = V$.
\end{definition}

An assumption which is sufficient and implementable ensures that the cooperative winning region of the original game coincides with the winning region under that assumption, i.e., $\team{0}_\assump\spec =\team{0,1}\spec$. 
However, it does not yet ensure that all cooperative strategy choices of both players are retained, which is ensured by the notion of \emph{permissiveness}. 

\begin{definition}\label{def:permissive}
  An assumption $\assump$ is \emph{permissive} for $(\gamegraph,\spec)$ if $\lang(\spec)\subseteq \lang(\assump)$.
\end{definition}

This notion of permissiveness is motivated by the intended use of assumptions for compositional synthesis. In the simplest scenario of two interacting processes, two synthesis tasks---one for each process---are considered in parallel. Here, generated assumptions in one synthesis task are used as additional specifications in the other synthesis problem. Therefore,  permissiveness is crucial to not \enquote{skip} over possible cooperative solutions---each synthesis task needs to keep all allowed strategy choices for both players intact to allow for compositional reasoning. This scenario is illustrated in the following example to motivate the considered class of assumptions. Formalizing assumption-based compositional synthesis in general is however out of the scope of this paper. 

\begin{example}\label{example:permissive}
Consider the (non-zerosum) two-player game in Fig.~\ref{fig:moreexamples} (middle) with two different specifications for both players, namely $\Phi_0=\lozenge\Box\{v_1,v_2\}$ and $\Phi_1=\lozenge\Box\{v_1\}$. Now consider two candidate assumptions $\assump_0 = \lozenge\Box \neg e_1$ and $\assump_0' = (\Box\lozenge v_1 \implies\Box\lozenge e_2)$ on $\po$.
Notice that both assumptions are sufficient and implementable for $(\gamegraph, \spec_0)$. 
However, $\assump_0'$ does not allow the play $\{v_1\}^\omega$  and hence is not permissive whereas $\assump_0$ is permissive for $(\gamegraph, \spec_0)$. As a consequence, there is no way $\po$ can satisfy both her objective $\spec_1$ and the assumption $\assump_0'$ even if $\pz$ cooperates, since $\lang(\spec_1) \cap \lang(\assump_0') = \emptyset$. 
 However, under the assumption $\assump_0$ on $\po$ and assumption $\assump_1 =  \lozenge\Box \neg e_3$ on $\pz$ (which is sufficient and implementable for $(\gamegraph, \spec_1)$ if we interchange the vertices of the players), they can satisfy both their own objectives and the assumptions on themselves. Therefore, they can collectively satisfy both their objectives.
\end{example}

\begin{figure}[t]
\vspace{-0.5cm}
		  \centering
		  \begin{tikzpicture}
		  		\node[player0] (0) at (0, 0) {$v_0$};
		  		\node[player1] (1) at (-\fpeval{\pos}, 0) {$v_1$};
				\node[player0] (2) at (\fpeval{\pos}, 0) {$v_2$};
						  		
		  		\path[->] (0) edge (2);
		  		\path[->] (1) edge[loop left] () edge node{$e_1$} (0);
				\path[->] (2) edge[loop right] ()  edge[bend left = 40] (1);
		  \end{tikzpicture}
\hspace{0.2cm}
 \begin{tikzpicture}
	  		\node[player1] (0) at (0, 0) {$v_0$};
	  		\node[player1] (1) at (\fpeval{\pos}, 0) {$v_1$};
			\node[player0] (2) at (2*\fpeval{\pos}, 0) {$v_2$};
					  		
	  		\path[->] (0) edge[bend left = 20] (1);
	  		\path[->] (1) edge[loop above] () edge[bend left = 20] node{$e_1$} (0) edge[bend left = 20] node{$e_2$} (2);
			\path[->] (2) edge[loop above] node{$e_3$} ()  edge[bend left = 20] (1);
 \end{tikzpicture}
 \hspace{0.2cm}
	\begin{tikzpicture}
		\node[player0, double] (0) at (0, 0) {$v_0$};
		\node[player0] (1) at (-\fpeval{\pos}, 0) {$v_1$};
		\node[player1] (2) at (\fpeval{\pos}, 0) {$v_2$};
		
		\path[->] (0) edge[bend right=20] (2) ;
		\path[->] (1) edge (0);
		\path[->] (2)  edge[bend left = 40] (1) edge[bend right=20] node[above]{$e_1$} (0);
	\end{tikzpicture}
\caption{Two-player games with $\po$ (squares) and $\pz$ (circles) vertices.}\label{fig:moreexamples}
\vspace{-0.5cm}
\end{figure}
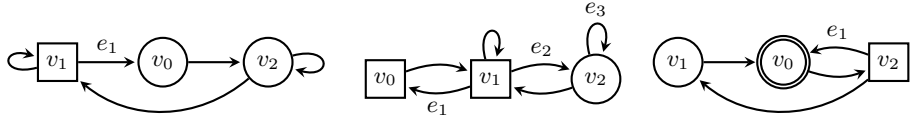

\begin{remark}
We remark that for Ex.~\ref{example:permissive}, the algorithm in~\cite{chatterjee2010gist} outputs $\assump_0'$ as the desired assumption for game $(\gamegraph, \spec_0)$ and their used assumption formalism is not rich enough to capture assumption $\assump_0$. This shows that the assumption type we are interested in is not computable by the algorithm from~\cite{chatterjee2010gist}.
\end{remark}

\begin{definition}
 An assumption $\assump$ is called \emph{adequately permissive} (an APA for short) for $(\gamegraph,\spec)$ if it is sufficient, implementable and permissive.
\end{definition}

\subsection{Discussion on Definition~\ref{def:winundera}}

We first note some simple but interesting consequences of Def.~\ref{def:winundera}. First, we have \emph{anti-monotonicity}, i.e, if assumption $ \assump_1 $ is stronger than assumption $ \assump_2 $ (in terms of play inclusion), and $ \stratz $ is winning under $ \assump_2 $, then it is also winning under $ \assump_1 $. As a direct consequence of this observation, we also have \emph{conjunctivity}, i.e., if $ \stratz $ is winning under $ \assump_1 $ and $ \stratz $ is winning under $ \assump_2 $, then $ \stratz $ is winning under $ \assump_1\wedge\assump_2 $. Interestingly, however, Def.~\ref{def:winundera} does not allow for \emph{disjunctivity}, i.e., if $ \stratz $ is winning under $ \assump_1 $ and $ \stratz $ is winning under $ \assump_2 $, then it need \emph{not} be winning under $ \assump_1\vee\assump_2 $. This last observation is illustrated by the following example. 
	\begin{example}\label{example:disjunction}
	Consider the game graph in Fig.~\ref{fig:Definition1disjunction} with the specification $ \spec = \lozenge\square \{a\} $ (which requires the play to eventually only see vertex $a$). Then consider the assumptions $ \assump_1 = \neg e_0 \mathcal{U} X e_1 $ (when edge $e_0$ is taken for the first time, the next edge should be $e_1$) and $ \assump_2=e_0 \mathcal{U} X e_2 $ (when edge $e_0$ is taken for the first time, the next edge should be $e_2$). Notice that there is only one $ \po $ strategy $\strato$, i.e., the one that never uses edge $e_0$, satisfying either assumption. So, any play compliant with $\strato$ eventually only visits vertex~$a$, and hence, is winning. Therefore, any $ \pz $ strategy is winning under either assumption. In particular, consider the strategy $ \stratz $ that only uses edge $e_1$. Then $\stratz$ is winning under assumption $\assump_i$ for each $i$. However, $\stratz$ is not winning under $\assump:=\assump_1\vee\assump_2 \equiv \mathtt{true}$. To see this, note that assumption $\assump$ can be satisfied by any $ \po $ strategy, in particular, the strategy $\widetilde{\pi}^1$ that always uses $e_0$ from state $a$. It is easy to see that the combination of $\widetilde{\pi}^1$ with $\stratz$ yields the play $(abc)^\omega$ that satisfies $\assump$ but not $\spec$. Hence, $\stratz$ is not winning under assumption $\assump:=\assump_1\vee\assump_2$.

	\end{example}
			\begin{figure}
			\centering
			\begin{subfigure}{.25\textwidth}
				\begin{tikzpicture}
					\node[player1] (0) at (-1.5, 0) {$a$};
					\node[player0] (1) at (0,0) {$b$};
					\node[player1] (2) at (1.5,0) {$c$};
					
					\path[->] (0) edge[loop left] () edge[bend left = 20] node{$ e_0 $} (1);
					\path[->] (1) edge[bend left=20] node{$e_2$} (0) edge node{$ e_1 $} (2);
					\path[->] (2) edge[bend left=50] (0);
				\end{tikzpicture}
			\end{subfigure}
			\caption{Example game graph illustrating non-disjunctivity of winning under assumption, as explained in Ex.~\ref{example:disjunction}.}
			\label{fig:Definition1disjunction}
			\vspace{-0.5cm}
		\end{figure}
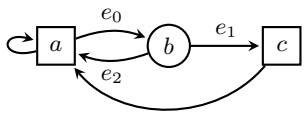

In addition, we want to remark that Def.~\ref{def:winundera} slightly differs from the typical linear-time synthesis setting, where winning under assumption would be naturally defined in terms of plays instead of strategies. We therefore want to briefly discuss this setting and give some intuition why it coincides with our definition of winning for the special type of assumptions we compute.
	
	We start by giving an alternative formulation of Def.~\ref{def:winundera} in terms of plays.
	\begin{definition}\label{def:winundera:alt}
		Let $\game=((V,\vertexz,\vertexo,E),\spec)$ be a game and $\assump$ be an LTL formula over $V$. Then a $\pz$ strategy $\stratz$ is winning in $\game$ under assumption $\assump$, if every play $ \play\in \lang(\stratz) $ either fails to satisfy the assumption $ \assump $ or satisfies the specification $ \spec $.
	\end{definition}
	It is easy to observe that a strategy $ \stratz $ that is winning under assumption by Def.~\ref{def:winundera:alt} is also winning under assumption by Def.~\ref{def:winundera}. However, the other direction is not true in general, as shown by the following example.
	\begin{example}\label{example:def1_def6}
		Consider the same game as in  Example~\ref{example:disjunction}, i.e., the game in Fig.~\ref{fig:Definition1disjunction} with specification $ \spec = \lozenge\square \{a\} $. Also, consider the assumption $ \assump_1 = \neg e_0 \mathcal{U} X e_1 $ as in Example~\ref{example:disjunction}. Then by the same arguments as before, the strategy $\stratz$, which only uses edge $e_1$, is winning under $\assump_1$ by Def.~\ref{def:winundera}.  However, note that the play $(abc)^\omega$ is compliant with $\stratz$ and satisfies $\assump_1$ but does not satisfy $\spec$. Hence, $\stratz$ is not winning under $\assump_1$ by Def.~\ref{def:winundera:alt}.
	\end{example}
	
	Interestingly, the class of assumptions we compute in this paper does not allow for examples of the sort presented above. Intuitively, this is due to the fact that these assumptions are \emph{implementable} by $\po$ and realized by a combination of very local templates. These assumptions can therefore be enforced by only restricting the moves of $\po$. This implies that for any play $\rho$ that complies with the assumption and $\stratz$, there does exist a strategy $\strato$ satisfying the assumption, which results in $\rho$. Hence, any strategy $ \stratz $ which is winning under assumption for by Def.~\ref{def:winundera} is also winning under assumption by Def.~\ref{def:winundera:alt}. The complete proof of the equivalence between the two definitions for our class of assumption can be found in Appendix~\ref{appendix:equivOfDefWinUnder} as it requires the results of the next sections.

	We conclude this subsection by noting that the choice of our formulation of `winning under assumption' is inspired by distributed synthesis. Here, the environment agents might be unknown to the system. Our definition allows us to naturally argue that the strategy $ \stratz $ of $ \pz $ (System) is winning for \emph{any} strategy $ \strato $ that $ \po $ (Environment) may choose to satisfy the assumption. 

\section{Computing Adequately Permissive Assumptions (APA)}\label{section:assumption computation}

In this section, we present our algorithm to compute \emph{adequately permissive assumptions} (APA for short) for \emph{parity games}, which are canonical representations of $\omega$-regular games. 
For a gradual exposition of the topic, we first present algorithms for simpler winning conditions, namely safety (Sec.~\ref{sec:assump:safety}), Büchi (Sec.~\ref{section:BuchiGames}), and Co-Büchi (Sec.~\ref{section:coBuchiGames}), which are used as building blocks while presenting the algorithm for parity games (Sec.~\ref{section:parityGames}). 
We first introduce some preliminaries.

\subsection{Preliminaries}\label{sec:assump:prelim}
We use symbolic fixpoint algorithms expressed in the $\mu$-calculus~\cite{Kozen:muCalculus} to compute the winning regions and to generate assumptions in simple post-processing steps. 
\smallskip
\noindent\textbf{Set Transformers.}  Let $ \gamegraph=(V,\vertexz, \vertexo, E) $ be a game graph, $ U\subseteq V $ be a subset of vertices, and $ a\in \{0,1\} $ be the player index. Then we define two types of predecessor operators:
\begin{eqnarray}
	\pre{\gamegraph}{U} =& \{v\in V\mid \exists u\in U . ~(v,u)\in E \}\\
	\cprea{\gamegraph}{U} =& \{v\in V^a\mid v\in \pre{\gamegraph}{U}\}\cup \{v\in V^{1-a}\mid \forall (v,u)\in E.~u\in U  \}\\
	\cpre{\gamegraph}{U}{a,1} =& \cpre{\gamegraph}{U}{a}\cup U\\
	\cpre{\gamegraph}{U}{a,i}=&\cpre{\gamegraph}{\cpre{\gamegraph}{U}{a,i-1}}{a} \cup \cpre{\gamegraph}{U}{a,i-1} \text{ with } i\geq 1
\end{eqnarray}
The predecessor operator $ \textsf{pre}_{\gamegraph}(U) $ computes the set of vertices with at least one successor in $ U $. 
The controllable predecessor operators $ \textsf{cpre}^a_{\gamegraph}(U) $ and $\cpre{\gamegraph}{U}{a,i}$ compute the set of vertices from which $ \p{a} $ can force visiting $ U $ in \emph{at most} \emph{one} and $i$ steps respectively. 
In the following, we introduce the attractor operator $ \textsf{attr}^a_{\gamegraph}(U) $ that computes the set of vertices from which $ \p{a}$ can force at least a single visit to $ U $ in \emph{finitely many but nonzero}\footnote{In existing literature, usually $ U\subseteq\mathsf{attr}^a(U) $, i.e., \ $\attra{}{U}$ contains vertices from which $U$ is visited in zero steps. We exclude $U$ from $\attra{}{U}$ for a minor technical reason.} steps: 
\begin{eqnarray}
	\attra{\gamegraph}{U} =&\big(  \bigcup_{i\geq 1} \textsf{cpre}^{a,i}{(U)} \big)\backslash U
\end{eqnarray}
When clear from the context, we drop the subscript $ \gamegraph $ from these operators.

\smallskip
\noindent\textbf{Fixpoint Algorithms in the $ \mu $-calculus.} 
\mucal~\cite{Kozen:muCalculus} offers a succinct representation of symbolic algorithms (i.e., algorithms manipulating sets of vertices instead of individual vertices) over a game graph $ \gamegraph $. 
The formulas of the $ \mu $-calculus, interpreted over a 2-player game graph $ \gamegraph $, are given by the grammar 
\[ \phi\coloneqq p \mid X \mid \phi\cup\phi \mid \phi\cap\phi \mid \mathit{pre}(\phi) \mid \mu X.\phi \mid \nu X.\phi \]
where $ p $ ranges over subsets of $ V $, $ X $ ranges over a set of formal variables, $ pre $ ranges over monotone set transformers in $ \{\textsf{pre}, \textsf{cpre}^a, \textsf{attr}^a  \} $, and $ \mu $ and $ \nu $ denote, respectively, the least and the greatest fixed point of the functional defined as $ X\mapsto \phi(X) $. 
Since the operations $ \cup, \cap $, and the set transformers $ \mathit{pre} $ are all monotonic, the fixed points are guaranteed to exist, due to the Knaster-Tarski Theorem~\cite{KnasterTarski:TraskiKnasterTheorem}.
We omit the (standard) semantics of formulas (see~\cite{Kozen:muCalculus}).

A \mucal formula evaluates to a set of vertices over $ \gamegraph $, and the set can be computed by induction over the structure of the formula, where the fixed points are evaluated by iteration. The reader may note that $ \textsf{pre} $ and $ \textsf{cpre} $ can be computed in time polynomial in number of vertices, and since the game graph is finite, $ \textsf{attr} $ is also computable in polynomial time.

\subsection{Safety Games}\label{sec:assump:safety}

A safety game is a game $\game=(\gamegraph,\spec)$ with $\spec\coloneqq \square U$
for some $U\subseteq V$, and a play fulfills $\spec$ if it never leaves $U$.
\anames for safety games disallow every $\po$ move that leaves the cooperative winning region in $\gamegraph$ w.r.t.\ $\safetygame(U)$. This is formalized in the following theorem\footnote{All proofs can be found in the appendix.}. 

\begin{restatable}{theorem}{restatesafety}\label{thm:safety assumption}
 Let $\game=(\gamegraph,\square U)$ be a safety game, $Z^*=\nu Y. U\cap \pre{}{Y}$, and $ \safegroup = \Set{(u,v)\in E\mid \left(u\in V^1\cap Z^*\right) \wedge \left(v \notin Z^*\right)}$. Then $Z^*=\team{0,1}\square U$ and
 \footnote{We use $ e=(u,v) $ in LTL formulas as a syntactic sugar for $ u\wedge \bigcirc v $, where $ \bigcirc$ is the LTL \emph{next} operator.
 A set of edges $E' = \set{e_i}_{i\in [0;k]}$, when used as atomic proposition, is a syntactic sugar for $\bigvee_{i\in [0;k]} e_i$. } 
\begin{equation}\label{eq:safety assumption definition}
	\textstyle \assumpsafe(\safegroup) \coloneqq \square \bigwedge_{e\in \safegroup} \neg e,
\end{equation}
is an \aname for the game $\game$. We denote by $\computeSafe(\gamegraph,U)$ the algorithm computing $\safegroup$ as above, which runs in time $ \bigO(n^2) $, where $ n=|V|$. 
\end{restatable}
We call the LTL formula in \eqref{eq:safety assumption definition} a \emph{safety template} and assumptions that solely use this template \emph{safety assumptions}.

\subsection{Live Group Assumptions for B\"uchi Games}\label{section:BuchiGames}
\noindent\textbf{Büchi games}.
A B\"uchi game is a game $\game=(\gamegraph,\spec)$ where $\spec=\square\lozenge U$
for some $U\subseteq V$.
Intuitively, a play is winning for a Büchi game if it visits the vertex set $U$ infinitely often. 
We first recall that the cooperative winning region $\team{0,1}\square\lozenge U$ can be computed by a two-nested symbolic fixpoint algorithm~\cite{ChatterjeeHenzingerPiterman:AlgoForBuchi}
	\begin{equation}\label{eq:EnvironmentBuchiFixpoint}
		\solveBuchi(\gamegraph,U):=\nu Y. \mu X.~(U\cap \pre{}{Y})\cup (\pre{}{X}).
	\end{equation}
\smallskip
\noindent\textbf{Live group templates.}
Given the standard algorithm in \eqref{eq:EnvironmentBuchiFixpoint}, the set $X^i$ computed in the  $ i $-th iteration of the fixpoint variable $ X $ in the last iteration of $ Y $ actually carries a lot of information to construct a very useful assumption for the Büchi game $\game$. To see this, recall that $ X^i $ contains all vertices which have an edge to vertices which can reach $ U $ in at most $ i-1 $ steps~\cite[sec. 3.2]{ChatterjeeHenzingerPiterman:AlgoForBuchi}. Hence, for all $\po$ vertices in $X^i\setminus X^{i-1}$ we need to assume that $ \po $ always eventually makes progress towards $U$ by moving to $X^i$. This can be formalized by a so called live group template.

\begin{definition}
 	Let $G=(V,E)$ be a game graph. Then a live group $\livegroupSingleN = \Set{e_j}_{j\geq 0}$ is a set of edges $e_j = (s_j,t_j)$ with source vertices $\src(\livegroupSingleN):=\Set{s_j}_{j\geq 0}$. Given a set of live groups 
 	$\livegroup=\left\{\livegroupSingle\right\}_{i\geq 0}$  we define a live group template as 
\begin{equation}\label{equ:livegroup}
	\assumpgrlive(\livegroup) \coloneqq \bigwedge_{i\geq 0}\square\lozenge src(\livegroupSingle)\implies\square\lozenge \livegroupSingle.
\end{equation}

\end{definition}
The live group template says that if some vertex from the source of a live group is visited infinitely often, then some edge from this group should be taken infinitely often. We will use this template to give the assumptions for \buchi games.
\begin{remark}\label{rem:liveedges}
 We note that Chatterjee et al. ~\cite{ChatterjeeHenzingerJobstmann:EnvironmentAssumptionsforSynthesis} used \emph{live edges} in their environment assumptions. Live edges are singleton live groups and are thereby less expressive. In particular, there are instances of \buchi games, where there is no permissive live edge assumption but there is a permissive live group assumption\footnote{ i.e., assumptions that use live group templates}.
E.g., \ in Fig.~\ref{fig:illustration_intro} (c) the live edge assumption $\square\lozenge e_1 \wedge \square\lozenge e_2$ is sufficient but not permissive, whereas the live group assumption $\square \lozenge src(\livegroupSingleN)\implies \square\lozenge \livegroupSingleN$ with $\livegroupSingleN = \{e_1,e_2\}$ is an \aname.
\end{remark}

In the context of the fixpoint computation of \eqref{eq:EnvironmentBuchiFixpoint}, we can construct live groups $\livegroup=\left\{\livegroupSingle\right\}_{i\geq 0}$ where each $\livegroupSingle$ contains all edges of $\po$ which originate in $X^i\setminus X^{i-1}$ and end in $X^{i-1}$. Then the live group assumption in \eqref{equ:livegroup} precisely captures the intuition that, in order to visit $U$ infinitely often, $\po$ should take edges in $\livegroupSingle$ infinitely often if vertices in $\src(\livegroupSingle)$ are seen infinitely often.
Unfortunately, it turns out that this live group assumption is not \emph{permissive}. The reason is that it restricts $ \po $ also on those vertices from which he will anyway go towards $ U $. 
For example, consider the game in Fig.~\ref{fig:moreexamples} (right). Here defining live groups through computations of \eqref{eq:FP:buechi}, will mark $ {e_1} $ as a live group, but then $ (v_2v_1v_0)^{\omega} $ will be in $ \lang(\Phi) $ but not in the language of the assumption. Here the permissive assumption would be $ \assump=\textsc{true} $.

\smallskip
\noindent\textbf{Accelerated fixpoint computation.}
In order to compute a permissive live group assumption, we use a slightly modified fixpoint algorithm which computes the same set $Z^*$ but allows us to extract \emph{permissive} assumptions directly from the fixpoint computations. 
Towards this goal, we introduce the \emph{together predecessor operator}.
\begin{equation}
	\tpre{\gamegraph}{U}= \attrz{\gamegraph}{U} \cup \cpreo{\gamegraph}{\attrz{\gamegraph}{U}\cup U}. 
\end{equation}
Intuitively, $\mathsf{tpre}$ adds all vertices from which $ \pz $ does not need any cooperation to reach $U$ in every iteration of the fixpoint computation.
The interesting observation we make is that substituting the inner pre operator in \eqref{eq:EnvironmentBuchiFixpoint} by $\mathsf{tpre}$ does not change the computed set but only accelerates the computation. This is formalized in the next proposition and visualized in Fig. \ref{fig:PreVsTpre}.

\begin{restatable}{proposition}{restate:prop:buechi cooperative winning region}\label{prop:buechi cooperative winning region}
	Let $ \game=\tup{\gamegraph,\square\lozenge U} $ be a $\buchi$ game and 
	\begin{equation}\label{eq:FP:buechi}
	\TsolveBuchi(\gamegraph,U)=\nu Y. \mu X.~(U\cap \pre{}{Y})\cup (\tpre{}{X}).
\end{equation}
Then $ \TsolveBuchi(\gamegraph,U)=\solveBuchi(\gamegraph,U)=\team{0,1}\square \lozenge U$.
\end{restatable}
Prop.~\ref{prop:buechi cooperative winning region} follows from the correctness proof of \eqref{eq:EnvironmentBuchiFixpoint} by using the observation that for all $ U\subseteq V $ we have $ \mu X. ~U\cup \pre{}{X}=\mu X. ~U\cup \tpre{}{X}$ which is proven in the Appendix, Lem.~\ref{lemma:tpre is accelerated pre}.

\tikzset{every picture/.style={line width=0.75pt}} 
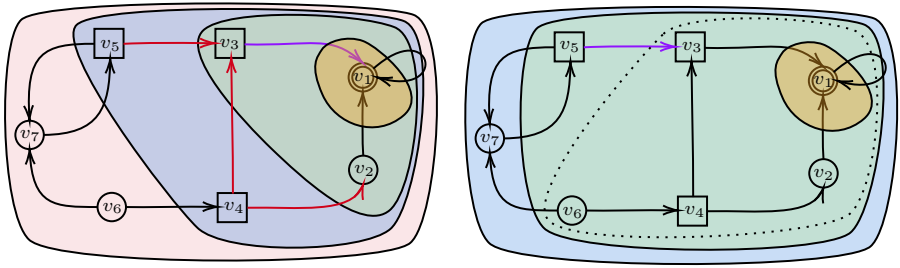
\begin{figure}[t]
	\centering
	\tikzset{every picture/.style={line width=0.75pt}} 
	
	\begin{tikzpicture}[x=0.75pt,y=0.75pt,yscale=-0.7,xscale=0.7]
		
		\draw  [fill={rgb, 255:red, 74; green, 144; blue, 226 }  ,fill opacity=0.3 ] (351.74,24.14) .. controls (368.94,7.74) and (607.34,10.14) .. (634.54,20.54) .. controls (661.74,30.94) and (652.15,167.76) .. (632.14,184.94) .. controls (612.13,202.12) and (378.14,202.54) .. (357.34,182.54) .. controls (336.54,162.54) and (334.54,40.54) .. (351.74,24.14) -- cycle ;
		
		\draw  [fill={rgb, 255:red, 184; green, 233; blue, 134 }  ,fill opacity=0.3 ] (391.74,27.74) .. controls (408.94,11.34) and (600.14,14.54) .. (627.34,24.94) .. controls (654.54,35.34) and (637.75,157.76) .. (617.74,174.94) .. controls (597.73,192.12) and (417.74,191.87) .. (396.94,171.87) .. controls (376.14,151.87) and (374.54,44.14) .. (391.74,27.74) -- cycle ;
		
		\draw   (404.53,30.6) -- (425.53,30.6) -- (425.53,51.6) -- (404.53,51.6) -- cycle ;
		
		\draw   (587.98,132.3) .. controls (587.98,126.64) and (592.57,122.05) .. (598.23,122.05) .. controls (603.89,122.05) and (608.48,126.64) .. (608.48,132.3) .. controls (608.48,137.96) and (603.89,142.55) .. (598.23,142.55) .. controls (592.57,142.55) and (587.98,137.96) .. (587.98,132.3) -- cycle ;
		
		\draw   (587.68,65.55) .. controls (587.68,59.89) and (592.27,55.3) .. (597.93,55.3) .. controls (603.59,55.3) and (608.18,59.89) .. (608.18,65.55) .. controls (608.18,71.21) and (603.59,75.8) .. (597.93,75.8) .. controls (592.27,75.8) and (587.68,71.21) .. (587.68,65.55) -- cycle ;
		
		\draw   (347.58,107.1) .. controls (347.58,101.44) and (352.17,96.85) .. (357.83,96.85) .. controls (363.49,96.85) and (368.08,101.44) .. (368.08,107.1) .. controls (368.08,112.76) and (363.49,117.35) .. (357.83,117.35) .. controls (352.17,117.35) and (347.58,112.76) .. (347.58,107.1) -- cycle ;
		\draw   (406.78,159.1) .. controls (406.78,153.44) and (411.37,148.85) .. (417.03,148.85) .. controls (422.69,148.85) and (427.28,153.44) .. (427.28,159.1) .. controls (427.28,164.76) and (422.69,169.35) .. (417.03,169.35) .. controls (411.37,169.35) and (406.78,164.76) .. (406.78,159.1) -- cycle ;
		\draw   (491.73,30.6) -- (512.73,30.6) -- (512.73,51.6) -- (491.73,51.6) -- cycle ;
		\draw   (493.33,149) -- (514.33,149) -- (514.33,170) -- (493.33,170) -- cycle ;
		\draw    (404.33,41.4) .. controls (367.7,41.2) and (354.6,47.67) .. (357.73,95.39) ;
		\draw [shift={(357.83,96.85)}, rotate = 265.92] [color={rgb, 255:red, 0; green, 0; blue, 0 }  ][line width=0.75]    (10.93,-3.29) .. controls (6.95,-1.4) and (3.31,-0.3) .. (0,0) .. controls (3.31,0.3) and (6.95,1.4) .. (10.93,3.29)   ;
		\draw [color={rgb, 255:red, 144; green, 19; blue, 254 }  ,draw opacity=1 ]   (425.53,41.4) .. controls (442,40.62) and (459.23,40.6) .. (489.65,40.98) ;
		\draw [shift={(491.53,41)}, rotate = 180.73] [color={rgb, 255:red, 144; green, 19; blue, 254 }  ,draw opacity=1 ][line width=0.75]    (10.93,-3.29) .. controls (6.95,-1.4) and (3.31,-0.3) .. (0,0) .. controls (3.31,0.3) and (6.95,1.4) .. (10.93,3.29)   ;
		\draw    (598.23,122.05) .. controls (597.55,109.72) and (597.53,101.8) .. (597.9,77.69) ;
		\draw [shift={(597.93,75.8)}, rotate = 90.9] [color={rgb, 255:red, 0; green, 0; blue, 0 }  ][line width=0.75]    (10.93,-3.29) .. controls (6.95,-1.4) and (3.31,-0.3) .. (0,0) .. controls (3.31,0.3) and (6.95,1.4) .. (10.93,3.29)   ;
		\draw    (512.73,41.8) .. controls (562.13,42.98) and (574.64,34.45) .. (596.58,54.06) ;
		\draw [shift={(597.93,55.3)}, rotate = 222.92] [color={rgb, 255:red, 0; green, 0; blue, 0 }  ][line width=0.75]    (10.93,-3.29) .. controls (6.95,-1.4) and (3.31,-0.3) .. (0,0) .. controls (3.31,0.3) and (6.95,1.4) .. (10.93,3.29)   ;
		\draw    (514.93,159.6) .. controls (551.4,159.02) and (589.19,164.32) .. (597.64,144.15) ;
		\draw [shift={(598.23,142.55)}, rotate = 107.85] [color={rgb, 255:red, 0; green, 0; blue, 0 }  ][line width=0.75]    (10.93,-3.29) .. controls (6.95,-1.4) and (3.31,-0.3) .. (0,0) .. controls (3.31,0.3) and (6.95,1.4) .. (10.93,3.29)   ;
		\draw    (504.33,149.4) .. controls (504.33,112.76) and (503.56,88.53) .. (503.15,53.03) ;
		\draw [shift={(503.13,51.4)}, rotate = 89.37] [color={rgb, 255:red, 0; green, 0; blue, 0 }  ][line width=0.75]    (10.93,-3.29) .. controls (6.95,-1.4) and (3.31,-0.3) .. (0,0) .. controls (3.31,0.3) and (6.95,1.4) .. (10.93,3.29)   ;
		\draw    (406.78,159.1) .. controls (375.61,159) and (360,159) .. (357.92,119.2) ;
		\draw [shift={(357.83,117.35)}, rotate = 87.66] [color={rgb, 255:red, 0; green, 0; blue, 0 }  ][line width=0.75]    (10.93,-3.29) .. controls (6.95,-1.4) and (3.31,-0.3) .. (0,0) .. controls (3.31,0.3) and (6.95,1.4) .. (10.93,3.29)   ;
		\draw    (427.28,159.1) .. controls (441.57,159.39) and (468.92,158.64) .. (491.78,158.97) ;
		\draw [shift={(493.53,159)}, rotate = 180.99] [color={rgb, 255:red, 0; green, 0; blue, 0 }  ][line width=0.75]    (10.93,-3.29) .. controls (6.95,-1.4) and (3.31,-0.3) .. (0,0) .. controls (3.31,0.3) and (6.95,1.4) .. (10.93,3.29)   ;
		\draw    (368.08,107.1) .. controls (408.52,106.61) and (415.33,87.58) .. (415.91,53.37) ;
		\draw [shift={(415.93,51.8)}, rotate = 90.65] [color={rgb, 255:red, 0; green, 0; blue, 0 }  ][line width=0.75]    (10.93,-3.29) .. controls (6.95,-1.4) and (3.31,-0.3) .. (0,0) .. controls (3.31,0.3) and (6.95,1.4) .. (10.93,3.29)   ;
		\draw   (590.33,65.55) .. controls (590.33,61.35) and (593.74,57.95) .. (597.93,57.95) .. controls (602.13,57.95) and (605.53,61.35) .. (605.53,65.55) .. controls (605.53,69.75) and (602.13,73.15) .. (597.93,73.15) .. controls (593.74,73.15) and (590.33,69.75) .. (590.33,65.55) -- cycle ;
		\draw  [fill={rgb, 255:red, 245; green, 166; blue, 35 }  ,fill opacity=0.4 ] (568.54,45.42) .. controls (572.14,39.82) and (586.54,35.02) .. (598.14,39.42) .. controls (609.74,43.82) and (647.34,75.82) .. (627.33,93) .. controls (607.33,110.18) and (578.54,99.82) .. (568.94,78.22) .. controls (559.34,56.62) and (564.94,51.02) .. (568.54,45.42) -- cycle ;
		\draw  [dash pattern={on 0.84pt off 2.51pt}] (482.94,31.74) .. controls (500.14,15.34) and (593.34,18.94) .. (620.54,29.34) .. controls (647.74,39.74) and (637.35,144.16) .. (617.34,161.34) .. controls (597.33,178.52) and (421.34,186.67) .. (400.54,166.67) .. controls (379.74,146.67) and (465.74,48.14) .. (482.94,31.74) -- cycle ;
		\draw    (605.54,59.31) .. controls (648.44,21.03) and (660.63,80.76) .. (609.74,66.02) ;
		\draw [shift={(608.18,65.55)}, rotate = 17.32] [color={rgb, 255:red, 0; green, 0; blue, 0 }  ][line width=0.75]    (10.93,-3.29) .. controls (6.95,-1.4) and (3.31,-0.3) .. (0,0) .. controls (3.31,0.3) and (6.95,1.4) .. (10.93,3.29)   ;
		\draw  [fill={rgb, 255:red, 208; green, 2; blue, 27 }  ,fill opacity=0.1 ] (20.14,21.62) .. controls (37.34,5.22) and (275.74,7.62) .. (302.94,18.02) .. controls (330.14,28.42) and (320.55,165.24) .. (300.54,182.42) .. controls (280.53,199.6) and (46.54,200.02) .. (25.74,180.02) .. controls (4.94,160.02) and (2.94,38.02) .. (20.14,21.62) -- cycle ;
		\draw  [fill={rgb, 255:red, 74; green, 144; blue, 226 }  ,fill opacity=0.3 ] (60.14,25.22) .. controls (77.34,8.82) and (268.54,12.02) .. (295.74,22.42) .. controls (322.94,32.82) and (306.15,155.24) .. (286.14,172.42) .. controls (266.13,189.6) and (167.88,191.47) .. (147.08,171.47) .. controls (126.28,151.47) and (42.94,41.62) .. (60.14,25.22) -- cycle ;
		\draw  [fill={rgb, 255:red, 184; green, 233; blue, 134 }  ,fill opacity=0.3 ] (151.34,29.22) .. controls (168.54,12.82) and (261.74,16.42) .. (288.94,26.82) .. controls (316.14,37.22) and (305.75,141.64) .. (285.74,158.82) .. controls (265.73,176) and (221.48,134.67) .. (200.68,114.67) .. controls (179.88,94.67) and (134.14,45.62) .. (151.34,29.22) -- cycle ;
		\draw   (72.93,28.08) -- (93.93,28.08) -- (93.93,49.08) -- (72.93,49.08) -- cycle ;
		\draw   (256.38,129.78) .. controls (256.38,124.12) and (260.97,119.53) .. (266.63,119.53) .. controls (272.29,119.53) and (276.88,124.12) .. (276.88,129.78) .. controls (276.88,135.44) and (272.29,140.03) .. (266.63,140.03) .. controls (260.97,140.03) and (256.38,135.44) .. (256.38,129.78) -- cycle ;
		\draw   (256.08,63.03) .. controls (256.08,57.37) and (260.67,52.78) .. (266.33,52.78) .. controls (271.99,52.78) and (276.58,57.37) .. (276.58,63.03) .. controls (276.58,68.69) and (271.99,73.28) .. (266.33,73.28) .. controls (260.67,73.28) and (256.08,68.69) .. (256.08,63.03) -- cycle ;
		\draw   (15.98,104.58) .. controls (15.98,98.92) and (20.57,94.33) .. (26.23,94.33) .. controls (31.89,94.33) and (36.48,98.92) .. (36.48,104.58) .. controls (36.48,110.24) and (31.89,114.83) .. (26.23,114.83) .. controls (20.57,114.83) and (15.98,110.24) .. (15.98,104.58) -- cycle ;
		\draw   (75.18,156.58) .. controls (75.18,150.92) and (79.77,146.33) .. (85.43,146.33) .. controls (91.09,146.33) and (95.68,150.92) .. (95.68,156.58) .. controls (95.68,162.24) and (91.09,166.83) .. (85.43,166.83) .. controls (79.77,166.83) and (75.18,162.24) .. (75.18,156.58) -- cycle ;
		\draw   (160.13,28.08) -- (181.13,28.08) -- (181.13,49.08) -- (160.13,49.08) -- cycle ;
		\draw   (161.73,146.48) -- (182.73,146.48) -- (182.73,167.48) -- (161.73,167.48) -- cycle ;
		\draw    (72.73,38.88) .. controls (36.1,38.68) and (23,45.15) .. (26.13,92.87) ;
		\draw [shift={(26.23,94.33)}, rotate = 265.92] [color={rgb, 255:red, 0; green, 0; blue, 0 }  ][line width=0.75]    (10.93,-3.29) .. controls (6.95,-1.4) and (3.31,-0.3) .. (0,0) .. controls (3.31,0.3) and (6.95,1.4) .. (10.93,3.29)   ;
		\draw [color={rgb, 255:red, 208; green, 2; blue, 27 }  ,draw opacity=1 ]   (93.93,38.88) .. controls (110.4,38.1) and (127.63,38.08) .. (158.05,38.46) ;
		\draw [shift={(159.93,38.48)}, rotate = 180.73] [color={rgb, 255:red, 208; green, 2; blue, 27 }  ,draw opacity=1 ][line width=0.75]    (10.93,-3.29) .. controls (6.95,-1.4) and (3.31,-0.3) .. (0,0) .. controls (3.31,0.3) and (6.95,1.4) .. (10.93,3.29)   ;
		\draw [color={rgb, 255:red, 0; green, 0; blue, 0 }  ,draw opacity=1 ]   (266.63,119.53) .. controls (265.95,107.2) and (265.93,99.28) .. (266.3,75.17) ;
		\draw [shift={(266.33,73.28)}, rotate = 90.9] [color={rgb, 255:red, 0; green, 0; blue, 0 }  ,draw opacity=1 ][line width=0.75]    (10.93,-3.29) .. controls (6.95,-1.4) and (3.31,-0.3) .. (0,0) .. controls (3.31,0.3) and (6.95,1.4) .. (10.93,3.29)   ;
		\draw [color={rgb, 255:red, 144; green, 19; blue, 254 }  ,draw opacity=1 ]   (181.13,39.28) .. controls (230.53,40.46) and (243.04,31.93) .. (264.98,51.54) ;
		\draw [shift={(266.33,52.78)}, rotate = 222.92] [color={rgb, 255:red, 144; green, 19; blue, 254 }  ,draw opacity=1 ][line width=0.75]    (10.93,-3.29) .. controls (6.95,-1.4) and (3.31,-0.3) .. (0,0) .. controls (3.31,0.3) and (6.95,1.4) .. (10.93,3.29)   ;
		\draw [color={rgb, 255:red, 208; green, 2; blue, 27 }  ,draw opacity=1 ]   (183.33,157.08) .. controls (219.8,156.49) and (257.59,161.8) .. (266.04,141.63) ;
		\draw [shift={(266.63,140.03)}, rotate = 107.85] [color={rgb, 255:red, 208; green, 2; blue, 27 }  ,draw opacity=1 ][line width=0.75]    (10.93,-3.29) .. controls (6.95,-1.4) and (3.31,-0.3) .. (0,0) .. controls (3.31,0.3) and (6.95,1.4) .. (10.93,3.29)   ;
		\draw [color={rgb, 255:red, 208; green, 2; blue, 27 }  ,draw opacity=1 ]   (172.73,146.88) .. controls (172.73,110.24) and (171.96,86.01) .. (171.55,50.51) ;
		\draw [shift={(171.53,48.88)}, rotate = 89.37] [color={rgb, 255:red, 208; green, 2; blue, 27 }  ,draw opacity=1 ][line width=0.75]    (10.93,-3.29) .. controls (6.95,-1.4) and (3.31,-0.3) .. (0,0) .. controls (3.31,0.3) and (6.95,1.4) .. (10.93,3.29)   ;
		\draw    (75.18,156.58) .. controls (44.01,156.48) and (28.4,156.48) .. (26.32,116.68) ;
		\draw [shift={(26.23,114.83)}, rotate = 87.66] [color={rgb, 255:red, 0; green, 0; blue, 0 }  ][line width=0.75]    (10.93,-3.29) .. controls (6.95,-1.4) and (3.31,-0.3) .. (0,0) .. controls (3.31,0.3) and (6.95,1.4) .. (10.93,3.29)   ;
		\draw [color={rgb, 255:red, 0; green, 0; blue, 0 }  ,draw opacity=1 ]   (95.68,156.58) .. controls (109.97,156.87) and (137.32,156.12) .. (160.18,156.45) ;
		\draw [shift={(161.93,156.48)}, rotate = 180.99] [color={rgb, 255:red, 0; green, 0; blue, 0 }  ,draw opacity=1 ][line width=0.75]    (10.93,-3.29) .. controls (6.95,-1.4) and (3.31,-0.3) .. (0,0) .. controls (3.31,0.3) and (6.95,1.4) .. (10.93,3.29)   ;
		\draw [color={rgb, 255:red, 0; green, 0; blue, 0 }  ,draw opacity=1 ]   (36.48,104.58) .. controls (76.92,104.09) and (83.73,85.06) .. (84.31,50.85) ;
		\draw [shift={(84.33,49.28)}, rotate = 90.65] [color={rgb, 255:red, 0; green, 0; blue, 0 }  ,draw opacity=1 ][line width=0.75]    (10.93,-3.29) .. controls (6.95,-1.4) and (3.31,-0.3) .. (0,0) .. controls (3.31,0.3) and (6.95,1.4) .. (10.93,3.29)   ;
		\draw   (258.73,63.03) .. controls (258.73,58.83) and (262.14,55.43) .. (266.33,55.43) .. controls (270.53,55.43) and (273.93,58.83) .. (273.93,63.03) .. controls (273.93,67.23) and (270.53,70.63) .. (266.33,70.63) .. controls (262.14,70.63) and (258.73,67.23) .. (258.73,63.03) -- cycle ;
		\draw  [fill={rgb, 255:red, 245; green, 166; blue, 35 }  ,fill opacity=0.4 ] (236.94,42.9) .. controls (240.54,37.3) and (254.94,32.5) .. (266.54,36.9) .. controls (278.14,41.3) and (315.74,73.3) .. (295.73,90.48) .. controls (275.73,107.66) and (246.94,97.3) .. (237.34,75.7) .. controls (227.74,54.1) and (233.34,48.5) .. (236.94,42.9) -- cycle ;
		\draw    (273.94,56.79) .. controls (316.84,18.51) and (329.03,78.24) .. (278.14,63.5) ;
		\draw [shift={(276.58,63.03)}, rotate = 17.32] [color={rgb, 255:red, 0; green, 0; blue, 0 }  ][line width=0.75]    (10.93,-3.29) .. controls (6.95,-1.4) and (3.31,-0.3) .. (0,0) .. controls (3.31,0.3) and (6.95,1.4) .. (10.93,3.29)   ;
		
		\draw (589,60) node [anchor=north west][inner sep=0.75pt]  [font=\scriptsize]  {$v_1$};
		\draw (406,34.8) node [anchor=north west][inner sep=0.75pt]  [font=\scriptsize]  {$v_5$};
		\draw (496,154) node [anchor=north west][inner sep=0.75pt]  [font=\scriptsize]  {$v_4$};
		\draw (494,35.2) node [anchor=north west][inner sep=0.75pt]  [font=\scriptsize]  {$v_3$};
		\draw (589,127.2) node [anchor=north west][inner sep=0.75pt]  [font=\scriptsize]  {$v_2$};
		\draw (408,154.4) node [anchor=north west][inner sep=0.75pt]  [font=\scriptsize]  {$v_6$};
		\draw (348.5,102) node [anchor=north west][inner sep=0.75pt]  [font=\scriptsize]  {$v_7$};
		\draw (257,57.88) node [anchor=north west][inner sep=0.75pt]  [font=\scriptsize]  {$v_1$};
		\draw (75,34) node [anchor=north west][inner sep=0.75pt]  [font=\scriptsize]  {$v_5$};
		\draw (164,151) node [anchor=north west][inner sep=0.75pt]  [font=\scriptsize]  {$v_4$};
		\draw (161,33.08) node [anchor=north west][inner sep=0.75pt]  [font=\scriptsize]  {$v_3$};
		\draw (258,124.28) node [anchor=north west][inner sep=0.75pt]  [font=\scriptsize]  {$v_2$};
		\draw (76.5,151.28) node [anchor=north west][inner sep=0.75pt]  [font=\scriptsize]  {$v_6$};
		\draw (17,99) node [anchor=north west][inner sep=0.75pt]  [font=\scriptsize]  {$v_7$};

	\end{tikzpicture}
	
\caption[Difference between $ \mathsf{pre} $ and $ \mathsf{tpre} $]{Computation of  $ \mu X. ~U\cup \pre{}{X}$ (left) and $\mu X. ~U\cup \tpre{}{X}  $ (right). Each colored region describes one iteration over $ X $. The dotted region on the right is added by the $ \mathsf{attr} $ part of $ \mathsf{tpre} $, and this allows only the vertex $ v_5 $ to be in $ \front(\{v_1\}) $.  Each set of the same colored edges defines a live transition group. 
}\label{fig:PreVsTpre}
\end{figure}

\smallskip
\noindent\textbf{Computing live group assumptions.}
Intuitively, the operator $ \textsf{tpre}_{\gamegraph} $ computes the union of (i)  the set of vertices from which $ \pz $ can reach $ U $ in a finite number of steps with no cooperation from $\po$ and (ii) the set of $\po$ vertices from which $\pz$ can reach $U$ with at most \emph{one-time} cooperation from $\po$.
 Looking at Fig.~\ref{fig:PreVsTpre}, case (i) is indicated by the dotted line, while case (ii) corresponds to the last added $\po$ vertex (e.g., $v_5$). Hence, we need to capture the cooperation needed by $\po$ only from the vertices added last, which we call the \emph{frontier} of $ U $ in $ \gamegraph $ and are formalized as follows:
\begin{equation}
  \front(U):=\tpre{\gamegraph}{U}\setminus \attrz{\gamegraph}{U}.
\end{equation}
It is easy to see that, indeed $ \front(U)\subseteq V^1 $, as whenever $ v\in\front(U)\cap V^0 $, then it would have been the case that $v\in\attrz{\gamegraph}{U} $ via \eqref{eq:FP:buechi}.

Defining live groups based on frontiers instead of all elements in $X^i$ indeed yields the desired permissive assumption for Büchi games. By observing that we additionally need to ensure that $\po$ never leaves the cooperative winning region by a simple safety assumption, we get the following result, which is the main contribution of this section and is proved in the appendix.

\begin{restatable}{theorem}{restatebuchi}\label{thm:Buechi assumptions}
	Let $ \game=\tup{\gamegraph = (V,E),\spec = \square\lozenge U} $ be a $\buchi$ game with $ Z^*=\TsolveBuchi(\gamegraph,U)$ and $\livegroup=\left\{\livegroupSingle\right\}_{i\geq 0}$ s.t.\
		\begin{equation}\label{eq: buechi live groups}
	\emptyset\neq \livegroupSingle:= (\front(X^i)\times (X^{i+1}\setminus \front(X^i)))\cap E,
\end{equation}
where $X^i$ is the set computed in the $ i $-th iteration of the computation over $ X $ and in the last iteration of the computation over $ Y $ in \TsolveBuchi.
	Then $\assump =  \assumpsafe(\safegroup)\land\assumpgrlive(\livegroup)$ is an \aname for $\game$, where $ \safegroup  = \computeSafe(\gamegraph, U)$. We write $\computeLive(\gamegraph,U)$ to denote the algorithm to construct live groups $\livegroup$ as above, which runs in time $ \bigO(n^3) $, where $ n=|V| $.
\end{restatable}

In fact, there is a faster algorithm that runs in time linear in the size of the graph for computation of APAs for \buchi games, which we present in Appendix \ref{sec:fasterbuchi}. We chose to present the mu-calculus based algorithm here, because it provides more insights into the nature of live groups.


\subsection{Co-Liveness Assumptions in Co-B\"uchi Games}\label{section:coBuchiGames}

A co-B\"uchi game is the \emph{dual} of a \buchi game, where a winning play should visit a designated set of vertices only finitely many times. Formally, a co-B\"uchi game is a tuple $\game=(\gamegraph,\spec)$ where $\spec=\lozenge\square U$ for some $U\subseteq V$.
 The standard symbolic algorithm to compute the cooperative winning region is as follows:
\begin{equation}\label{eq:EnvironmentcoBuchiFixpoint}
		\solveCobuchi(\gamegraph,U):=\mu X. \nu Y.~(U\cap \pre{}{Y})\cup (\pre{}{X}). 
	\end{equation}
As before, the sets $ X^i $ obtained in the $ i $-th computation of $ X $ during the evaluation of \eqref{eq:EnvironmentcoBuchiFixpoint} carry essential information for constructing assumptions. Intuitively, $ X^1 $ gives precisely the set of vertices from which the play can stay in $ U $ with $ \po $'s cooperation and we would like an assumption to capture the fact that we do not want $ \po $ to go further away from $ X^1 $ infinitely often. This observation is naturally described by so called co-liveness templates.

\begin{definition}
 	Let $\gamegraph=(V,E)$ be a game graph and $ \colivegroup\subseteq V\times V $ a set of edges. Then a \emph{co-liveness template} over $\gamegraph$ w.r.t.\ $\colivegroup$ is defined by the LTL formula
 	\begin{equation}
 		\textstyle \assumpdep(\colivegroup) \coloneqq \lozenge\square \bigwedge_{e\in \colivegroup} \neg e.
 	\end{equation}
\end{definition}

The assumptions employing co-liveness templates will be called co-liveness assumptions. With this, we can state the main result of this section.

\begin{restatable}{theorem}{restatecobuchi}\label{thm:coBuechi assumptions}
	Let $ \game=\tup{\gamegraph = (V,E), \lozenge\square U} $, $ Z^*=\solveCobuchi(\gamegraph,U)$ and 
		\begin{equation}\label{eq: cobuechi colive edges}
		\textstyle
			\colivegroup =\left(\
	\begin{aligned}\textstyle
		\left[(X^1\cap V^1) \times (Z^*\setminus X^{1})\right] ~\cup 
		\left[\bigcup_{i>1} (X^{i}\cap V^1)\times (Z^*\setminus X^{i-1})\right]
	\end{aligned}\right)\cap E,
		\end{equation}
	where $ X^i $ is the set computed in the $ i $-th iteration of fixpoint variable $ X $ in $\solveCobuchi$. Then $\assump =  \assumpsafe(\safegroup)\land\assumpdep(\colivegroup)$ is an \aname for $\game$, where $ \safegroup  = \computeSafe(\gamegraph, U)$. 
	We write $\computeCoLive(\gamegraph,U)$ to denote the algorithm constructing co-live edges~$\colivegroup$ as above which runs in time $ \bigO(n^3) $, where $ n=|V| $.
\end{restatable}

We observe that $ X_1 $ is a subset of $ U $ such that if a play reaches $ X^1 $, $ \pz $ and $ \po $ can cooperatively keep the play in $ X^1 $. To do so, we ensure via the definition of $ \colivegroup $ in \eqref{eq: cobuechi colive edges} that $ \po $ can only leave $ X^1 $ finitely often. Moreover, with the other co-live edges in $ \colivegroup $, we ensure that $ \po $ can only go away from $ X^1 $ finitely often, and hence if $ \pz $ plays their strategy to reach $ X^1 $ and then stay there, the play will be winning. The permissiveness of the assumption comes from the observation that if co-liveness is violated, then $ \po $ takes a co-live edge infinitely often, and hence leaves $ X^1 $ infinitely often, implying leaving $ U $ infinitely often. We refer the reader to the Appendix~\ref{proof:cobuechi assumption is correct} for a formal proof of Theorem \ref{thm:coBuechi assumptions}.

 In the context of \eqref{eq:EnvironmentcoBuchiFixpoint} the set
 \begin{equation}\label{eq: cobuechi colive edges with pre}
 \textstyle
 	\colivegroup =\left(\
 	\begin{aligned}
 		\left[(X^1\cap V^1) \times (Z^*\setminus X^{1})\right] ~\cup 
 		\left[\bigcup_{i>1} (X^{i}\cap V^1)\times (Z^*\setminus X^{i-1})\right]
 	\end{aligned}\right)\cap E,
 \end{equation}
  results in the desired co-live assumptions. We argued that this defines the adequately permissive assumption for co-B\"uchi games. However, utilizing the observation from Section ~\ref{section:BuchiGames} we can equivalently use the accelerated fixed-point algorithm resulting from replacing the pre-operator over $X$ in \eqref{eq:EnvironmentcoBuchiFixpoint} by the $\mathsf{tpre}$ operator. This again only accelerates the computation, as formalized in the following proposition and visualized in Fig.~\ref{fig:PreVsTpreCoBuchi}.

 \begin{restatable}{proposition}{restate:prop:cobuechi_tpre}\label{prop:cobuechi_tpre}
 	Let $ \game=\tup{\gamegraph,\lozenge\square U} $ be a $\cobuchi$ game and 
 	\begin{equation}\label{eq:FP:cobuechi}
 	\TsolveCobuchi(\gamegraph,U)= \mu X.\nu Y.~(U\cap \pre{}{Y})\cup (\tpre{}{X}).
 \end{equation}
 Then $ \TsolveCobuchi(\gamegraph,U)=\solveCobuchi(\gamegraph,U)=\team{0,1}\lozenge\square U$.
 \end{restatable}
 	
 When using the accelerated fixed-point algorithm, we can again restrict attention to $\po$ vertices in the frontier of $X^i$ for the construction of assumptions. In section \ref{section:BuchiGames}, for \buchi game $ \buchigame(U) $, we introduced live groups to take the play towards $ U $, whenever $ \po $ can. But for \cobuchi games, we need to restrict $ \po $ from going away from the region $ U_{good} $ where the play stays in $ U $. This requires co-liveness assumption only over frontiers of $ X^i $s, since any other vertex of $ \po $ added in $ X^{i+1} $ is added in the $ \mathsf{attr}^0 $ part of $ \mathsf{tpre} $, and hence can not go away from $ U_{good} $ anyway.
  With this, we have the following main result of this section, for which we provide the proof in Appendix~\ref{sec:accelcobuchi}.

 \begin{restatable}{theorem}{Trestatecobuchi}\label{thm:TcoBuechi assumptions}
 	Let $ \game=\tup{\gamegraph = (V,E), \spec = \lozenge\square U} $ be a \cobuchi game with $ Z^*=\TsolveCobuchi(\gamegraph,U)$ and 
 		\begin{equation}\label{eq: Tcobuechi colive edges}
 			\colivegroup =\left(\
 			 \begin{aligned}
 			 	(X^1\cap V^1) \times (Z^*\setminus X^{1})~\cup \\
 				\bigcup_{i>1} \front(X^i)\times (\front(X^i) \cup Z^*\setminus X^{i+1})
 			\end{aligned}\right)\cap E,
 		\end{equation}
 	where $ X^i $ is the set computed in the $ i $-th iteration of fixpoint variable $ X $. Then $\assump =  \assumpsafe(\safegroup)\land\assumpdep(\colivegroup)$ is an \aname for $\game$, where $ \safegroup  = \computeSafe(\gamegraph, U)$. Moreover, $ \colivegroup $ can be constructed in time $ \bigO(n^3) $, where $ n $ is the number of vertices.
 \end{restatable}
 

 	In fact, there is again a faster algorithm that runs in time linear in size of the graph for computation of APAs for \cobuchi games, which we present in the Appendix \ref{sec:fastercobuchi}. We chose to present this version for the same reasons as for the \buchi games. 
 	
 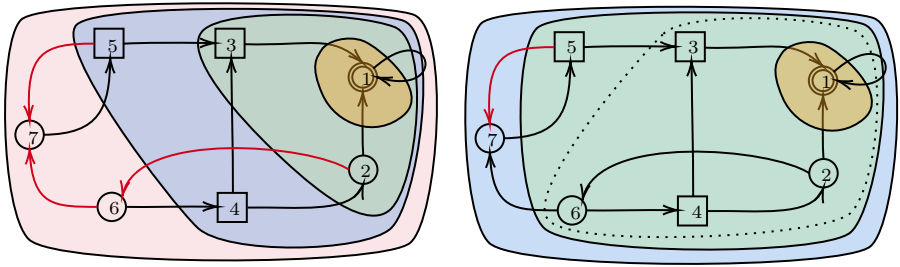
\begin{figure}[h]
 	\centering

 	\tikzset{every picture/.style={line width=0.75pt}} 
 	
 	\begin{tikzpicture}[x=0.75pt,y=0.75pt,yscale=-0.7,xscale=0.7]
 		
 		\draw  [fill={rgb, 255:red, 74; green, 144; blue, 226 }  ,fill opacity=0.3 ] (351.41,43.34) .. controls (368.61,26.94) and (607.01,29.34) .. (634.21,39.74) .. controls (661.41,50.14) and (651.81,186.96) .. (631.81,204.14) .. controls (611.8,221.32) and (377.81,221.74) .. (357.01,201.74) .. controls (336.21,181.74) and (334.21,59.74) .. (351.41,43.34) -- cycle ;
 		\draw  [fill={rgb, 255:red, 184; green, 233; blue, 134 }  ,fill opacity=0.3 ] (391.41,46.94) .. controls (408.61,30.54) and (599.81,33.74) .. (627.01,44.14) .. controls (654.21,54.54) and (637.41,176.96) .. (617.41,194.14) .. controls (597.4,211.32) and (417.41,211.07) .. (396.61,191.07) .. controls (375.81,171.07) and (374.21,63.34) .. (391.41,46.94) -- cycle ;
 		\draw   (404.2,49.8) -- (425.2,49.8) -- (425.2,70.8) -- (404.2,70.8) -- cycle ;
 		\draw   (587.65,151.5) .. controls (587.65,145.84) and (592.24,141.25) .. (597.9,141.25) .. controls (603.56,141.25) and (608.15,145.84) .. (608.15,151.5) .. controls (608.15,157.16) and (603.56,161.75) .. (597.9,161.75) .. controls (592.24,161.75) and (587.65,157.16) .. (587.65,151.5) -- cycle ;
 		\draw   (587.35,84.75) .. controls (587.35,79.09) and (591.94,74.5) .. (597.6,74.5) .. controls (603.26,74.5) and (607.85,79.09) .. (607.85,84.75) .. controls (607.85,90.41) and (603.26,95) .. (597.6,95) .. controls (591.94,95) and (587.35,90.41) .. (587.35,84.75) -- cycle ;
 		\draw   (347.25,126.3) .. controls (347.25,120.64) and (351.84,116.05) .. (357.5,116.05) .. controls (363.16,116.05) and (367.75,120.64) .. (367.75,126.3) .. controls (367.75,131.96) and (363.16,136.55) .. (357.5,136.55) .. controls (351.84,136.55) and (347.25,131.96) .. (347.25,126.3) -- cycle ;
 		\draw   (406.45,178.3) .. controls (406.45,172.64) and (411.04,168.05) .. (416.7,168.05) .. controls (422.36,168.05) and (426.95,172.64) .. (426.95,178.3) .. controls (426.95,183.96) and (422.36,188.55) .. (416.7,188.55) .. controls (411.04,188.55) and (406.45,183.96) .. (406.45,178.3) -- cycle ;
 		\draw   (491.4,49.8) -- (512.4,49.8) -- (512.4,70.8) -- (491.4,70.8) -- cycle ;
 		\draw   (493,168.2) -- (514,168.2) -- (514,189.2) -- (493,189.2) -- cycle ;
 		\draw [color={rgb, 255:red, 208; green, 2; blue, 27 }  ,draw opacity=1 ]   (404,60.6) .. controls (367.37,60.4) and (354.26,66.87) .. (357.4,114.59) ;
 		\draw [shift={(357.5,116.05)}, rotate = 265.92] [color={rgb, 255:red, 208; green, 2; blue, 27 }  ,draw opacity=1 ][line width=0.75]    (10.93,-3.29) .. controls (6.95,-1.4) and (3.31,-0.3) .. (0,0) .. controls (3.31,0.3) and (6.95,1.4) .. (10.93,3.29)   ;
 		\draw [color={rgb, 255:red, 0; green, 0; blue, 0 }  ,draw opacity=1 ]   (425.2,60.6) .. controls (441.66,59.82) and (458.9,59.8) .. (489.32,60.18) ;
 		\draw [shift={(491.2,60.2)}, rotate = 180.73] [color={rgb, 255:red, 0; green, 0; blue, 0 }  ,draw opacity=1 ][line width=0.75]    (10.93,-3.29) .. controls (6.95,-1.4) and (3.31,-0.3) .. (0,0) .. controls (3.31,0.3) and (6.95,1.4) .. (10.93,3.29)   ;
 		\draw    (597.9,141.25) .. controls (597.22,128.92) and (597.2,121) .. (597.57,96.89) ;
 		\draw [shift={(597.6,95)}, rotate = 90.9] [color={rgb, 255:red, 0; green, 0; blue, 0 }  ][line width=0.75]    (10.93,-3.29) .. controls (6.95,-1.4) and (3.31,-0.3) .. (0,0) .. controls (3.31,0.3) and (6.95,1.4) .. (10.93,3.29)   ;
 		\draw    (512.4,61) .. controls (561.79,62.18) and (574.3,53.65) .. (596.24,73.26) ;
 		\draw [shift={(597.6,74.5)}, rotate = 222.92] [color={rgb, 255:red, 0; green, 0; blue, 0 }  ][line width=0.75]    (10.93,-3.29) .. controls (6.95,-1.4) and (3.31,-0.3) .. (0,0) .. controls (3.31,0.3) and (6.95,1.4) .. (10.93,3.29)   ;
 		\draw    (514.6,178.8) .. controls (551.07,178.22) and (588.86,183.52) .. (597.31,163.35) ;
 		\draw [shift={(597.9,161.75)}, rotate = 107.85] [color={rgb, 255:red, 0; green, 0; blue, 0 }  ][line width=0.75]    (10.93,-3.29) .. controls (6.95,-1.4) and (3.31,-0.3) .. (0,0) .. controls (3.31,0.3) and (6.95,1.4) .. (10.93,3.29)   ;
 		\draw    (504,168.6) .. controls (504,131.96) and (503.22,107.73) .. (502.82,72.23) ;
 		\draw [shift={(502.8,70.6)}, rotate = 89.37] [color={rgb, 255:red, 0; green, 0; blue, 0 }  ][line width=0.75]    (10.93,-3.29) .. controls (6.95,-1.4) and (3.31,-0.3) .. (0,0) .. controls (3.31,0.3) and (6.95,1.4) .. (10.93,3.29)   ;
 		\draw    (406.45,178.3) .. controls (375.27,178.2) and (359.67,178.2) .. (357.59,138.4) ;
 		\draw [shift={(357.5,136.55)}, rotate = 87.66] [color={rgb, 255:red, 0; green, 0; blue, 0 }  ][line width=0.75]    (10.93,-3.29) .. controls (6.95,-1.4) and (3.31,-0.3) .. (0,0) .. controls (3.31,0.3) and (6.95,1.4) .. (10.93,3.29)   ;
 		\draw    (426.95,178.3) .. controls (441.23,178.59) and (468.59,177.84) .. (491.45,178.17) ;
 		\draw [shift={(493.2,178.2)}, rotate = 180.99] [color={rgb, 255:red, 0; green, 0; blue, 0 }  ][line width=0.75]    (10.93,-3.29) .. controls (6.95,-1.4) and (3.31,-0.3) .. (0,0) .. controls (3.31,0.3) and (6.95,1.4) .. (10.93,3.29)   ;
 		\draw    (367.75,126.3) .. controls (408.18,125.81) and (415,106.78) .. (415.58,72.57) ;
 		\draw [shift={(415.6,71)}, rotate = 90.65] [color={rgb, 255:red, 0; green, 0; blue, 0 }  ][line width=0.75]    (10.93,-3.29) .. controls (6.95,-1.4) and (3.31,-0.3) .. (0,0) .. controls (3.31,0.3) and (6.95,1.4) .. (10.93,3.29)   ;
 		\draw   (590,84.75) .. controls (590,80.55) and (593.4,77.15) .. (597.6,77.15) .. controls (601.8,77.15) and (605.2,80.55) .. (605.2,84.75) .. controls (605.2,88.95) and (601.8,92.35) .. (597.6,92.35) .. controls (593.4,92.35) and (590,88.95) .. (590,84.75) -- cycle ;
 		\draw  [fill={rgb, 255:red, 245; green, 166; blue, 35 }  ,fill opacity=0.4 ] (568.21,64.62) .. controls (571.81,59.02) and (586.21,54.22) .. (597.81,58.62) .. controls (609.41,63.02) and (647.01,95.02) .. (627,112.2) .. controls (606.99,129.38) and (578.21,119.02) .. (568.61,97.42) .. controls (559.01,75.82) and (564.61,70.22) .. (568.21,64.62) -- cycle ;
 		\draw  [dash pattern={on 0.84pt off 2.51pt}] (482.61,50.94) .. controls (499.81,34.54) and (593.01,38.14) .. (620.21,48.54) .. controls (647.41,58.94) and (637.01,163.36) .. (617.01,180.54) .. controls (597,197.72) and (421.01,205.87) .. (400.21,185.87) .. controls (379.41,165.87) and (465.41,67.34) .. (482.61,50.94) -- cycle ;
 		\draw    (605.21,78.51) .. controls (648.11,40.23) and (660.3,99.96) .. (609.41,85.22) ;
 		\draw [shift={(607.85,84.75)}, rotate = 17.32] [color={rgb, 255:red, 0; green, 0; blue, 0 }  ][line width=0.75]    (10.93,-3.29) .. controls (6.95,-1.4) and (3.31,-0.3) .. (0,0) .. controls (3.31,0.3) and (6.95,1.4) .. (10.93,3.29)   ;
 		\draw  [fill={rgb, 255:red, 208; green, 2; blue, 27 }  ,fill opacity=0.1 ] (19.81,40.82) .. controls (37.01,24.42) and (275.41,26.82) .. (302.61,37.22) .. controls (329.81,47.62) and (320.21,184.44) .. (300.21,201.62) .. controls (280.2,218.8) and (46.21,219.22) .. (25.41,199.22) .. controls (4.61,179.22) and (2.61,57.22) .. (19.81,40.82) -- cycle ;
 		\draw  [fill={rgb, 255:red, 74; green, 144; blue, 226 }  ,fill opacity=0.3 ] (59.81,44.42) .. controls (77.01,28.02) and (268.21,31.22) .. (295.41,41.62) .. controls (322.61,52.02) and (305.81,174.44) .. (285.81,191.62) .. controls (265.8,208.8) and (167.55,210.67) .. (146.75,190.67) .. controls (125.95,170.67) and (42.61,60.82) .. (59.81,44.42) -- cycle ;
 		\draw  [fill={rgb, 255:red, 184; green, 233; blue, 134 }  ,fill opacity=0.3 ] (151.01,48.42) .. controls (168.21,32.02) and (261.41,35.62) .. (288.61,46.02) .. controls (315.81,56.42) and (305.41,160.84) .. (285.41,178.02) .. controls (265.4,195.2) and (221.15,153.87) .. (200.35,133.87) .. controls (179.55,113.87) and (133.81,64.82) .. (151.01,48.42) -- cycle ;
 		\draw   (72.6,47.28) -- (93.6,47.28) -- (93.6,68.28) -- (72.6,68.28) -- cycle ;
 		\draw   (256.05,148.98) .. controls (256.05,143.32) and (260.64,138.73) .. (266.3,138.73) .. controls (271.96,138.73) and (276.55,143.32) .. (276.55,148.98) .. controls (276.55,154.64) and (271.96,159.23) .. (266.3,159.23) .. controls (260.64,159.23) and (256.05,154.64) .. (256.05,148.98) -- cycle ;
 		\draw   (255.75,82.23) .. controls (255.75,76.57) and (260.34,71.98) .. (266,71.98) .. controls (271.66,71.98) and (276.25,76.57) .. (276.25,82.23) .. controls (276.25,87.89) and (271.66,92.48) .. (266,92.48) .. controls (260.34,92.48) and (255.75,87.89) .. (255.75,82.23) -- cycle ;
 		\draw   (15.65,123.78) .. controls (15.65,118.12) and (20.24,113.53) .. (25.9,113.53) .. controls (31.56,113.53) and (36.15,118.12) .. (36.15,123.78) .. controls (36.15,129.44) and (31.56,134.03) .. (25.9,134.03) .. controls (20.24,134.03) and (15.65,129.44) .. (15.65,123.78) -- cycle ;
 		\draw   (74.85,175.78) .. controls (74.85,170.12) and (79.44,165.53) .. (85.1,165.53) .. controls (90.76,165.53) and (95.35,170.12) .. (95.35,175.78) .. controls (95.35,181.44) and (90.76,186.03) .. (85.1,186.03) .. controls (79.44,186.03) and (74.85,181.44) .. (74.85,175.78) -- cycle ;
 		\draw   (159.8,47.28) -- (180.8,47.28) -- (180.8,68.28) -- (159.8,68.28) -- cycle ;
 		\draw   (161.4,165.68) -- (182.4,165.68) -- (182.4,186.68) -- (161.4,186.68) -- cycle ;
 		\draw [color={rgb, 255:red, 208; green, 2; blue, 27 }  ,draw opacity=1 ]   (72.4,58.08) .. controls (35.77,57.88) and (22.66,64.35) .. (25.8,112.07) ;
 		\draw [shift={(25.9,113.53)}, rotate = 265.92] [color={rgb, 255:red, 208; green, 2; blue, 27 }  ,draw opacity=1 ][line width=0.75]    (10.93,-3.29) .. controls (6.95,-1.4) and (3.31,-0.3) .. (0,0) .. controls (3.31,0.3) and (6.95,1.4) .. (10.93,3.29)   ;
 		\draw [color={rgb, 255:red, 0; green, 0; blue, 0 }  ,draw opacity=1 ]   (93.6,58.08) .. controls (110.06,57.3) and (127.3,57.28) .. (157.72,57.66) ;
 		\draw [shift={(159.6,57.68)}, rotate = 180.73] [color={rgb, 255:red, 0; green, 0; blue, 0 }  ,draw opacity=1 ][line width=0.75]    (10.93,-3.29) .. controls (6.95,-1.4) and (3.31,-0.3) .. (0,0) .. controls (3.31,0.3) and (6.95,1.4) .. (10.93,3.29)   ;
 		\draw [color={rgb, 255:red, 0; green, 0; blue, 0 }  ,draw opacity=1 ]   (266.3,138.73) .. controls (265.62,126.4) and (265.6,118.48) .. (265.97,94.37) ;
 		\draw [shift={(266,92.48)}, rotate = 90.9] [color={rgb, 255:red, 0; green, 0; blue, 0 }  ,draw opacity=1 ][line width=0.75]    (10.93,-3.29) .. controls (6.95,-1.4) and (3.31,-0.3) .. (0,0) .. controls (3.31,0.3) and (6.95,1.4) .. (10.93,3.29)   ;
 		\draw [color={rgb, 255:red, 0; green, 0; blue, 0 }  ,draw opacity=1 ]   (180.8,58.48) .. controls (230.19,59.66) and (242.7,51.13) .. (264.64,70.74) ;
 		\draw [shift={(266,71.98)}, rotate = 222.92] [color={rgb, 255:red, 0; green, 0; blue, 0 }  ,draw opacity=1 ][line width=0.75]    (10.93,-3.29) .. controls (6.95,-1.4) and (3.31,-0.3) .. (0,0) .. controls (3.31,0.3) and (6.95,1.4) .. (10.93,3.29)   ;
 		\draw [color={rgb, 255:red, 0; green, 0; blue, 0 }  ,draw opacity=1 ]   (183,176.28) .. controls (219.47,175.69) and (257.26,181) .. (265.71,160.83) ;
 		\draw [shift={(266.3,159.23)}, rotate = 107.85] [color={rgb, 255:red, 0; green, 0; blue, 0 }  ,draw opacity=1 ][line width=0.75]    (10.93,-3.29) .. controls (6.95,-1.4) and (3.31,-0.3) .. (0,0) .. controls (3.31,0.3) and (6.95,1.4) .. (10.93,3.29)   ;
 		\draw [color={rgb, 255:red, 0; green, 0; blue, 0 }  ,draw opacity=1 ]   (172.4,166.08) .. controls (172.4,129.44) and (171.62,105.21) .. (171.22,69.71) ;
 		\draw [shift={(171.2,68.08)}, rotate = 89.37] [color={rgb, 255:red, 0; green, 0; blue, 0 }  ,draw opacity=1 ][line width=0.75]    (10.93,-3.29) .. controls (6.95,-1.4) and (3.31,-0.3) .. (0,0) .. controls (3.31,0.3) and (6.95,1.4) .. (10.93,3.29)   ;
 		\draw [color={rgb, 255:red, 208; green, 2; blue, 27 }  ,draw opacity=1 ]   (74.85,175.78) .. controls (43.67,175.68) and (28.07,175.68) .. (25.99,135.88) ;
 		\draw [shift={(25.9,134.03)}, rotate = 87.66] [color={rgb, 255:red, 208; green, 2; blue, 27 }  ,draw opacity=1 ][line width=0.75]    (10.93,-3.29) .. controls (6.95,-1.4) and (3.31,-0.3) .. (0,0) .. controls (3.31,0.3) and (6.95,1.4) .. (10.93,3.29)   ;
 		\draw [color={rgb, 255:red, 0; green, 0; blue, 0 }  ,draw opacity=1 ]   (95.35,175.78) .. controls (109.63,176.07) and (136.99,175.32) .. (159.85,175.65) ;
 		\draw [shift={(161.6,175.68)}, rotate = 180.99] [color={rgb, 255:red, 0; green, 0; blue, 0 }  ,draw opacity=1 ][line width=0.75]    (10.93,-3.29) .. controls (6.95,-1.4) and (3.31,-0.3) .. (0,0) .. controls (3.31,0.3) and (6.95,1.4) .. (10.93,3.29)   ;
 		\draw [color={rgb, 255:red, 0; green, 0; blue, 0 }  ,draw opacity=1 ]   (36.15,123.78) .. controls (76.58,123.29) and (83.4,104.26) .. (83.98,70.05) ;
 		\draw [shift={(84,68.48)}, rotate = 90.65] [color={rgb, 255:red, 0; green, 0; blue, 0 }  ,draw opacity=1 ][line width=0.75]    (10.93,-3.29) .. controls (6.95,-1.4) and (3.31,-0.3) .. (0,0) .. controls (3.31,0.3) and (6.95,1.4) .. (10.93,3.29)   ;
 		\draw   (258.4,82.23) .. controls (258.4,78.03) and (261.8,74.63) .. (266,74.63) .. controls (270.2,74.63) and (273.6,78.03) .. (273.6,82.23) .. controls (273.6,86.43) and (270.2,89.83) .. (266,89.83) .. controls (261.8,89.83) and (258.4,86.43) .. (258.4,82.23) -- cycle ;
 		\draw  [fill={rgb, 255:red, 245; green, 166; blue, 35 }  ,fill opacity=0.4 ] (236.61,62.1) .. controls (240.21,56.5) and (254.61,51.7) .. (266.21,56.1) .. controls (277.81,60.5) and (315.41,92.5) .. (295.4,109.68) .. controls (275.39,126.86) and (246.61,116.5) .. (237.01,94.9) .. controls (227.41,73.3) and (233.01,67.7) .. (236.61,62.1) -- cycle ;
 		\draw    (273.61,75.99) .. controls (316.51,37.71) and (328.7,97.44) .. (277.81,82.7) ;
 		\draw [shift={(276.25,82.23)}, rotate = 17.32] [color={rgb, 255:red, 0; green, 0; blue, 0 }  ][line width=0.75]    (10.93,-3.29) .. controls (6.95,-1.4) and (3.31,-0.3) .. (0,0) .. controls (3.31,0.3) and (6.95,1.4) .. (10.93,3.29)   ;
 		\draw [color={rgb, 255:red, 208; green, 2; blue, 27 }  ,draw opacity=1 ]   (256.05,148.98) .. controls (228.64,130.94) and (104.09,119.78) .. (93.06,166.53) ;
 		\draw [shift={(92.76,167.96)}, rotate = 280.3] [color={rgb, 255:red, 208; green, 2; blue, 27 }  ,draw opacity=1 ][line width=0.75]    (10.93,-3.29) .. controls (6.95,-1.4) and (3.31,-0.3) .. (0,0) .. controls (3.31,0.3) and (6.95,1.4) .. (10.93,3.29)   ;

 		\draw    (587.65,151.5) .. controls (560.24,133.46) and (435.69,122.3) .. (424.66,169.05) ;
 		\draw [shift={(424.36,170.48)}, rotate = 280.3] [color={rgb, 255:red, 0; green, 0; blue, 0 }  ][line width=0.75]    (10.93,-3.29) .. controls (6.95,-1.4) and (3.31,-0.3) .. (0,0) .. controls (3.31,0.3) and (6.95,1.4) .. (10.93,3.29)   ;
 		
 		\draw (594.2,79.2) node [anchor=north west][inner sep=0.75pt]  [font=\scriptsize]  {$1$};
 		\draw (410.6,54) node [anchor=north west][inner sep=0.75pt]  [font=\scriptsize]  {$5$};
 		\draw (501,173.2) node [anchor=north west][inner sep=0.75pt]  [font=\scriptsize]  {$4$};
 		\draw (498.2,54.4) node [anchor=north west][inner sep=0.75pt]  [font=\scriptsize]  {$3$};
 		\draw (594.2,146.4) node [anchor=north west][inner sep=0.75pt]  [font=\scriptsize]  {$2$};
 		\draw (413.4,173.6) node [anchor=north west][inner sep=0.75pt]  [font=\scriptsize]  {$6$};
 		\draw (353.4,121.2) node [anchor=north west][inner sep=0.75pt]  [font=\scriptsize]  {$7$};
 		\draw (262.8,77.08) node [anchor=north west][inner sep=0.75pt]  [font=\scriptsize]  {$1$};
 		\draw (80,53.68) node [anchor=north west][inner sep=0.75pt]  [font=\scriptsize]  {$5$};
 		\draw (167.8,171.08) node [anchor=north west][inner sep=0.75pt]  [font=\scriptsize]  {$4$};
 		\draw (165.4,52.28) node [anchor=north west][inner sep=0.75pt]  [font=\scriptsize]  {$3$};
 		\draw (262.2,143.48) node [anchor=north west][inner sep=0.75pt]  [font=\scriptsize]  {$2$};
 		\draw (81,170.48) node [anchor=north west][inner sep=0.75pt]  [font=\scriptsize]  {$6$};
 		\draw (22.8,119.48) node [anchor=north west][inner sep=0.75pt]  [font=\scriptsize]  {$7$};

 	\end{tikzpicture}
 	\caption{Left picture describes \cobuchi computation with $ \mathsf{pre} $, and the right with $ \mathsf{tpre} $. Each colored region describes how $ X $ grows after every iteration, and the dotted region on the right is added by the $ \mathsf{attr} $ part of $ \mathsf{tpre} $. The edges in {\color{red} red} describe the co-live edges in both cases. Again $ \mathsf{pre} $ computation would give assumptions on $ \pz $ vertices, while that with $ \mathsf{tpre} $ only gives assumptions on $ \po $'s vertices.}\label{fig:PreVsTpreCoBuchi}
 \end{figure}
\subsection{\aname Assumptions for Parity Games}\label{section:parityGames}
\noindent\textbf{Parity games.}
Let $\gamegraph = \tup{V,\vertexz,\vertexo,E}$ be a game graph, and $C = \set{C_0,\ldots,C_k}$ be a set of subsets of vertices which form a partition of $V$. Then the game $\game=(\gamegraph,\spec)$ is called a \emph{parity game} if 
 \begin{equation}\label{equ:parity}
	\textstyle\spec = \paritygame(C) \coloneqq \bigvee_{i\inodd [0;k]} \square\lozenge C_i \implies \bigvee_{j\ineven [i+1;k]} \square\lozenge C_j.
\end{equation}
The set $C$ is called the \emph{priority set} and a vertex $v$ in the set $C_i$, for $i\in [1;k]$, is said to have \emph{priority} $i$. An infinite play $\rho$ is winning for $\spec = \paritygame(C)$ if the highest priority appearing infinitely often along $\rho$ is even.

\smallskip
\noindent\textbf{Conditional live group templates.}
As seen in the previous sections, for games with simple winning conditions which require visiting a fixed set of edges infinitely often or only finitely often, a single assumption (conjoined with a simple safety assumption) suffices to characterize \anames, as there is just one way to win. However, in general parity games, there are usually multiple ways of winning: for example, in parity games with priorities $ \{0,1,2\} $, a play will be winning if either (i) it only infinitely often sees vertices of priority 0, or (ii) it sees priority 1 infinitely often but also sees priority 2 infinitely often. Intuitively, winning option (i) requires the use of co-liveness assumptions as in Sec.~\ref{section:coBuchiGames}. However, winning option (ii) actually requires the live group assumptions discussed in Sec.~\ref{section:BuchiGames} to be \emph{conditional} on whether certain states with priority 1 have actually been visited infinitely often. This is formalized by generalizing live group templates to \emph{conditional live group templates}.

\begin{definition}
Let $\gamegraph=(V,\vertexz,\vertexo,E)$ be a game graph. Then a \emph{conditional live group} over $\gamegraph$ is a pair $ (R, \livegroup ) $, where $ R\subseteq V $ and $ \livegroup $ is a live group. Given a set of conditional live groups $ \condlivegroup $ we define a \emph{conditional live group template} as the LTL formula
\begin{equation}
\textstyle
\assumpcondlive(\condlivegroup) \coloneqq \bigwedge_{(R,\livegroup)\in \condlivegroup}\left(\square\lozenge R\implies \assumpgrlive(\livegroup)\right).
\end{equation}
\end{definition}

Again, the assumptions employing conditional live group templates will be called conditional live group assumptions. With the generalization of live group assumptions to \emph{conditional} live group assumptions, we actually have all the ingredients to define an \aname for parity games as a conjunction
\begin{equation}\label{eq:parityassumption}
 \assump=  \assumpsafe(\safegroup)\land\assumpdep(\colivegroup)\land\assumpcondlive(\condlivegroup)
\end{equation}
of a safety, a co-liveness, and a conditional live group assumptions. 
Intuitively, we use (i) a safety assumption to prevent $\po$ to leave the cooperative winning region, (ii) a co-live assumption for each winning option that requires seeing a particular \emph{odd} priority only finitely often, and (iii) a conditional live group assumption for each winning option that requires seeing an \emph{even} priority infinitely often if certain \emph{odd} priority have been seen infinitely often. 
 The remainder of this section gives an algorithm (Alg.~\ref{alg:compute parity assumption}) to compute the actual safety, co-live and conditional live group sets $\safegroup$, $\colivegroup$ and $\condlivegroup$, respectively, and proves that the resulting assumption $\assump$ (as in \eqref{eq:parityassumption}) is actually an \aname for the parity game $\game$.

\begin{algorithm}[h!]
	\scriptsize
	\caption{\parityAssump}
	\label{alg:compute parity assumption}
	\begin{algorithmic}[1]
		\Require $ \gamegraph=\tup{V,E},~C:V\rightarrow \{0,1,\ldots\}$
		\Ensure $\assump$
		\State $Z^* \gets \solveParity(\gamegraph,C)$\label{algo:coop parity computation}
		\State $\safegroup\gets  \computeSafe(\gamegraph,Z^*)$
		\State $\gamegraph\gets \gamegraph|_{Z^*}$, $C\gets C|_{Z^*}$\label{algo:restricted game}
		\State $ (\colivegroup,\condlivegroup) \gets $\textsc{ComputeSets}$ ((\gamegraph,C),\emptyset,\emptyset) $
		\State \Return $\safegroup,\colivegroup,\condlivegroup$ 
		
		\Statex
		\Procedure {ComputeSets}{$(\gamegraph,C),\colivegroup,\condlivegroup$}
		\State $d\gets \mathrm{max}\{i \mid C_i \neq \emptyset\}$
			\If {$ d $ is odd}
				\State $ W_{\neg d}\gets \solveParity(\gamegraph|_{V\setminus C_d},C) $ \label{algo:parity without d}
				\State $ \colivegroup\gets \colivegroup \cup \computeCoLive(\gamegraph,W_{\neg d})$ \label{algo:add depressed edges}
			\Else
				\State $ W_{d}\gets \solveBuchi(\gamegraph,C_d) $, $W_{\neg d}\gets V\setminus W_{d}$\label{algo:buechi to reach d}
				\ForAll{odd $ i\in[0;d] $}
					\State $ \condlivegroup\gets \condlivegroup \cup (W_d\cap C_i, \computeLive( \gamegraph|_{W_{d}},C_{i+1}\cup C_{i+3}\cdots \cup C_d)) $\label{algo:add cond live groups}
				\EndFor
			\EndIf
			\If {$d> 0$}
			\State $ \gamegraph\gets \gamegraph|_{W_{\neg d}} $ , $ C_0\gets C_0\cup C_d $, $C_d\gets\emptyset$\label{algo:reduce game to few color}
			\State \textsc{ComputeSets}$ ((\gamegraph,C),\colivegroup,\condlivegroup) $\label{algo:recursively compute}
			\Else 
			\State \Return $(\colivegroup,\condlivegroup)$
			\EndIf
		\EndProcedure
	\end{algorithmic}
\end{algorithm}

\smallskip
\noindent\textbf{Computing \anames.}
The computation of unsafe, co-live, and conditional live group sets $\safegroup$, $\colivegroup$, and $\condlivegroup$ to make $\assump$ in \eqref{eq:parityassumption} an \aname  is formalized in Alg.~\ref{alg:compute parity assumption}. Alg.~\ref{alg:compute parity assumption} utilizes the standard fixpoint algorithm $\solveParity(\gamegraph,C)$~\cite{EmersonJutla:ParityFixpoint} to compute the cooperative winning region for a parity game $\game$, defined as 
\begin{equation}
\textstyle
\solveParity(\gamegraph,C):=\tau X_d \cdots \nu X_2 ~\mu X_1 ~\nu X_0. \bigcup_{i\in[0;d]} (C_i\cap\pre{}{X_i}),
\end{equation}
where $ \tau $ is $ \nu $ if $ d $ is even, and $ \mu $ otherwise.
In addition, Alg.~\ref{alg:compute parity assumption} involves the algorithms $\computeSafe$ (Thm.~\ref{thm:safety assumption}), $\computeLive$ (Thm.~\ref{thm:Buechi assumptions}), and $\computeCoLive$ (Thm.~\ref{thm:coBuechi assumptions}) to compute safety, live group, and co-liveness assumptions in an iterative manner. In addition, $\gamegraph|_U \coloneqq \tup{U,U^0, U^1, E'}$ s.t.\ $U^0 \coloneqq V^0\cap U$, $U^1 \coloneqq V^1\cap U$, and $E'\coloneqq E\cap (U\times U)$ denotes the restriction of a game graph $\gamegraph \coloneqq \tup{V,V^0, V^1, E}$ to a subset of its vertices $U\subseteq V$. Further, $C|_U$ denotes the restriction of the priority set $C$ from $V$ to $U\subseteq V$.

\begin{figure}[t]
\centering
\begin{tikzpicture}
		  		  		\node[player1, blue!60, label={below:$c_1$}] (1) at (0, 0) {$v_1$};
		  		  		\node[player1, label={below:$c_2$}] (2) at (\pos, 0) {$v_2$};
		  				\node[player1, red!60, label={below:$c_3$}] (3) at (-\pos,0) {$v_3$};
		  				\node[player1, label={below:$c_4$}] (4) at (\fpeval{2*\pos},0) {$v_4$};
		  				\node[player1, label={below:$c_5$}] (5) at (\fpeval{3*\pos},0) {$v_5$};
		  				\node[player1, label={below:$c_4$}] (6) at (\fpeval{4*\pos},0) {$v_6$};
		  				\node[player1, label={below:$c_3$}] (7) at (\fpeval{5*\pos},0) {$v_7$};
		  				
		  				\path[->] (1) edge[loop above] () edge[blue!60,bend left = 20] (2);  		
						\path[->] (2) edge[bend left=20] (1) edge[bend left = 40] (3) edge[red!60,bend left = 20] (4);  		
						\path[->] (3) edge (1);
						\path[->] (4) edge[bend left=20] (2);
						\path[->] (5) edge[dashed,loop above] () edge (4) edge[bend left=20] (6);
						\path[->] (6) edge[loop above] () edge[dashed, bend left = 20] (5) edge[densely dotted] (7);
						\path[->] (7) edge[loop right] ();
\end{tikzpicture}
\caption{A parity game, where a vertex with priority $i$ has label $ c_i $. The dotted edges are the unsafe edges, the dashed edges are the co-live edges, and every similarly colored vertex-edge pair forms a conditional live group.}\label{fig:parity_algo_idea}
\vspace{-0.5cm}
\end{figure}
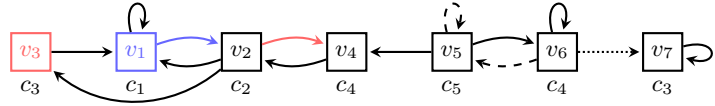

We illustrate the steps of Alg.~\ref{alg:compute parity assumption} by an example depicted in Fig.~\ref{fig:parity_algo_idea}. In line \ref{algo:coop parity computation}, we begin with computing the cooperative winning region $Z^*$ of the entire game, to find that from vertex $ v_7 $, there is no way of satisfying the parity condition even with $ \po $'s cooperation, i.e., $Z^*=\{v_1,\hdots,v_6\}$. So we mark the edge from $ v_6 $ to $ v_7 $ to be a safety-assumption edge, restrict the game to $G=G|_{Z^*}$ and run \textsc{ComputeSets} on the new game.

In the new restricted game $G$ the highest priority is $ d=5 $, which is odd, hence we execute lines \ref{algo:parity without d}-\ref{algo:add depressed edges}. Now a play would be winning only if eventually the play does not see $ v_5 $ any more. Hence, in step \ref{algo:parity without d}, we find the region $W_{\neg 5}=\{v_1,\hdots,v_4,v_6\}$ of the restricted graph $G|_{V\setminus C_5}$ (only containing nodes $v_i$ with priority $C(v_i)<5)$) from where we can satisfy the parity condition without seeing $ v_5 $. We then make sure that we do not leave $W_{\neg 5}$ to visit $v_5$ in the game $G$ infinitely often by executing $\computeCoLive(\gamegraph,W_{\neg 5})$ in line \ref{algo:add depressed edges}.
This puts a co-liveness assumption on the edges $(v_5,v_5)$ and $(v_6,v_5)$. 

Once we restrict a play from visiting $ v_5 $ infinitely often, we only need to focus on satisfying parity without visiting $ v_5 $ within $W_{\neg 5}$. This observation allows us to further restrict our computation to the game $ \game= \game|_{W_{\neg 5}}$ in line \ref{algo:reduce game to few color}, where we also update the priorities to only range from $0$ to $4$. In our example this step does not change anything. We then re-execute \textsc{ComputeSets} on this game. 

In the restricted graph, the highest priority is $ 4 $ which is even, hence we execute lines \ref{algo:buechi to reach d}-\ref{algo:add cond live groups}. One way of winning in this game is to visit $ C_4 $ infinitely often, so we compute the respective cooperative winning region $W_4$ in line \ref{algo:buechi to reach d}. In our example we have $W_4=W_{\neg 5}=\{v_1,\hdots,v_4,v_6\}$. Now, to ensure that from the vertices from which we can cooperatively see $ 4 $, we actually win, we have to make sure that every time a lower odd priority vertex is visited infinitely often, a higher priority is also visited. This can be ensured by conditional live group fairness as computed in line \ref{algo:add cond live groups}. For every odd priority $i<4$, (i.e, for $i=1$ and $i=3$) we have to make sure that either $2$ or $4$ (if $i=1$) or $4$ 
(if $ i=3 $) is visited infinitely often. The resulting live groups $\condlivegroup_i=(R_i,H^\ell_i)$ collect all vertices in $W_4$ with priority $i$ in $R_i$ and all live groups allowing to see even priorities $j$ with $i<j\leq 4$ in $H^\ell_i$, where the latter is computed using the fixed-point algorithm $ \computeLive $ to compute live groups. The resulting live groups for $i=1$ (blue) and $i=3$ (red) are depicted in Fig.~\ref{fig:parity_algo_idea} and given by $(\{v_1\},\{(v_1,v_2)\})$ and $(\{v_3\},\{(v_2,v_4)\}, \{(v_1,v_2)\})$, respectively.

At this point we have $W_{\neg 4}=\emptyset$. With this the game graph computed in line~\ref{algo:reduce game to few color} becomes empty, and the algorithm eventually terminates after iteratively removing all priorities from $C$ after \textsc{ComputeSets} has been run (without any computations, as $\game$ is empty) for priorities $3$, $2$ and $1$. In a different game graph, the reasoning done for priorities $5$ and $4$ above can also repeat for lower priorities if there are other parts of the game graph not contained in $W_{4}$, from where the game can be won by seeing priority $2$ infinitely often. The main insight into the correctness of the outlined algorithm is that all computed assumptions can be conjoined to obtain an \aname for the original parity game.

\smallskip
\noindent\textbf{Main result.}
With Alg.~\ref{alg:compute parity assumption} in place, we can now state the main result of this section, and in particular, of the entire paper, proven in Appendix~\ref{proof:parity assumption is correct}.

\begin{restatable}{theorem}{restateparity}\label{thm:parity assumption}
Let $ \game=\tup{\gamegraph = (V,E), \paritygame(C)} $ be a parity game such that \linebreak$(\safegroup,\colivegroup,\condlivegroup)=\parityAssump(\gamegraph,C)$.
	Then  $\assump=  \assumpsafe(\safegroup)\land\assumpdep(\colivegroup)\land\assumpcondlive(\condlivegroup)$ is an \aname  for $\game$. Moreover, Algo.~\ref{alg:compute parity assumption} terminates in time $ \bigO(n^4) $, where $ n=|V| $. 
\end{restatable}


\section{Experimental Evaluation}\label{section:experiment}
We have developed a C++-based prototype tool {$\tool$}\footnote{Repository URL: \url{https://gitlab.mpi-sws.org/kmallik/simpa}} computing \textbf{S}ufficient, \textbf{Im}plementable and \textbf{P}ermissive \textbf{A}ssumptions for B\"uchi, co-B\"uchi, and parity games. We first compare $\tool$ against the closest related tool  $\krishtool$~\cite{chatterjee2010gist} in Sec.~\ref{sec:toolcompare}. We then show that $\tool$ gives small and meaningful assumptions for the well-known 2-client arbiter synthesis problem from~\cite{NirGR1} in Sec.~\ref{sec:arbiter}. 

\vspace{0.5cm}
\begin{minipage}{\textwidth}
  \begin{minipage}[b]{0.45\textwidth}
   \centering
    \includegraphics[width=1\linewidth]{"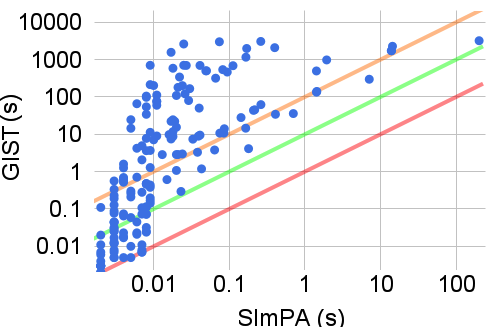"}
    \captionof{figure}{Running times of $\tool$ vs $\krishtool$ (in seconds, log-scale)}\label{fig:experiment}
  \end{minipage}
  \begin{minipage}[b]{0.5\textwidth}
  \centering
   \begin{tabular}{|l|r|r|}
   \hline
    & \multicolumn{1}{c|}{$\tool$} & \multicolumn{1}{c|}{$\krishtool$} \\ \hline
   Mean-time & 64.8s & \cellcolor[HTML]{FFFFFF}1079.0s \\ \hline
   \begin{tabular}[c]{@{}l@{}}Non-timeout \\ mean-time\end{tabular} & 64.8s & 209.2s \\ \hline
   Timeouts (1hr) & 0(0\%) & 59(26\%) \\ \hline
   \begin{tabular}[c]{@{}l@{}}No assumption \\ generated\end{tabular} & 0(0\%) & 20(9\%) \\ \hline
   Faster & 230(100\%) & 0(0\%) \\ \hline
   
   \end{tabular}
      \captionof{table}{Summary of the experimental results}\label{tab:benchmarks}
    \end{minipage}
  \end{minipage}
\subsection{Performance Evaluation}\label{sec:toolcompare}
We compare the effectiveness of our tool against a re-implementation of the closest related tool called $\krishtool$~\cite{chatterjee2010gist}, which is not available anymore from the authors\footnote{The link provided in the paper is broken, and the authors informed us that the implementation is not available.}. $\krishtool$ originally computes assumptions only enabling a particular initial vertex to become winning for $\pz$. However, for the experiments, we run $\krishtool$ until one of the cooperatively winning vertices is not winning anymore. Since $ \krishtool $ starts with a maximal assumption and keeps making it smaller until a fixed initial vertex is not winning anymore, our modification makes $ \krishtool $ faster as the modified termination condition is satisfied earlier. As our tool does not depend on any fixed initial vertex and the dependency on the initial vertex makes $ \krishtool $ slower, this modification allows a fair comparison.

We compared the performance and the quality of the assumptions computed by $\tool$ and $\krishtool$ on a set of parity games collected from the SYNTCOMP benchmark suite~\cite{syntcomp}.
For computing assumptions using both $\tool $ and $\krishtool$, we set a timeout of one hour.
All the experiments were performed on a computer equipped with Intel(R) Core(TM) i5-10600T CPU @ 2.40GHz and 32 GiB RAM.

We provide all details of the experimental results in Table~\ref{appendix:extended_table} in the appendix and summarize them in Table~\ref{tab:benchmarks}. In addition, Fig.~\ref{fig:experiment} shows a scatter plot, where every instance of the benchmarks is depicted as a point, where the X and the Y coordinates represent the running time for $\tool$ and $\krishtool$ (in seconds), respectively.
We see that $\tool$ is computationally much faster than $\krishtool$ in every instance (all dots lie above the lower red line) -- most times by one (above the middle green line) and many times even two (above the upper orange line) orders of magnitude.

Moreover, in some experiments, $\krishtool$ fails to compute a sufficient assumption (in the sense of Def.~\ref{def:sufficient}), whereas our algorithm successfully computes an \aname
(given in the row labeled `no assumption generated' in Table~\ref{tab:benchmarks} and marked using `*' next to the computation times in Table~\ref{appendix:extended_table}). This is not surprising, as the class of assumptions used by $\krishtool$ are only unsafe edges and live edges (i.e., singleton live groups) which are not expressive enough to provide sufficient assumptions for all parity games (see Fig.~\ref{fig:illustration_intro}(b) for a simple example where there is no sufficient assumption that can be expressed using live edges). Furthermore, we note that in all cases where the assumptions computed by $\krishtool$ are actually \anames, $\tool$ computes the same assumptions orders or magnitudes faster.

\subsection{2-Client Arbiter Example}\label{sec:arbiter}
We consider the 2-client arbiter example from~\cite{NirGR1}. In this example, clients $i\in\{1,2\}$ ($\po$) can request or free a shared resource by setting the input variables $r_i$ to true or false, and the arbiter ($\pz$) can set the output variables $g_i$ to true or false to grant or withdraw the shared resource to/from client $i$. The game graph for this example is implicitly given as part of the specification (as this is a GR(1) synthesis problem, see~\cite{NirGR1} for details). We depict a relevant part of this game graph schematically in Fig.~\ref{fig:arbiter_ex}. 
Here, rectangles and circles represent $\po$ and $\pz$ vertices, respectively, and the double-lined vertices have priority $2$ (are Büchi vertices), while all other vertices have priority $1$. The labels of the $\pz$ states indicate the current status of the request and grant bits, and in addition, remember if a request is currently \emph{pending}, i.e., $F_i = \overline{g_i}~\mathsf{S}~ (r_i\wedge \overline{g_i})$, where $\mathsf{S}$ denotes the LTL operator \enquote{since}. Labels of $\po$ vertices additionally remember the last move chosen by $\pz$. We see that all vertices with no pending requests have priority $2$. It is known that there does not exist a winning strategy in this game for $\pz$ if the moves of the clients ($\po$) are unconstrained.

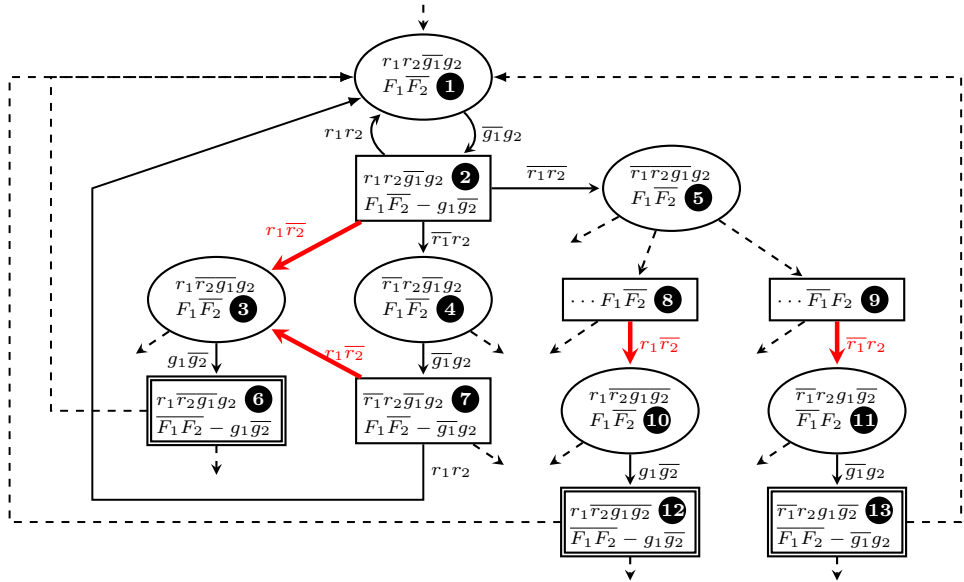
\begin{figure}
\small
\centering
\vspace{-0.5cm}
\begin{scriptsize}
\begin{tikzpicture}[yscale=0.8,xscale=1.1]
		  		  		\node[bplayer0] (1) at (0, 0) {$r_1r_2\overline{g_1}g_2\newline F_1\overline{F_2} \circled{1}$};
		  		  		\node[bplayer1] (2) at (0,-\ypos) {$r_1r_2\overline{g_1}g_2\circled{2}\newline F_1\overline{F_2} - g_1\overline{g_2}$};
		  				\node[bplayer0] (3) at (-\hpos,-2*\ypos) {$r_1\overline{r_2}\overline{g_1}g_2\newline F_1\overline{F_2}\circled{3}$};
		  				\node[bplayer0] (4) at (0,-2*\ypos) {$\overline{r_1}r_2\overline{g_1}g_2\newline F_1\overline{F_2}\circled{4}$};
		  				\node[bplayer0] (5) at (1.2*\hpos,-1*\ypos) {$\overline{r_1}\overline{r_2}\overline{g_1}g_2\newline F_1\overline{F_2}\circled{5}$};
		  				\node[bplayer1,accepting] (6) at (-1*\hpos,-3*\ypos) {$r_1\overline{r_2}\overline{g_1}g_2\circled{6}\newline \overline{F_1}\overline{F_2} - g_1\overline{g_2}$};
		  				\node[bplayer1] (7) at (0,-3*\ypos) {$\overline{r_1}r_2\overline{g_1}g_2\circled{7}\newline F_1\overline{F_2} - \overline{g_1}g_2$};
		  				
		  				\node[bplayer1] (8) at (1*\hpos,-2*\ypos) {$\cdots F_1\overline{F_2}\circled{8}$};
		  				\node[bplayer1] (9) at (2*\hpos,-2*\ypos) {$\cdots \overline{F_1}F_2\circled{9}$};
		  				\node[bplayer0] (10) at (1*\hpos,-3*\ypos) {$r_1\overline{r_2}\overline{g_1}\overline{g_2} \newline F_1\overline{F_2}\circled{10}$};
		  				\node[bplayer0] (11) at (2*\hpos,-3*\ypos) {$\overline{r_1}r_2g_1\overline{g_2}\newline \overline{F_1}F_2\circled{11} $};
		  				\node[bplayer1,accepting] (12) at (1*\hpos,-4*\ypos) {$r_1\overline{r_2}\overline{g_1}\overline{g_2}\circled{12} \newline \overline{F_1}\overline{F_2} - g_1\overline{g_2}$};
		  				\node[bplayer1,accepting] (13) at (2*\hpos,-4*\ypos) {$\overline{r_1}r_2g_1\overline{g_2}\circled{13} \newline \overline{F_1}\overline{F_2} - \overline{g_1}g_2$};
		  				
		  				\path[->] (0,0.65*\ypos) edge[dashed] (1); 
		  				\path[->] (1) edge[bend left =40] node[right] {$\overline{g_1}g_2$} (2); 
						\path[->] (2) edge[bend left =40] node[left] {$r_1r_2$} (1) edge[above left,pos=0.5,red, line width = 0.06cm] node {$r_1\overline{r_2}$} (3) edge node {$\overline{r_1}r_2$} (4) edge node {$\overline{r_1}\overline{r_2}$} (5); 
						\path[->] (3) edge[left] node {$g_1\overline{g_2}$} (6) edge[dashed] (-1.4*\hpos, -2.5*\ypos); 
						\path[->] (4) edge node {$\overline{g_1}g_2$} (7) edge[dashed] (0.4*\hpos, -2.5*\ypos); 
						\path[->] (5) edge[dashed]  (8) edge[dashed]  (9) edge[dashed] (0.7*\hpos, -1.5*\ypos); 
						\path[->] (6) edge[dashed] (-1*\hpos,-3.6*\ypos); 
						\draw[dashed,->,>=latex] (6) -- (-1.8*\hpos,-3*\ypos) -| (-1.8*\hpos,0) -> (1);
						\path[->] (7) edge[red, line width = 0.06cm] node[right] {$r_1\overline{r_2}$} (3) edge[dashed] (0.4*\hpos,-3.5*\ypos); 
						\draw[->,>=latex] (7) -- node[right] {$r_1r_2$} (0*\hpos,-3.8*\ypos) -| (-1.6*\hpos,-3.8*\ypos) -| (-1.6*\hpos,-1*\ypos) -> (1);
						\path[->] (8) edge[red, line width = 0.06cm] node {$r_1\overline{r_2}$} (10) edge[dashed] (0.6*\hpos,-2.5*\ypos); 
						\path[->] (9) edge[red, line width = 0.06cm] node {$\overline{r_1}r_2$} (11) edge[dashed] (1.6*\hpos,-2.5*\ypos); 
						\path[->] (10) edge node {$g_1\overline{g_2}$} (12) edge[dashed] (0.6*\hpos, -3.5*\ypos); 
						\path[->] (11) edge node {$\overline{g_1}g_2$} (13) edge[dashed] (1.6*\hpos, -3.5*\ypos); 
						\path[->] (12) edge[dashed] (1*\hpos,-4.55*\ypos); 
						\draw[dashed,->,>=latex] (12) -- (-2*\hpos,-4*\ypos) -| (-2*\hpos,0) -> (1);
						\path[->] (13) edge[dashed] (2*\hpos,-4.55*\ypos); 
						\draw[dashed,->,>=latex] (13) -- (2.6*\hpos,-4*\ypos) -| (2.6*\hpos,0) -> (1);

\end{tikzpicture}
\end{scriptsize}
\vspace{-0.5cm}
\caption{Illustration of a relevant part of the game graph for the 2-client arbiter.}\label{fig:arbiter_ex}
\vspace{-0.5cm}
\end{figure}

Running $\tool$ on this example yields only live group assumptions (as this is a Büchi game and all vertices are cooperatively winning) which were computed in $0.01$ seconds.  The edges of one live group are indicated schematically by thick red arrows in Fig.~\ref{fig:arbiter_ex}.  We see that this live group ensures that the play eventually moves to vertices where the $\pz$ can force a visit to a Büchi vertex. 
In~\cite{NirGR1}, the assumption used to restrict the clients' behavior in order to render the synthesis problem realizable is given by 
\begin{equation*}
 \tilde{\assump} = \bigwedge_i (\overline{r_i} \wedge \square\big((r_i\neq g_i)\Rightarrow(r_i=\bigcirc r_i)\big)\wedge\square\big((r_i\wedge g_i)\Rightarrow \lozenge\overline{r_i})\big).
\end{equation*}
We see that our live group assumptions are similar but more permissive. For example when we persistently see states with label $r_1\overline{g_1}$ and $r_2g_2$ (e.g., cycling through states $\circled{1}$ and $\circled{2}$ in Fig.~\ref{fig:arbiter_ex}) we enforce that \emph{eventually} the edge labeled $r_1\overline{r_2}$ needs to be taken. On the other hand, the second and third condition in $\tilde{\assump}$ (which are triggered by $r_1\overline{g_1}$ and $r_2g_2$, respectively) enforces that no other outgoing transition is allowed from state $\circled{2}$ except for the one labeled with $r_1\overline{r_2}$, which is strictly more restrictive. 

We have also run $\krishtool$ on this example. It took $6.44$ seconds to compute live edge assumptions for unrestricted initial conditions, which is two orders of magnitude slower than $\tool$. Further, in order to see $\circled{6}$ infinitely often $\krishtool$ returns the live edges $\circled{2}-\circled{3}$ and $\circled{7}-\circled{1}$. This assumption is not permissive, as there exist winning plays that do not use either of these edges infinitely often. 
It turns out that an \aname  for this example will unavoidably require live groups -- singleton live edges, as computed by $\krishtool$, will not suffice.

\subsubsection*{Acknowledgements}
S.~P.~Nayak and A.-K.~Schmuck are partially supported by the DFG project 389792660 TRR 248–CPEC. 
A.~Anand and A.-K.~Schmuck are partially supported by the DFG project SCHM 3541/1-1.
K.~Mallik is supported by the ERC project ERC-2020-AdG 101020093.

\newpage
\bibliographystyle{splncs04}
\bibliography{main}

\newpage
\appendix

\section{\aname assumptions for safety games}\label{proof:safety assumption is correct}
For the convenience of the reader, we restate Thm. \ref{thm:safety assumption} here.
\restatesafety*

\begin{proof}
	We refer the reader to \cite[chapter 2]{AutomataLogicsandInfiniteGames} for the proof of $ Z^*=\team{0,1}\square U $, and we show below that $ \assumpsafe(\safegroup) $ is a \aname assumption for safety games by proving sufficiency, implementability and permissiveness below.

	\begin{inparaitem}[$\blacktriangleright$]
		\noindent \item \textbf{(ii)} Implementability: 
	It is easy to observe that $ \assumpsafe(\safegroup) $ is implementable by $ \po $, if he does not take the edges in $ \safegroup $ ever.

		\noindent \item \textbf{(ii)} Sufficiency: 
	For sufficiency of the assumption, consider the strategy $ \stratz $ for $ \pz $: at a vertex $ v\in Z^* $, plays the transition that keeps the play in $ Z^* $, and for other vertices, plays arbitrarily. The strategy is well-defined, since if at $ v\in Z^* $, there is no such transition, $ v $ would not be in $ Z^* $, by definition.
	
	Let $ v_0\in Z^*=\team{0,1}\square U $. Let $ \strato $ be an arbitrary strategy of $ \po $ such that $ \lang(\strato)\subseteq \lang(\assumpsafe(\safegroup)) $, and $ \rho=v_0v_1\ldots $ be an arbitrary $ \stratz\strato $-play. Then $ \rho\in \lang(\assumpsafe(\safegroup)) $.
	
	Now, suppose $ \rho \not\in \lang(\square U) $, i.e. $ v_i \not\in U $ for some $ i $. WLOG assume that $ i $ is the least such index, that is, for all $ j<i $, $ v_j\in U $. Then by definition of $ \safegroup $, $ (v_{i-1}, v_i)\in S $, which is a contradiction. Hence, $ \rho \in \lang(\square U) $.

	\noindent \item \textbf{(ii)} Permissiveness: 
	Now for the permissiveness, let $ \rho\in \lang{(\square U)} $. Suppose that $ \rho\not\in\lang{(\assumpsafe(\safegroup))} $. Then some edge $ (v,v')\in S $ is taken in $ \rho $. Then after reaching $ v' $, $ \rho $ still satisfies the safety condition. Hence, by Prop. \ref{prop:buechi cooperative winning region}, $ v'\in Z^* $, but then $ (v,v')\not\in S $, which is a contradiction. Hence, $ \rho\in\lang{(\assumpsafe(\safegroup))}  $.

	\noindent \item \textit{Complexity analysis.}
	The fixpoint computation takes time $ \bigO(n^2) $ and then computing $ \safegroup $ takes another $ \bigO(n^2) $. 
\end{inparaitem}
\end{proof}

\section{$ \mathsf{tpre} $ accelerates the $ \mu $-fixpoint computation  }
\newcommand{\Xt}{X^{\mathsf{t}}}
In sec. \ref{section:BuchiGames}, we mentioned that Prop.~\ref{prop:buechi cooperative winning region} follows from the correctness proof of \eqref{eq:EnvironmentBuchiFixpoint} (given in \cite{ChatterjeeHenzingerPiterman:AlgoForBuchi}) due to the following Lemma, which we prove here.
\begin{lemma}\label{lemma:tpre is accelerated pre}
	For $ U\subseteq V $, $ \mu X. ~U\cup \pre{}{X}=\mu X. ~U\cup \tpre{}{X}  $.
\end{lemma}
\begin{proof}
	Let $ \mu X. ~U\cup \pre{}{X} $ gives the following computation of $ X $, which we refer to as the \emph{pre computation},
	\begin{table}[H]\centering
		\begin{tabular}{cccccccc}
			$ X_0$ & ${\color{gray!50}~\subsetneq}  $&$ X_1$ & ${\color{gray!50}~\subsetneq}  $&$ \quad\cdots\quad $&$ X_{{\color{red!50}k}}$ & $\color{red!50}{=} $&$ X_{k+1} $\\
			
			$ {\protect\rotatebox{90}{$=$}}  $&&$ {\protect\rotatebox{90}{$=$}}  $&&&$ {\protect\rotatebox{90}{$=$}} $&&{\protect\rotatebox{90}{$=$}}\\
			
			$ \emptyset  $&&$ U\cup \pre{}{X_0}  $&&$ \cdots $&$ X_{k-1}\cup \pre{}{X_{k-1}} $&&$ X_{k}\cup \pre{}{X_{k}} $\\
		\end{tabular}
	\end{table}

\noindent
	and similarly let $ \mu X. ~U\cup \tpre{}{X} $ gives the following computation of $ X $, which we refer to as the \emph{tpre computation},

\begin{table}[H]\centering
\begin{tabular}{cccccccc}
	$ \Xt_0$ & ${\color{gray!50}~\subsetneq}  $&$ \Xt_1$ & ${\color{gray!50}~\subsetneq}  $&$ \quad\cdots\quad $&$ \Xt_{{\color{red!50}l}}$ & $\color{red!50}{=} $&$ \Xt_{l+1} $\\
	
	$ {\protect\rotatebox{90}{$=$}}  $&&$ {\protect\rotatebox{90}{$=$}}  $&&&$ {\protect\rotatebox{90}{$=$}} $&&{\protect\rotatebox{90}{$=$}}\\
	
	$ \emptyset  $&&$ U\cup \tpre{}{\Xt_0}  $&&$ \cdots $&$ \Xt_{l-1}\cup \tpre{}{\Xt_{l-1}} $&&$ \Xt_{l}\cup \tpre{}{\Xt_{l}} $\\
\end{tabular}
\end{table}

We first observe that $ \mathsf{tpre} $ of a set contains the $ \mathsf{pre} $ of the set.
\begin{claim}
	For $ U\subset V $, $ \pre{}{U}\subseteq \tpre{}{U} $.
\end{claim}
\begin{proof}
	We know $ \cprez{}{U}\subseteq \attrz{}{U} $ by the definition of $ \mathsf{attr} $, and $ \cpreo{}{U}\subseteq \cpreo{}{\attrz{}{U}\cup U} $ by the monotonicity of $ \mathsf{cpre}^0 $. Then $ \pre{}{U}=\cprez{}{U}\cup \cpreo{}{U} \subseteq \attrz{}{U} \cup \cpreo{}{\attrz{}{U}\cup U} = \tpre{}{U} $. \tqed
\end{proof}

Now we show that every vertex that appears in the $ i $-th iteration of the pre computation, also appears in the $ i $-th iteration of tpre computation.

\begin{claim}
	$ \forall i\in \N, X_i\subseteq \Xt_i $.
\end{claim} 
\begin{proof}
	We prove the claim by induction on $ i $. For the base case, when $ i=0 $, the statement is trivially true. For induction hypothesis (IH), assume that the statement holds for some $ i\in \N $.
	\begin{align}
		\Xt_{i+1}&= \tpre{}{\Xt_i} \cup \Xt_i &\\
		&\supseteq\tpre{}{X_i}\cup X_i ,& \text{ by IH and monotonicity of }\mathsf{tpre}\\
		&\supseteq\pre{}{X_i}\cup X_i ,& \text{ by the claim above} \\
		&= X_{i+1},&
	\end{align}	
	Hence, by induction, the statement holds for any $ i\in \N $.\tqed
\end{proof}

The claim show one direction of the lemma, that is $  \mu X. ~U\cup \pre{}{X}\subseteq \mu X. ~U\cup \tpre{}{X} $. Both the claims above also give that $ l\leq k $, where $ l $ and $ k $ are the terminating step of tpre and pre computations respectively. For the other direction, we show that every vertex that appears in the $ i $-th iteration of tpre computation also eventually appears in the pre computation.

\begin{claim}
	$ \forall i\in \N, \exists j \geq i $ such that $ \Xt_i\subseteq X_j $.
\end{claim}
\begin{proof}
	We again prove the claim by induction on $ i $. For base case, again the statement holds trivially with $ j=0 $. Then for induction hypothesis (IH), assume that the statement holds for some $ i\in\N $, with some $ j\geq i $, that is $ \Xt_i\subseteq X_j $.
	
	Let $ v\in \Xt_{i+1}\setminus \Xt_i $ be an arbitrary vertex. Then $ v\in \attrz{}{\Xt_i} $ or $ v\in \cpreo{}{\attrz{}{\Xt_i}\cup \Xt_i} $.
	
	In the earlier case, when $ v\in \attrz{}{\Xt_i} $, we have
	\begin{align}
		v&\in \cpre{}{\Xt_i}{0,t} ,& \text{ for some }p\in \N\\
		&\subseteq \cpre{}{X_i}{0,p},& \text{ by IH}\\
		&\subseteq\mathsf{pre}^{0,p}{(X_i)},& \\
		&\subseteq X_{p+i}.&
	\end{align}	
	Then since $ \attrz{}{\Xt_i} $ terminates in at most $ |V|=n $ many iterations, that is $ \cup_{i\geq 1}^{n}\cpre{\gamegraph}{\Xt_i}{a,i}=\attrz{}{\Xt_i} $, we have that $ \attrz{}{\Xt_i}\subseteq X_{n+i} $.
	While in the later case, when $ v\in \cpreo{}{\attrz{}{\Xt_i}\cup \Xt_i} $, we have
	\begin{align}
		v&\in \cpreo{}{\attrz{}{\Xt_i}\cup \Xt_i} ,& \text{}\\
		&\subseteq \cpreo{}{X_{n+i}\cup X_j},& \text{ by IH and discussion above}\\
		&\subseteq\pre{}{X_{max\{n+i,j\}}},& \\
		&\subseteq X_{max\{n+i,j\}+1}. &
	\end{align}	
	Since $ i\leq p+i\leq max\{n+i,j\}+1 $, $ v\in X_{max\{n+i,j\}+1} $. Then by induction, the claim holds true.\tqed
\end{proof}
This claim shows that $ X_k=\Xt_k $, since $ l\leq k $ and $ \Xt_k=\Xt_l $. Hence the lemma is proved.
\end{proof}

\section{\aname assumptions for \buchi games}\label{proof:buechi assumption is correct}

For the convenience of the reader, we restate Thm. \ref{thm:Buechi assumptions} here.
\restatebuchi*

\begin{proof}
	
	Since $ \front(X^i)\subseteq V^1 $, we observe that every vertex in $ Z^*\cap V^0 $ is added to the least fixpoint computation of $ X $ in the $ \textsf{attr}^0 $ part of $ \textsf{tpre} $. Then $ V^0\cap Z^* $ can be partitioned into sets $ V_1, V_2\ldots, V_p $, where $ V_i= (X^i\setminus X^{i-1})\cap V^0 $. With this observation we prove sufficiency, implementability and permissiveness below and finally comment on the complexity of $ \computeSafe$.

	\begin{inparaitem}[$\blacktriangleright$]
		\noindent \item Implementability: 
		Since the source of live groups is a subset of $ \po $'s vertices, the assumption is easily implementable if $ \po $ plays one of these live group edges infinitely often, when the sources are visited infinitely often, and $ \pz $ can not falsify it.
	
		\noindent \item Sufficiency: 	
	Consider the strategy $ \stratz $ for $ \pz $: at a vertex $ v\in V_i $, she plays the $ \textsf{attr}^0 $ strategy to reach $ X^{i-1} $, and for other vertices, she plays arbitrarily. We show that $ \stratz $ is winning under assumption $ \assump $ for all vertices in the cooperative winning region.
	
	Let $ v_0\in Z^*=\team{0,1}\square\lozenge U $ (from Prop. \ref{prop:buechi cooperative winning region}). Let $ \strato $ be an arbitrary strategy of $ \po $ such that $ \lang(\strato)\subseteq \lang(\assump) $, and $ \rho=v_0v_1\ldots $ be an arbitrary $ \stratz\strato $-play. Then $ \rho\in \lang(\assump) $. It remains to show that 
	$ \rho \in \lang(\spec) $.
	
	Suppose $ \rho \not\in \lang(\spec) $, i.e. $ inf(\rho)\cap U=\emptyset $. Let $ Y $ saturate after $ k $ iterations, that is $ Y^{k-1}\supsetneq Y^k=Y^{k+1} $. 
	Note that by standard fixed-point computations (see \cite{ChatterjeeHenzingerPiterman:AlgoForBuchi} for more details),
	we have
	\begin{align}
		X^1&= \left( U \cap \tpre{}{Y^k} \right)= \left( U \cap \tpre{}{Z^*} \right). \label{eq:x1isU}
	\end{align}
	Then $ \exists ~i\in [2,p] $, such that $ inf(\rho)\cap X^i\not= \emptyset $ but $ inf(\rho)\cap X^j= \emptyset,~ \forall j<i $, since if there is no such $ i $, $ inf(\rho)\cap X^1\not= \emptyset $, and $ inf(\rho)\cap U\not=\emptyset $, by eq. \eqref{eq:x1isU}, contradicting our assumption.

	By the definition of $ \pi^0 $ and $ \textsf{attr}^0 $, every time $ \rho $ visits $ X^i\setminus \front(X^{i-1}) $, it also eventually visits $ X^{i-1} $. If $ \rho $ visits $ \front(X^{i-1}) $ infintely often, by the definition of $ \livegroup $, an edge $ (v,v')\in (\front(X^{i-1})\times (X^{i}\setminus \front(X^{i-1})))\cap E $ is taken infinitely often, where $ v'\in (X^{i}\setminus \front(X^{i-1})) = \attrz{}{X^{i-1}}\cup X^{i-1} $, and hence again $ X^{i-1} $ is visited infinitely often, contradicting our assumption. Hence, $ \rho\in\lang{(\spec)} $, and $ v_0\in\team{0}_\assump\spec $.

		\noindent \item Permissiveness: 
	Now for the permissiveness of the assumption, let $ \rho\in \lang{(\spec)} $. Suppose that $ \rho\not\in\lang{(\assump)} $. 
	
	Case 1: If $ \rho\not\in\lang{(\assumpsafe(\safegroup))} $, then some edge $ (v,v')\in \safegroup $ is taken in $ \rho $. Then after reaching $ v' $, $ \rho $ still satisfies the \buchi condition. Hence, by Prop. \ref{prop:buechi cooperative winning region}, $ v'\in Z^* $, but then $ (v,v')\not\in S $, which is a contradiction.
	
	Case 2: If $ \rho\not\in\lang(\assumpgrlive(\livegroup)) $, then $ \exists \livegroupSingle\in \livegroup, $ such that $ \rho $ visits $ src(\livegroupSingle)=\front(X^i) $ infinitely often, but no edge in $ \livegroupSingle $ is taken infinitely often. Then since the edges in $ \livegroupSingle $ lead to $  (X^{i+1}\setminus \front(X^i)) $, the play must  stay in either $ \front(X^i) $ or goes to $ X^{j}\setminus X^i $ for some $ j>i+1 $. In the first case, since $ U\cap Z^*\subseteq X^1 $, $ \rho\not\in\spec $, which would be a contradiction. On the other hand, in the second case, after going to $ X^j $, $ \rho $ has an edge going from some $ v\in Z^*\backslash X^{j-1} $ to some $ v'\in X^i $ (else $ U\cap Z^*\subseteq X^1\subseteq X^i $ can not be reached). But then $ v $ would be added to $ X^{i+1} $, which contradicts to the fact that $ j>i+1 $.
	In either case, we get a contradiction, so $ \rho\in\lang(\assump)$.

	\noindent \item \textit{Complexity analysis.}
	The computation of live groups takes $ \bigO(n^2) $ time to compute the inner least fixpoint variable $ X $ and there will be at most $ n $ such computations. While the inner fixpoint is being computed, in the last iteration of $ Y $, with additive overhead of $ \bigO(n^2) $, the live groups can be computed. Then the total computation time is $ \bigO(n^3) $.	
\end{inparaitem}
\end{proof}

\subsection{Faster algorithm for \buchi games}\label{sec:fasterbuchi}

\begin{algorithm}[H]
	\caption{\computeLive}
	\label{alg:compute buechi assumption}
		\begin{algorithmic}[1]
			\Require $ \gamegraph=\tup{V=\vertexz\cup \vertexo, E},$ \buchi objective $\spec=\square\lozenge I, $ for  $I\subseteq V$
			\Ensure Assumption $\assump$ on $ \p{1} $
			\State $Z^* \gets \solveBuchi_{0,1}(\gamegraph,I)$\label{algo:coop buechi computation}
			\State $(\safegroup^i, \safegroup^j)\gets  \computeSafe(\gamegraph,Z^*)$
			\State $\gamegraph\gets \gamegraph|_{Z^*}, I \gets I\cap Z^*$\label{algo:restricted buechi game}\Comment{All vertices are cooperatively \buchi winning}
			\State $ \livegroup \gets $\textsc{ComputeLiveGroups}$ ((\gamegraph,I),\emptyset) $
			\State \Return $(\safegroup,\livegroup)$ 
			
			\Statex
			\Procedure {ComputeLiveGroups}{$(\gamegraph,I),\livegroup$}
			\State $ U\gets I $
			\While{$ U\not=V $}
			\State $ W_{attr}\gets \attr{\gamegraph}{\gamegraph, U}{0} $\label{algo:buechi:compute attr}
			\State $ U\gets U\cup W_{attr} $\label{algo:buechi:add attr}
			\State $ C\gets \cpre{\gamegraph}{U}{1}  \setminus U $\label{algo:buechi:compute front}
			\State $ \livegroup\gets \livegroup \cup \{\{ (u,v)\in E\cap (C\times U) \}\} $\label{algo:buechi:add assump live group}
			\State $ U\gets U\cup C $\label{algo:buechi:add front}
			\EndWhile
			\State \Return $ (\livegroup) $
			\EndProcedure
		\end{algorithmic}
\end{algorithm}

\begin{theorem}
	Given a game graph $ \gamegraph=\tup{V=\vertexz\cup \vertexo, E} $, with \buchi winning condition $ \spec=\square\lozenge I $ for $ \p{0} $. Then Alg. \ref{alg:compute buechi assumption} terminates in time $ \bigO(m+n)$, and $ \assump = \assumpsafe (S) \wedge \assumpgrlive(\livegroup) $ is an adequately permissive assumption on $ \p{1} $ from $ Z^* $. Here, $n=|V|$ and $m=|E|$.
\end{theorem}
\begin{proof}
	We first show that the algorithm terminates. We show that the procedure $ \textsc{ComputeLiveGroups} $ terminates. Since in step \ref{algo:restricted buechi game}, the game graph is restricted to cooperative \buchi winning region $ Z^* $, we need to show that in the procedure, $ U=V=Z^* $ eventually. Let $ U_l $ be the value of $ U $ after the $ l$-th iteration of $ \textsc{ComputeLiveGroups} ((\gamegraph,I),\emptyset,\emptyset)  $, with $ U^0=I $. Since vertices are only added to $ U $ (and never removed) and there are only finitely many vertices, 
    $ U_0\subset U_1\subset\ldots\subset U_{m}=U_{m+1} $
	for some $ m\in\N $. 
		
	Since the other direction is trivial, we show that $ Z^*\subseteq U_m $. Suppose this is not the case, i.e. $ v\in Z^*\backslash U_m $. Since $ v\in Z^* $, both players cooperately can visit $ I $ from $ v $. Then there is a finite path $ \playprefix = v_0v_1\cdots v_k $ for $ v_k\in I $ and $v_0=v$. But since $ I=U_0\subseteq U_m $, but $ v\not \in U_m $, $ \playprefix $ enters $ U_m $ eventually. Let $ l $ be the highest index such that $ v_l\not\in U_m $ but $ v_{l+1}\in U_m $.
	
	Then if $ v_l\in V^0 $, it would be added to $ U $ in step \ref{algo:buechi:add attr} of $ (m+1) $-th iteration, i.e. $ U_m\not=U_{m+1} $. Else if $ v_l\in V^1 $, it would be added to $ U $ in step \ref{algo:buechi:add front} of $ (m+1) $-th iteration since $ v_{l+1}\in U_m $, i.e. $ U_m\not=U_{m+1} $. In either case, we get a contradiction. Hence, $ v\in U_m $, implying $ Z^*=U_m $.  Hence the procedure \textsc{ComputeLiveGroups}, and hence the Algo. \ref{alg:compute buechi assumption}, terminates.
	
	We now show that the assumption obtained is adequately permissive. 
	
	\begin{inparaitem}[$\blacktriangleright$]
		\noindent \item Implementability: We note that in step \ref{algo:buechi:compute front}, $ C\subseteq V^1 $: since if $ v\in V^0\cap C $, then there is an edge from $ v $ to $ U $, and hence $ v\in U $ already by steps \ref{algo:buechi:compute attr} and \ref{algo:buechi:add attr}. Then since the source of live groups (which are only added in step \ref{algo:buechi:add assump live group}) is a subset of $ \p{1} $'s vertices, the assumption is easily implementable if $ \p{1} $ plays one of these live group edges infinitely often, when the sources are visited infinitely often, and $ \p{0} $ can not falsify it.
		
		\noindent \item Sufficiency: Again, let $ U_l $ and $ m $ be as defined earlier. Define $ X_l\coloneqq U_l\backslash U_{l-1} $ for $ 1\leq l\leq m $, and $ X_0=U_0=I $. Then every vertex $ v\in Z^* $ is in $ X_l $ for some $ l\in [0;m] $.
		
		Consider the strategy $ \stratz$ for $ \p{0} $: at a vertex $ v\in V_0\cap X_l $, she plays the $ \textsf{attr}^0 $ strategy to reach $ U_{l-1} $, and for other vertices, she plays arbitrarily. We show that $ \stratz $ is winning under assumption $ \assump $ for $ \p{0} $ from all vertices in the cooperative winning region $ Z^* $.
		
		Let $ v_0\in Z^*=\team{0,1}\square\lozenge I $ (from Step \ref{algo:coop buechi computation}). Let $ \strato $ be an arbitrary strategy of $ \p{1} $ such that $ \lang(\strato)\subseteq \lang(\assump) $, and $ \rho=v_0v_1\ldots $ be an arbitrary $ \stratz\strato $-play. Then $ \rho\in \lang(\assump) $. It remains to show that 
		$ \rho \in \lang(\spec) $.
		
		Suppose $ \rho \not\in \lang(\spec) $, i.e. $ inf(\play)\cap I=\emptyset $. Note that $ \play $ never leaves $ Z^* $ due to safety assumption template. Then consider the set $ R $ of vertices which occur infinitely often in $ \rho $. Let $ 0\leq k\leq m $ be the least index such that $ R\cap X_{k} \not=\emptyset$. From the assumption, $ k>0 $. Let $ v\in R\cap X_k $.
		
		If $ v\in V^0 $, by the definition of $ \stratz $, every time $ \play $ reaches $ v $, it must reach $ U_{k-1} $, contradicting the minimality of $ k $. Else if $ v\in V^1 $, then by the definition of $ \livegroup $, infinitely often reaching $ v $ implies infinitely often reaching $ \attr{0}{U_{k-1}}{} $. But again the play visits $ U_{k-1} $ by arguments above, giving the contradiction.
		
		In any case, we get a contradiction, implying that the assumption is wrong. Hence, $ \play \in \lang(\spec)$, and $ v_0\in\team{0}_\assump\spec  $.

		\noindent \item Permissiveness: 
		Now for the permissiveness of the assumption, let $ \rho\in \lang{(\spec)} $. Suppose that $ \rho\not\in\lang{(\assump)} $. 
		
		Case 1: If $ \rho\not\in\lang{(\assumpsafe(\safegroup))} $, then some edge $ (v,v')\in \safegroup $ is taken in $ \rho $. Then after reaching $ v' $, $ \rho $ still satisfies the \buchi condition. Hence, $ v'\in Z^*=\team{0,1}\square\lozenge I $, but then $ (v,v')\not\in S $, which is a contradiction.
		
		Case 2: If $ \rho\not\in\lang(\assumpgrlive(\livegroup)) $, then $ \exists \livegroupSingle\in \livegroup, $ such that $ \rho $ visits $ src(\livegroupSingle)=C_l $ (for the value of C after $ l $-th iteration) infinitely often, but no edge in $ \livegroupSingle $ is taken infinitely often. Then since the edges in $ \livegroupSingle $ lead to $  \attr{}{U_{l-1}}{0} $, the play must  stay in either $ C_l $ or goes to $ U^{k}\setminus U^l $ for some $ k>l+1 $. In the first case, since $ I\cap Z^*\subseteq U^0 $, $ \rho\not\in\spec $, which would be a contradiction. On the other hand, in the second case, after going to $ U^k\setminus U^l $, $ \rho $ has an edge going from some $ v\in Z^*\backslash U^{k-1} $ to some $ v'\in U^l $ (else $ I\cap Z^*\subseteq U^0\subseteq U^l $ can not be reached). But then $ v $ would be added to $ U^{k+1} $, which contradicts to the fact that $ k>l+1 $.
		In either case, we get a contradiction, so $ \rho\in\lang(\assump)$. 

		\noindent \item \textit{Complexity analysis.}
		The computation of cooperative winning region can be done in time linear in size of the game graph, i.e. $ \bigO(m+n) $. The procedure \textsc{ComputeLiveGroups} takes $ \bigO(m+n) $ time. Hence, resulting in time linear in number of edges in the game graph.
	\end{inparaitem}
	
\end{proof}

\section{\aname Assumptions for \cobuchi games}\label{proof:cobuechi assumption is correct}

For the convenience of the reader, we restate Thm. \ref{thm:coBuechi assumptions} here.
\restatecobuchi*

\begin{proof}
	We note that $ V^0\cap Z^* $ can be partitioned into sets $ V_1, V_2\ldots, V_p $, where $ V_i= (X^i\setminus X^{i-1})\cap V^0 $. We say $ v\in Z^* $ has \emph{rank} $ i $ if $ v\in X^i\setminus X^{i-1} $.

	\begin{inparaitem}[$\blacktriangleright$]
		\noindent \item Implementability: 
		We again observe that the sources of the co-live edges are $ \po $'s vertices and by construction, each source has at least one alternative edge that is neither co-live nor unsafe. Hence, they can be easily implemented by $ \po $, by taking the co-live edges only finitely often.  We again prove sufficiency, implementability and permissiveness next and finally comment on the complexity of $\computeCoLive$.

	\noindent \item Sufficiency: 	
	Consider the following strategy $ \pi^0 $ for $ \pz $: at a vertex $ v\in V_1 $, she takes edge $ (v,v')\in E $ such that $ v'\in X^1  $, at a vertex $ v\in V_i $, for $ i\in [2;p] $, she plays the existing edge to reach $ X^{i-1} $, and for all other vertices, she plays arbitrarily.
	
	Let $ v_0\in Z^*=\team{0,1}\lozenge\square U $, by correctness of the fixpoint (see \cite[chapter 2]{AutomataLogicsandInfiniteGames}). Let $ \strato $ be an arbitrary strategy of $ \po $ such that $ \lang(\strato)\subseteq \lang(\assump) $, and $ \rho=v_0v_1\ldots $ be an arbitrary $ \stratz\strato $-play. Then $ \rho\in \lang(\assump) $.
	
	Since $ \rho\in \lang(\assumpsafe(\safegroup)) $, $ v_i\in Z^* $ for all $ i $. Now suppose $ \rho\not\in \lang(\spec) $, i.e. $ \inf(\rho)\cap (Z^*\setminus U)\not=\emptyset $. Let $ u\in Z^*\setminus U $. Then to reach $ u $ infinitely often some edge from $ \colivegroup$ must be taken infinitely often in $ \rho $, which contradicts the fact that $ \rho\in \lang(\assump) $. Hence, $ \rho\in\lang(\spec) $.

	\noindent \item Permissiveness: 
	 Let $ \rho= v_0v_1\ldots $ such that $ v_0\in Z^* $ and $ \rho\in\lang(\spec) $. Suppose that $ \rho\not\in \lang(\assump) $.
	
	Case 1: If $ \rho\not\in\assumpsafe(\safegroup) $. Then the same argument as in the \buchi case gives a contradiction.
	
	Case 2: If $ \rho\not\in\assumpdep(\colivegroup) $, that is $ \exists (u,v)\in \colivegroup, $ such that $ \rho $ takes $ (u,v) $ infinitely often. By the definition of $ \colivegroup$, $ v\in Z^*\setminus X^1 $, implying $ v\not\in U $, since if $ v\in U $ then it would have been in $ X^1 $ (see \cite{AutomataLogicsandInfiniteGames}). Hence, $ \rho\not\in\lang(\spec) $, giving a contradiction. So $ \rho\in\lang(\assump)$.

	\noindent \item \textit{Complexity analysis.}
		Very similar to that for live group assumptions and therefore omitted.
	\end{inparaitem}
	
\end{proof}

\subsection{Accelerated fix-point algorithm for \cobuchi}\label{sec:accelcobuchi}
We now prove the correctness of the accelerated fix-point algorithm.
\Trestatecobuchi*

\begin{proof}
	Analogous to the \buchi case, every vertex in $ Z^*\cap V^0 $ is added to the least fixpoint computation of $ X $ in the $ \textsf{attr}^0 $ part of $ \textsf{tpre} $, and $ V^0\cap Z^* $ can be partitioned into sets $ V_1, V_2\ldots, V_p $, where $ V_i= (X^i\setminus X^{i-1})\cap V^0 $. We say $ v\in Z^* $ has \emph{rank} $ i $ if $ v\in X^i\setminus X^{i-1} $.
	
	Consider the strategy $ \pi^0 $ for $ \pz $: at a vertex $ v\in V_1 $, takes edge $ (v,v')\in E $ such that $ v'\in X^1  $, at a vertex $ v\in V_i $, for $ i\in [2;p] $, plays the $ \textsf{attr}^0 $ strategy to reach $ X^{i-1} $, and for other vertices, plays arbitrarily.
	
	Let $ v_0\in Z^*=\team{0,1}U $, by application of Prop. \ref{prop:cobuechi_tpre}. Let $ \strato $ be an arbitrary strategy of $ \po $ such that $ \lang(\strato)\subseteq \lang(\assump) $, and $ \rho=v_0v_1\ldots $ be an arbitrary $ \stratz\strato $-play. Then $ \rho\in \lang(\assump) $.
	
	Since $ \rho\in \lang(\assumpsafe(\safegroup)) $, $ v_i\in Z^* $ for all $ i $. Now suppose $ \rho\not\in \lang(\spec) $, i.e. $ \inf(\rho)\cap (Z^*\setminus U)\not=\emptyset $. Let $ u\in Z^*\setminus U $. Then to reach $ u $ infinitely often some edge from $ \colivegroup$ must be taken infinitely often in $ \rho $, which contradicts the fact that $ \rho\in \lang(\assump) $. Hence, $ \rho\in\lang(\spec) $.
	
	Now we show the permissiveness of the assumption. Let $ \rho= v_0v_1\ldots $ such that $ v_0\in Z^* $ and $ \rho\in\lang(\spec) $. Suppose that $ \rho\not\in \lang(\assump) $.
	
	Case 1: If $ \rho\not\in\assumpsafe(\safegroup) $. Then the same argument as in the \buchi case gives a contradiction.
	
	Case 2: If $ \rho\not\in\assumpdep(\colivegroup) $, that is $ \exists (u,v)\in \colivegroup, $ such that $ \rho $ takes $ (u,v) $ infinitely often. By the definition of $ \colivegroup$, $ v\in Z^*\setminus X^1 $, implying $ v\not\in U $, since if $ v\in U $ then it would have been in $ X^1 $ (see \cite{AutomataLogicsandInfiniteGames}). Hence, $ \rho\not\in\lang(\spec) $, giving a contradiction. So $ \rho\in\lang(\assump)$.

	We again observe that the sources of the co-live edges are $ \po $'s vertices and by construction, each source has at least one alternative edge that is neither co-live nor unsafe. Hence, they can be easily implemented by $ \po $, by taking those edges only finitely often.
	
	The complexity analysis is similar to that for live group assumptions.
	
\end{proof}

\subsection{Faster algorithm for \cobuchi games}\label{sec:fastercobuchi}

\begin{algorithm}[H]
	\caption{\computeCoLive}
	\label{alg:compute cobuechi assumption}
		\begin{algorithmic}[1]
			\Require $ \gamegraph=\tup{V=\vertexz\cup \vertexo, E}, I\subseteq V$
			\Ensure Assumption $\assump$ on $ \p{1} $ 
			\State $Z^* \gets \solveCobuchi_{0,1}(\gamegraph,I)$\label{algo:coop cobuechi computation}
			\State $\safegroup\gets  \computeSafe(\gamegraph,Z^*)$
			\State $\gamegraph\gets \gamegraph|_{Z^*}, I \gets I\cap Z^*$\label{algo:restricted cobuechi game}\Comment{All vertices are cooperatively \cobuchi winning}
			\State $ \colivegroup \gets $\textsc{ComputeCoLive}$ ((\gamegraph,I),\emptyset) $
			\State \Return $(\safegroup,\colivegroup)$ 
			
			\Statex
			\Procedure {ComputeCoLive}{$(\gamegraph,I),\colivegroup$}
			\State $ U\gets \solveSafety_{0,1}(\gamegraph, I) $\Comment{$ U\subseteq I $}
			\State $ D\gets (U\times V\backslash U)\cap E $
			
			\While{$ U\not=V $}
			\State $ W_{attr}\gets \attr{\gamegraph}{\gamegraph, U}{0} $\label{algo:cobuechi:compute attr}
			\State $ U\gets U\cup W_{attr} $\label{algo:cobuechi:add attr}
			\State $ C\gets \cpre{\gamegraph}{U}{1}  \setminus U $\label{algo:cobuechi:compute front}
			\State $ D\gets D \cup \{ (u,v)\in E\cap (C\times V\setminus U) \} $\label{algo:cobuechi:add assump live group}
			\State $ U\gets U\cup C $\label{algo:cobuechi:add front}
			\EndWhile
						
			\State \Return $ \colivegroup $
			\EndProcedure
		\end{algorithmic}
\end{algorithm}

\begin{theorem}
	Given a game graph $ \gamegraph=\tup{V=\vertexz\cup \vertexo, E} $, with \cobuchi objective $ \spec=\lozenge\square I $ for $ \p{i} $. Then Alg. \ref{alg:compute cobuechi assumption} terminates in time $\bigO(m+n)$, and $ \assump = \assumpsafe (S) \wedge \assumpdep(\colivegroup) $ is an adequately permissive assumption on $ \p{1} $ from $ Z^* $. Here, $n=|V|$ and $m=|E|$.
\end{theorem}

\begin{proof}
	We first show that the algorithm terminates. We show that the procedure $ \textsc{ComputeCoLive} $ terminates when all the vertices of the game graph are cooperatively winning for the \cobuchi objective $ \spec=\lozenge\square I $, since we restrict the graph to the cooperative winning region in step \ref{algo:restricted cobuechi game}. We claim that $ U=V=Z^* $, eventually.
	
	Let $ U_l $ be the value of the variable $ U $ after $ l $-th iteration of the while loop, with $ U_0= \solveSafety_{0,1}(\gamegraph, I) $. Since vertices are only added in $ U $, 
	 $ U_0\subset U_1\subset\ldots\subset U_{m}=U_{m+1} $
	 for some $ m\in \N $. Suppose $ V\not\subseteq U_m $, then there exists $ v\in V\backslash U_m $. Since $ v\in Z^* $, there is a $ \play = vv_1v_2\ldots $ from $ v $ to $ U_0 $ and stays there forever. Then consider the largest index $l $ such that $ v_l\not\in U_m $, but $ v_{l+1}\in U_m $. Note that this index exists because $ U_0\subseteq U_m$. 
	 
	 
	 Then if $ v_l\in V^0 $, it would be added to $ U $ in step \ref{algo:cobuechi:add attr} of $ (m+1) $-th iteration, i.e. $ U_m\not=U_{m+1} $. Else if $ v_l\in V^1 $, it would be added to $ U $ in step \ref{algo:cobuechi:add front} of $ (m+1) $-th iteration since $ v_{l+1}\in U_m $, i.e. $ U_m\not=U_{m+1} $. In either case, we get a contradiction. Hence, $ v\in U_m $, implying $ Z^*=U_m $.  Hence the procedure \textsc{ComputeCoLive}, and hence the Algo. \ref{alg:compute cobuechi assumption}, terminates.
	
	Now we show that $ \assump $ is an adequately permissive assumption. Again, let $ U_l $ and $ m $ be as defined earlier. Define $ X_l\coloneqq U_l\backslash U_{l-1} $ for $ 1\leq l\leq m $, and $ X_0=U_0=I $. Then every vertex $ v\in Z^* $ is in $ X_l $ for some $ l\in [0;m] $.
	
	We again prove sufficiency, implementability and permissiveness separately and finally comment on the complexity of $\computeCoLive$.

	\begin{inparaitem}[$\blacktriangleright$]
		\noindent \item Implementability: 
		We again observe that the sources of the co-live edges in $ \colivegroup $ are $ \p{1} $'s vertices and  by construction, each source has at least one alternative edge that is neither co-live nor unsafe. Hence, they can be easily implemented by $ \p{1} $, by taking those edges only finitely often.

		\noindent \item Sufficiency: 	
		Consider the following strategy $ \stratz $ for $ \p{0} $: at a vertex $ v\in X_0\cap V^0  $, she takes edge $ (v,v')\in E $ such that $ v'\in X_0  $, at a vertex $ v\in X_l\cap V^0 $, for $ l\in [2;m] $, she plays the $ \textsf{attr}^0 $ strategy to reach $ U_{l-1} $, and for all other vertices, she plays arbitrarily.
		
		Let $ v_0\in Z^*=\team{0,1}\lozenge\square I $ (from Step \ref{algo:coop cobuechi computation}). Let $ \strato $ be an arbitrary strategy of $ \p{1} $ such that $ \lang(\strato)\subseteq \lang(\assump) $, and $ \rho=v_0v_1\ldots $ be an arbitrary $ \strati\strato $-play. Then $ \rho\in \lang(\assump) $. It remains to show that $ \rho \in \lang(\spec) $.

		
		Since $ \rho\in \lang(\assumpsafe(\safegroup)) $ and by definition of $ \strati $, $ v_i\in Z^* $ for all $ i $. Now suppose $ \rho\not\in \lang(\spec) $, i.e. $ \inf(\rho)\cap (Z^*\setminus I)\not=\emptyset $. Let $ u\in Z^*\setminus I $. Then to reach $ u $ infinitely often some edge from $ \colivegroup$ must be taken infinitely often in $ \rho $, since $ \strati $ makes the play go towards $ I $. But this contradicts the fact that $ \rho\in \lang(\assump) $. Hence, $ \rho\in\lang(\spec) $.

		\noindent \item Permissiveness: 
		Let $ \rho= v_0v_1\ldots $ such that $ v_0\in Z^* $ and $ \rho\in\lang(\spec) $. Suppose that $ \rho\not\in \lang(\assump) $.
		
		Case 1: If $ \rho\not\in\assumpsafe(\safegroup) $. Then the same argument as in the \buchi case gives a contradiction.
		
		Case 2: If $ \rho\not\in\assumpdep(\colivegroup) $, that is $ \exists (u,v)\in \colivegroup, $ such that $ \rho $ takes $ (u,v) $ infinitely often. By the definition of $ \colivegroup$, $ v\in Z^*\setminus U_0 $, implying $ v\not\in I $, since if $ v\in I $ then it would have been in $ U_0 $. Hence, $ \rho\not\in\lang(\spec) $, giving a contradiction. So $ \rho\in\lang(\assump)$.

		\noindent \item \textit{Complexity analysis.}
		Very similar to that for \buchi objectives and therefore omitted.
	\end{inparaitem}
	
\end{proof}

\section{\aname ssumption for parity games}\label{proof:parity assumption is correct}
For the convenience of the reader, we restate Thm. \ref{thm:parity assumption} here.
\restateparity*

\begin{proof}
We prove sufficiency, implementability and permissiveness below and then analyze the complexity of Alg.~\ref{alg:compute parity assumption}.

\begin{inparaitem}[$\blacktriangleright$]
	\noindent \item Implementability: 
	We note that the assumption is implementable by the implementability of safety, liveness and co-liveness assumptions: if for a conditional live group, the corresponding vertex set is reached infinitely often, and also the sources of live groups are visited infinitely often, $ \po $ can choose the live group edges, since they are controlled by $ \po $. Moreover, there won't be any conflict due to the conditional live groups as there can be no unsafe or co-live edge that is included in a conditional live group by construction.
	
\noindent \item Sufficiency:
 We give a strategy for $ \pz $ depending on the parity of the highest priority $ d $ occurring in the game and show that it is winning under assumption $\assump$ for all vertices in the cooperative winning region $ Z^*=\solveParity(\gamegraph,C) $. The strategy uses finite memory and the winning strategies for $ \po $ in subgames with \buchi (Thm.~\ref{thm:Buechi assumptions}) and \cobuchi (Thm.~ \ref{thm:coBuechi assumptions}) objectives. 
 
 By $ \buchigame(\gamegraph, U) $, we denote the game $ (\gamegraph, \square\lozenge U) $, and by $ \cobuchigame(\gamegraph, U) $, we denote the game $ (\gamegraph, \lozenge\square U) $. We also use the definitions of $d$, $ W_d $ and $ W_{\neg d} $, as in the Algo.~\ref{alg:compute parity assumption}. Consider the following strategy $ \stratz $ of $ \pz $:
 
 \begin{inparaitem}[$\triangleright$]
 \item $ d $ is \textbf{odd}: If the play is in $ V\setminus W_{\neg d} $, then $ \pz $ plays the $ \cobuchigame(\gamegraph, W_{\neg d}) $ winning strategy to eventually end up in $ W_{\neg d} $. If the play is in $ W_{\neg d} \cap Z^* $, $ \pz $ plays the recursive winning strategy for $ (G|_{W_{\neg d}}, \paritygame(C)) $. Otherwise, she plays arbitrarily.
 
  \item $ d $ is \textbf{even}: If the play is in $ W_d $, $ \pz $ switches its strategy among $ \buchigame(\gamegraph, W_d) $, $\buchigame(\gamegraph, W_d \cup W_{d-2})$, $\ldots$, $\buchigame(\gamegraph, W_d \cup W_{d-2} \cup \cdots W_2)$ winning strategies, i.e., for each vertex, she first uses the first strategy in the above sequence, then when that vertex is repeated, she uses the second strategy for the next move, and keeps switching to the next strategies for every move from the same vertex. If the play is in $ V\setminus W_d  \cap Z^* $, then she plays the recursive winning strategy for $ (G|_{W_{\neg d}}, \paritygame(C)) $, where $ C $ is modified again as in line~\ref{algo:reduce game to few color}. Otherwise, she plays arbitrarily.
 \end{inparaitem}

	We prove by induction, on the highest occurring priority $ d $, that the above constructed strategy $ \stratz $ for $ \pz $ ensures satisfying the parity objective on the original game graph if the assumption $ \assump $ is satisfied. For the base case, when $ d=0 $, the constructed strategy is trivially winning, because the only existing color is even. Now let the strategy be winning for $ d-1\geq0 $.
	
	Let $ v_0\in Z^*= \solveParity(\gamegraph,C)$. Let $ \strato $ be an arbitrary strategy of $ \po $ such that $ \lang(\strato)\subseteq \lang(\assump) $, and $ \rho=v_0v_1\ldots $ be an arbitrary $ \stratz\strato $-play. Then $ \rho\in \lang(\assump) $. We need to show that $\rho$ is winning, i.e., $\rho\in\lang(\Phi)$. Note that by the safety assumption and by the construction of $\stratz$, $\rho$ stays in the vertex set $Z^*$. 
	
	Case 1: If $ d $ is odd, then since at vertices in $ V\setminus W_{\neg d} $, $ \pz $ plays to eventually stay in $ W_{\neg d} $, the play can not stay in $ W_{\neg d} $ without violating $ \assumpdep(\colivegroup) $. And if $ \rho $ eventually stays in $ W_{\neg d} $, then by the induction hypothesis, it is winning, since $ W_{\neg d}\cap C_d=\emptyset $.
	
	Case 2: If $ d $ is even, then if the play stays in $ W_d $ eventually, and if the play visits vertices of an odd priority $i$ infinitely often, then $ \strato $ satisfies $\computeLive(\gamegraph,C_{i+1}\cup C_{i+2}\cup \cdots \cup C_{d})$ by the conditional live group assumption. Note that $ \pz $ plays the $ \buchigame(\gamegraph, C_{i+1}\cup C_{i+2}\cup \cdots \cup C_{d})$ winning strategy for infinitely many moves from every vertex occurring in $ \rho $. Since  $ \strato $ satisfies $\computeLive(\gamegraph,C_{i+1}\cup C_{i+2}\cup \cdots \cup C_{d})$, after these moves as well, the play visits $ (C_{i+1}\cup C_{i+2}\cup \cdots \cup C_{d}) $. Hence the play will visit vertices of an even color $>i$ infinitely often, implying that $ \rho$ is winning. Else if $\rho$ stays in $ V\setminus W_{d} $ eventually, then it is winning by induction hypothesis. This gives the sufficiency of the assumptions computed by the algorithm.
	
	\noindent \item  Permissiveness:
	Now for the permissiveness of the assumption,  let $ \rho\in \lang(\spec)$. We prove the claim by contradiction and suppose that $ \rho\not\in\lang(\assump) $. 
	
	Case 1: If $ \rho\not\in\assumpsafe(\safegroup) $. Then some edge $ (v,v')\in \safegroup $ is taken in $ \rho $. Then after reaching $ v' $, $ \rho $ still satisfies the parity objective. Hence, $ v'\in Z^* $, but then $ (v,v')\not\in \safegroup $, which is a contradiction.
	
	Case 2: If $ \rho\not\in\assumpcondlive(\condlivegroup) $. Then for some even $ j $ and odd $i<j$, $\rho$ visits $W_j\cap C_i$ infinitely often but does not satisfy the live transition group assumption $\computeLive(\gamegraph',C_{i+1}\cup C_{i+2}\cdots \cup C_j)) $, where $\gamegraph' = \gamegraph|_{W_j}$. Due to the construction of the set $W_j$, it is easy to see that once $\rho$ visits $W_j$, it can never visit $V\setminus W_j$. Hence, eventually $\rho$ stays in the game $\game'$ and visits $C_i$ infinitely often. Since $\rho \in \lang(\spec)$, it also visits some vertices of some even priority $>i$ infinitely often, and hence, it satisfies $\square\lozenge(C_{i+1}\cup C_{i+3}\cdots \cup C_j)) $ in $ \gamegraph' $. Since $\computeLive(\gamegraph',C_{i+1}\cup C_{i+3}\cdots \cup C_j)) $ is a permissive assumption for $(\gamegraph', \square\lozenge(C_{i+1}\cup C_{i+3}\cdots \cup C_j))) $, the play $\rho$ must satisfy $\computeLive(\gamegraph',C_{i+1}\cup C_{i+3}\cdots \cup C_j)) $, which contradicts the assumption.
	
	Case 3: If $ \rho\not\in\assumpdep(\colivegroup) $. Then for some odd $ i $ an edge $ (u,v)\in \computeCoLive(\gamegraph,W_{\neg i}) $ is taken infinitely often. Then the vertex $ v\in V\setminus W_{\neg i} $ is visited infinitely often. Note that $ \rho $ can not be winning by visiting an even $ j>i $, since otherwise $ v $ would have been in $ \solveBuchi(\gamegraph, C_j) $ as from $ v $ we can infinitely often see $ j $, and hence would have been removed from $ \game $ for the next recursive step. Hence, $ \rho $ visits some even $ j<i $ infinitely often, i.e. $ i $ is not visited infinitely often. Then $ v $ would be in $ W_{\neg d} $, which is a contradiction.

    \noindent \item \textit{Complexity analysis.} We note that the cooperative parity game can be solved in time $ \bigO((n+m)\log d) $, where $ n, m $ and $ d $ are the number of of vertices, edges and priorities respectively: consider the graph where $ pz $ owns all the vertices, find the strongly connected components in time $ \bigO(n+m) $, check which of these components have a cycle with highest priority even by reduction to even-cycle problem \cite{valerie2001Paritywordproblem}. Then \textsc{ComputeSets} takes time $ \bigO(n^2) $ for the even case, but is dominated by $ \bigO(n^3) $ time for the odd case. For every priority, \textsc{ComputeSets} is called once, that is at most $ 2n $ calls in total. Then the total running time of the algorithm is $ \bigO((n+m)\log d+2n.(n^3)) = \bigO(n^4) $.
\end{inparaitem}
\end{proof}

\section{Equivalence of Def.~\ref{def:winundera} and Def.~\ref{def:winundera:alt}}\label{appendix:equivOfDefWinUnder}
We prove the following result stating the equivalence of between Def.~\ref{def:winundera} and Def.~\ref{def:winundera:alt} for the class of assumptions we consider.
\begin{proposition}
Given a parity game $\game$, let $\assump$ be an \aname computed by Thm.~\ref{thm:parity assumption}. Then, a $\pz$ strategy is winning under $\assump$ by Def.~\ref{def:winundera} if and only  if it is winning under $\assump$ by Def.~\ref{def:winundera:alt}.
\end{proposition}
\begin{proof}
Suppose a $\pz$ strategy $\stratz$ is winning under assumption $\assump$ by Def.~\ref{def:winundera:alt}. Then every play $\play\in \lang(\stratz)$ either fails to satisfy the assumption $\assump$ or satisfies the specification $\spec$. Hence, $\lang(\stratz) \subseteq \lang(\spec) \cup \lang(\neg\assump)$.
Now, let $\strato$ be a $\po$ strategy s.t. $\lang(\strato) \subseteq \lang(\assump)$. Then,
\begin{align*}
\lang(\stratz\strato) &= \lang(\stratz) \cap \lang(\strato)\\
&\subseteq \big( \lang(\spec) \cup \lang(\neg\assump)\big) \cap \lang(\assump)\\
&= \big(\lang(\spec)\cap\lang(\assump)\big) \cup \big(\lang(\neg\assump)\cap \lang(\assump)\big)\\
&=\lang(\spec)\cap\lang(\assump)\\
&\subseteq \lang(\spec).
\end{align*}
Hence, $\stratz$ is also winning under assumption $\assump$ by Def.~\ref{def:winundera}.

Now, for the other direction, suppose $\stratz$ is winning under assumption $\assump$ by Def.~\ref{def:winundera}. Let $\play\in\lang(\stratz)$. If $\play\not\in\lang(\assump)$, then we are done. Suppose $\play\in\lang(\assump)$, then we have to show that $\play\in\lang(\spec)$. We claim that there exists a $\po$ strategy $\strato$ such that $\strato\subseteq \lang(\assump)$ and $\play$ is compliant with it. Then, by Def.~\ref{def:winundera}, $\lang(\stratz\strato)\subseteq\lang(\spec)$. As $\play$ is compliant with both $\stratz$ and $\strato$, $\play\in\lang(\spec)$, and hence, we are done.

Now we only need to prove the claim. As $\assump$ is an implementable assumption, there exists a $\po$ strategy $\strato_*$ such that $\lang(\strato_*)\subseteq\lang(\assump)$ and by definition, the strategy is defined on all vertices. Let $\play = v_0v_1\cdots$. Now, let $\strato$ be another $\po$ strategy such that for every play prefix $\playprefix\in V^*V^1$, it is defined as follows:
\[\strato(\playprefix)=
\begin{cases}
v_k &\text{if $\playprefix = v_0v_1\cdots v_{k-1}$}\\
\strato_*(\playprefix) &\text{otherwise.}
\end{cases}\]
Then, clearly, $\play$ is compliant with $\strato$. Now, let $\play'\in \lang(\strato)$, then it is enough to show that $\play'\in\lang(\assump)$. If $\play'=\play \in \lang(\assump)$, then we are done. Suppose not and let $\playprefix'$ be the maximal prefix of $\play'$ that is also a prefix of $\play$ (which can also be empty). By construction, the moves taken after the prefix $\playprefix'$ in $\play'$ are compliant with $\strato_*$. As the conditional live group templates and co-liveness templates are tail properties and are independent of prefixes, $\play'$ satisfies those templates of assumption $\assump$. Furthermore, as $\playprefix'$ is a prefix of $\play\in\lang(\assump)$ and the moves taken after $\playprefix'$ in $\play'$ are compliant with $\strato_*$, the play $\play'$ can not contain any unsafe edges marked by assumption $\assump$. Therefore, $\play'\in\lang(\assump)$.

\end{proof}

	
\section{Complete table for experimental results}\label{appendix:extended_table}
\scriptsize
\centering
	\begin{longtable}{
			|>{\centering\arraybackslash}m{3.5cm}|
			>{\centering\arraybackslash}m{1.5cm}|
			>{\centering\arraybackslash}m{1.5cm}|
			>{\centering\arraybackslash}m{1.3cm}|
			>{\centering\arraybackslash}m{2cm}|
			>{\centering\arraybackslash}m{2cm}|}
		\hline
		Name & Number of Vertices	&	Number of Edges	& Number of Priorities &	Computation time of $\tool$ (in~seconds) & Computation time of $\krishtool$ (in~seconds) 	\\ 
		\hline
		\name{abcg arbiter} & \sepcomma{417}	&	\sepcomma{638} & 3 &  \secformat{0.016}	&	\secformat{9.79}\\ 
		\name{ActionConverter} & \sepcomma{134}	&	\sepcomma{200} & 3 &  \secformat{0.003}	&	\secformat{0.005}\\ 
		\name{amba decomposed arbiter 2} & \sepcomma{141}	&	\sepcomma{212} & 4 &  \secformat{0.005}	&	\secformat{0.21}\hspace*{0.2cm}*\\ 
		\name{amba decomposed arbiter 3} & \sepcomma{741}	&	\sepcomma{1176} & 4 &  \secformat{0.064}	&	\secformat{3.798}\hspace*{0.2cm}*\\ 
		\name{amba decomposed arbiter 4} & \sepcomma{2468}	&	\sepcomma{3950} & 4 &  \secformat{0.703}	&	\secformat{35.911}\hspace*{0.2cm}*\\ 
		\name{amba decomposed arbiter 5} & \sepcomma{7273}	&	\sepcomma{11576} & 4 &  \secformat{7.095}	&	\secformat{297.975}\hspace*{0.2cm}*\\ 
		\name{amba decomposed arbiter} & \sepcomma{36824}	&	\sepcomma{67018} & 4 &  \secformat{203.104}	&		\timeout\\ 
		\name{amba decomposed decode} & \sepcomma{19}	&	\sepcomma{26} & 3 &  \secformat{0.002}	&	\secformat{0.004}\\ 
		\name{amba decomposed encode 10} & \sepcomma{78313}	&	\sepcomma{117676} & 3 &  \secformat{11.882}	&		\timeout\\ 
		\name{amba decomposed encode 11} & \sepcomma{156184}	&	\sepcomma{234506} & 3 &  \secformat{45.455}	&		\timeout\\ 
		\name{amba decomposed encode 12} & \sepcomma{311879}	&	\sepcomma{468072} & 3 &  \secformat{181.623}	&		\timeout\\ 
		\name{amba decomposed encode 13} & \sepcomma{623222}	&	\sepcomma{935110} & 3 &  \secformat{732.713}	&		\timeout\\ 
		\name{amba decomposed encode 14} & \sepcomma{1245861}	&	\sepcomma{1869092} & 3 &  \secformat{2960.11}	&		\timeout\\ 
		\name{amba decomposed encode 2} & \sepcomma{95}	&	\sepcomma{140} & 3 &  \secformat{0.003}	&	\secformat{0.562}\\ 
		\name{amba decomposed encode 3} & \sepcomma{264}	&	\sepcomma{402} & 3 &  \secformat{0.005}	&	\secformat{14.5}\\ 
		\name{amba decomposed encode 4} & \sepcomma{499}	&	\sepcomma{760} & 3 &  \secformat{0.008}	&	\secformat{97.901}\\ 
		\name{amba decomposed encode 5} & \sepcomma{1534}	&	\sepcomma{2342} & 3 &  \secformat{0.025}	&	\secformat{2616.04}\\ 
		\name{amba decomposed encode 6} & \sepcomma{2965}	&	\sepcomma{4500} & 3 &  \secformat{0.06}	&		\timeout\\ 
		\name{amba decomposed encode 7} & \sepcomma{5804}	&	\sepcomma{8770} & 3 &  \secformat{0.168}	&		\timeout\\ 
		\name{amba decomposed encode 8} & \sepcomma{11459}	&	\sepcomma{17264} & 3 &  \secformat{0.534}	&		\timeout\\ 
		\name{amba decomposed encode 9} & \sepcomma{39354}	&	\sepcomma{59214} & 3 &  \secformat{3.296}	&		\timeout\\ 
		\name{amba decomposed encode} & \sepcomma{292}	&	\sepcomma{476} & 3 &  \secformat{0.006}	&	\secformat{4.23}\\ 
		\name{amba decomposed lock 2} & \sepcomma{398}	&	\sepcomma{592} & 3 &  \secformat{0.04}	&	\secformat{50.282}\\ 
		\name{amba decomposed lock 3} & \sepcomma{1558}	&	\sepcomma{2336} & 3 &  \secformat{0.074}	&	\secformat{2999.65}\\ 
		\name{amba decomposed lock 4} & \sepcomma{6182}	&	\sepcomma{9280} & 3 &  \secformat{0.608}	&		\timeout\\ 
		\name{amba decomposed lock 5} & \sepcomma{24646}	&	\sepcomma{36992} & 3 &  \secformat{9.563}	&		\timeout\\ 
		\name{amba decomposed lock 6} & \sepcomma{98438}	&	\sepcomma{147712} & 3 &  \secformat{158.854}	&		\timeout\\ 
		\name{amba decomposed lock 7} & \sepcomma{393478}	&	\sepcomma{590336} & 3 &  \secformat{2801.93}	&		\timeout\\ 
		\name{amba decomposed lock} & \sepcomma{180}	&	\sepcomma{288} & 3 &  \secformat{0.009}	&	\secformat{1.288}\\ 
		\name{amba decomposed shift} & \sepcomma{44}	&	\sepcomma{64} & 3 &  \secformat{0.006}	&	\secformat{0.021}\\ 
		\name{amba decomposed tburst4} & \sepcomma{1061}	&	\sepcomma{1618} & 4 &  \secformat{0.112}	&	\secformat{685.687}\\ 
		\name{amba decomposed tincr} & \sepcomma{2082}	&	\sepcomma{3280} & 3 &  \secformat{0.227}	&		\timeout\\ 
		\name{amba decomposed tsingle} & \sepcomma{950}	&	\sepcomma{1450} & 4 &  \secformat{0.094}	&	\secformat{464.176}\\ 
		\name{arbiter} & \sepcomma{39}	&	\sepcomma{54} & 3 &  \secformat{0.008}	&	\secformat{0.024}\\ 
		\name{arbiter with buffer} & \sepcomma{84}	&	\sepcomma{132} & 3 &  \secformat{0.003}	&	\secformat{0.032}\\ 
		\name{arbiter with cancel} & \sepcomma{139}	&	\sepcomma{212} & 3 &  \secformat{0.004}	&	\secformat{1.592}\\ 
		\name{Button} & \sepcomma{12}	&	\sepcomma{16} & 3 &  \secformat{0.007}	&	\secformat{0.007}\\ 
		\name{detector} & \sepcomma{84}	&	\sepcomma{126} & 4 &  \secformat{0.008}	&	\secformat{0.457}\\ 
		\name{detector unreal} & \sepcomma{141}	&	\sepcomma{210} & 4 &  \secformat{0.005}	&	\secformat{0.301}\hspace*{0.2cm}*\\ 
		\name{EnemeyModule} & \sepcomma{20}	&	\sepcomma{28} & 3 &  \secformat{0.005}	&	\secformat{0.005}\\ 
		\name{EscalatorBidirectional} & \sepcomma{1622}	&	\sepcomma{2464} & 3 &  \secformat{0.18}	&	\secformat{4.123}\\ 
		\name{EscalatorCountingInit} & \sepcomma{99}	&	\sepcomma{148} & 3 &  \secformat{0.006}	&	\secformat{0.045}\\ 
		\name{EscalatorCounting} & \sepcomma{79}	&	\sepcomma{118} & 3 &  \secformat{0.003}	&	\secformat{0.012}\\ 
		\name{EscalatorNonCounting} & \sepcomma{21}	&	\sepcomma{30} & 3 &  \secformat{0.002}	&	\secformat{0.003}\\ 
		\name{EscalatorNonReactive} & \sepcomma{7}	&	\sepcomma{8} & 3 &  \secformat{0.002}	&	\secformat{0.002}\\ 
		\name{EscalatorSmart} & \sepcomma{1783}	&	\sepcomma{3008} & 5 &  \secformat{0.146}	&		\timeout\\ 
		\name{full arbiter 2} & \sepcomma{204}	&	\sepcomma{324} & 3 &  \secformat{0.007}	&	\secformat{4.939}\\ 
		\name{full arbiter 3} & \sepcomma{1403}	&	\sepcomma{2396} & 3 &  \secformat{0.082}	&	\secformat{538.772}\\ 
		\name{full arbiter 4} & \sepcomma{7444}	&	\sepcomma{12764} & 3 &  \secformat{3.171}	&		\timeout\\ 
		\name{full arbiter 5} & \sepcomma{44019}	&	\sepcomma{76920} & 3 &  \secformat{109.252}	&		\timeout\\ 
		\name{full arbiter} & \sepcomma{774}	&	\sepcomma{1218} & 3 &  \secformat{0.025}	&	\secformat{704.267}\\ 
		\name{full arbiter unreal1} & \sepcomma{774}	&	\sepcomma{1218} & 3 &  \secformat{0.041}	&	\secformat{697.743}\\ 
		\name{full arbiter unreal2} & \sepcomma{774}	&	\sepcomma{1218} & 3 &  \secformat{0.027}	&	\secformat{698.161}\\ 
		\name{Gamemodule} & \sepcomma{143}	&	\sepcomma{214} & 3 &  \secformat{0.008}	&	\secformat{0.032}\\ 
		\name{Increment} & \sepcomma{12}	&	\sepcomma{16} & 3 &  \secformat{0.002}	&	\secformat{0.002}\\ 
		\name{KitchenTimerV0} & \sepcomma{20}	&	\sepcomma{28} & 3 &  \secformat{0.002}	&	\secformat{0.002}\\ 
		\name{KitchenTimerV1} & \sepcomma{80}	&	\sepcomma{124} & 3 &  \secformat{0.003}	&	\secformat{0.03}\\ 
		\name{KitchenTimerV2} & \sepcomma{731}	&	\sepcomma{1138} & 3 &  \secformat{0.038}	&	\secformat{3.302}\\ 
		\name{KitchenTimerV3} & \sepcomma{1583}	&	\sepcomma{2482} & 3 &  \secformat{0.165}	&	\secformat{14.573}\\ 
		\name{KitchenTimerV4} & \sepcomma{2272}	&	\sepcomma{3554} & 3 &  \secformat{0.408}	&	\secformat{34.324}\\ 
		\name{KitchenTimerV5} & \sepcomma{4111}	&	\sepcomma{6584} & 3 &  \secformat{1.422}	&	\secformat{134.261}\\ 
		\name{KitchenTimerV6} & \sepcomma{4099}	&	\sepcomma{6560} & 3 &  \secformat{1.438}	&	\secformat{140.525}\\ 
		\name{lilydemo01} & \sepcomma{104}	&	\sepcomma{148} & 3 &  \secformat{0.003}	&	\secformat{0.045}\\ 
		\name{lilydemo02} & \sepcomma{104}	&	\sepcomma{148} & 3 &  \secformat{0.003}	&	\secformat{0.044}\\ 
		\name{lilydemo03} & \sepcomma{211}	&	\sepcomma{312} & 3 &  \secformat{0.005}	&	\secformat{24.546}\\ 
		\name{lilydemo04} & \sepcomma{304}	&	\sepcomma{448} & 3 &  \secformat{0.011}	&	\secformat{77.08}\\ 
		\name{lilydemo05} & \sepcomma{293}	&	\sepcomma{436} & 3 &  \secformat{0.011}	&	\secformat{59.884}\\ 
		\name{lilydemo06} & \sepcomma{369}	&	\sepcomma{548} & 3 &  \secformat{0.026}	&	\secformat{121.068}\\ 
		\name{lilydemo07} & \sepcomma{78}	&	\sepcomma{108} & 3 &  \secformat{0.004}	&	\secformat{1.26}\\ 
		\name{lilydemo08} & \sepcomma{26}	&	\sepcomma{36} & 4 &  \secformat{0.005}	&	\secformat{0.01}\\ 
		\name{lilydemo09} & \sepcomma{47}	&	\sepcomma{66} & 4 &  \secformat{0.006}	&	\secformat{0.023}\hspace*{0.2cm}*\\ 
		\name{lilydemo10} & \sepcomma{46}	&	\sepcomma{68} & 4 &  \secformat{0.002}	&	\secformat{0.11}\\ 
		\name{lilydemo11} & \sepcomma{65}	&	\sepcomma{100} & 3 &  \secformat{0.007}	&	\secformat{0.108}\\ 
		\name{lilydemo12} & \sepcomma{110}	&	\sepcomma{172} & 3 &  \secformat{0.003}	&	\secformat{0.23}\\ 
		\name{lilydemo13} & \sepcomma{11}	&	\sepcomma{14} & 3 &  \secformat{0.007}	&	\secformat{0.011}\\ 
		\name{lilydemo14} & \sepcomma{399}	&	\sepcomma{660} & 6 &  \secformat{0.03}	&	\secformat{166.612}\\ 
		\name{lilydemo15} & \sepcomma{133}	&	\sepcomma{206} & 3 &  \secformat{0.008}	&	\secformat{0.48}\\ 
		\name{lilydemo17} & \sepcomma{3102}	&	\sepcomma{5334} & 7 &  \secformat{0.22}	&		\timeout\\ 
		\name{lilydemo18} & \sepcomma{449}	&	\sepcomma{728} & 9 &  \secformat{0.017}	&	\secformat{1566.9}\\ 
		\name{lilydemo19} & \sepcomma{108}	&	\sepcomma{162} & 4 &  \secformat{0.009}	&	\secformat{0.167}\hspace*{0.2cm}*\\ 
		\name{lilydemo20} & \sepcomma{1989}	&	\sepcomma{3252} & 4 &  \secformat{0.4}	&		\timeout\\ 
		\name{lilydemo21} & \sepcomma{4087}	&	\sepcomma{6824} & 3 &  \secformat{0.428}	&		\timeout\\ 
		\name{lilydemo22} & \sepcomma{373}	&	\sepcomma{556} & 3 &  \secformat{0.02}	&	\secformat{109.443}\\ 
		\name{lilydemo23} & \sepcomma{38}	&	\sepcomma{48} & 3 &  \secformat{0.008}	&	\secformat{0.072}\\ 
		\name{lilydemo24} & \sepcomma{263}	&	\sepcomma{406} & 4 &  \secformat{0.008}	&	\secformat{54.354}\\ 
		\name{load balancer} & \sepcomma{255}	&	\sepcomma{390} & 4 &  \secformat{0.008}	&	\secformat{105.297}\\ 
		\name{load balancer unreal1} & \sepcomma{328}	&	\sepcomma{506} & 4 &  \secformat{0.01}	&	\secformat{202.229}\\ 
		\name{load balancer unreal2} & \sepcomma{326}	&	\sepcomma{502} & 4 &  \secformat{0.024}	&	\secformat{200.707}\\ 
		\name{loadcomp2} & \sepcomma{387}	&	\sepcomma{614} & 4 &  \secformat{0.024}	&	\secformat{2.974}\\ 
		\name{loadcomp3} & \sepcomma{1556}	&	\sepcomma{2734} & 4 &  \secformat{0.262}	&	\secformat{62.287}\\ 
		\name{loadcomp4} & \sepcomma{3943}	&	\sepcomma{6968} & 4 &  \secformat{1.94}	&	\secformat{979.018}\\ 
		\name{loadcomp5} & \sepcomma{10657}	&	\sepcomma{19136} & 4 &  \secformat{15.525}	&		\timeout\\ 
		\name{loadfull2} & \sepcomma{331}	&	\sepcomma{552} & 4 &  \secformat{0.043}	&	\secformat{1.187}\\ 
		\name{loadfull3} & \sepcomma{1159}	&	\sepcomma{2030} & 4 &  \secformat{0.143}	&	\secformat{28.051}\\ 
		\name{loadfull4} & \sepcomma{3399}	&	\sepcomma{6084} & 4 &  \secformat{1.421}	&	\secformat{496.464}\\ 
		\name{loadfull5} & \sepcomma{6483}	&	\sepcomma{11184} & 4 &  \secformat{5.595}	&		\timeout\\ 
		\name{ltl2dba01} & \sepcomma{101}	&	\sepcomma{152} & 4 &  \secformat{0.011}	&	\secformat{9.012}\\ 
		\name{ltl2dba02} & \sepcomma{489}	&	\sepcomma{758} & 4 &  \secformat{0.011}	&	\secformat{74.949}\\ 
		\name{ltl2dba03} & \sepcomma{135}	&	\sepcomma{200} & 4 &  \secformat{0.006}	&	\secformat{65.785}\\ 
		\name{ltl2dba04} & \sepcomma{165}	&	\sepcomma{260} & 4 &  \secformat{0.007}	&	\secformat{0.28}\\ 
		\name{ltl2dba05} & \sepcomma{361}	&	\sepcomma{560} & 4 &  \secformat{0.01}	&	\secformat{7.105}\hspace*{0.2cm}*\\ 
		\name{ltl2dba06} & \sepcomma{329}	&	\sepcomma{496} & 4 &  \secformat{0.009}	&	\secformat{703.712}\\ 
		\name{ltl2dba07} & \sepcomma{4368}	&	\sepcomma{6656} & 4 &  \secformat{0.399}	&	\secformat{2085.74}\hspace*{0.2cm}*\\ 
		\name{ltl2dba08} & \sepcomma{41071}	&	\sepcomma{67522} & 4 &  \secformat{5.515}	&		\timeout\\ 
		\name{ltl2dba09} & \sepcomma{75}	&	\sepcomma{120} & 4 &  \secformat{0.004}	&	\secformat{0.023}\\ 
		\name{ltl2dba10} & \sepcomma{77}	&	\sepcomma{118} & 4 &  \secformat{0.003}	&	\secformat{0.026}\\ 
		\name{ltl2dba11} & \sepcomma{21}	&	\sepcomma{30} & 4 &  \secformat{0.002}	&	\secformat{0.003}\\ 
		\name{ltl2dba12} & \sepcomma{81}	&	\sepcomma{120} & 4 &  \secformat{0.003}	&	\secformat{0.335}\\ 
		\name{ltl2dba13} & \sepcomma{151}	&	\sepcomma{222} & 4 &  \secformat{0.009}	&	\secformat{0.198}\hspace*{0.2cm}*\\ 
		\name{ltl2dba14} & \sepcomma{97}	&	\sepcomma{144} & 4 &  \secformat{0.004}	&	\secformat{0.065}\hspace*{0.2cm}*\\ 
		\name{ltl2dba15} & \sepcomma{39}	&	\sepcomma{60} & 4 &  \secformat{0.002}	&	\secformat{0.007}\\ 
		\name{ltl2dba16} & \sepcomma{75}	&	\sepcomma{114} & 4 &  \secformat{0.003}	&	\secformat{0.025}\\ 
		\name{ltl2dba17} & \sepcomma{663}	&	\sepcomma{1020} & 4 &  \secformat{0.009}	&	\secformat{213.754}\\ 
		\name{ltl2dba18} & \sepcomma{265}	&	\sepcomma{420} & 4 &  \secformat{0.009}	&	\secformat{11.288}\\ 
		\name{ltl2dba19} & \sepcomma{106}	&	\sepcomma{152} & 4 &  \secformat{0.004}	&	\secformat{0.652}\hspace*{0.2cm}*\\ 
		\name{ltl2dba20} & \sepcomma{1078}	&	\sepcomma{1760} & 4 &  \secformat{0.018}	&	\secformat{587.085}\\ 
		\name{ltl2dba21} & \sepcomma{2197}	&	\sepcomma{3380} & 4 &  \secformat{0.043}	&		\timeout\\ 
		\name{ltl2dba22} & \sepcomma{21}	&	\sepcomma{30} & 4 &  \secformat{0.008}	&	\secformat{0.025}\\ 
		\name{ltl2dba23} & \sepcomma{51}	&	\sepcomma{78} & 4 &  \secformat{0.008}	&	\secformat{0.073}\\ 
		\name{ltl2dba24} & \sepcomma{19}	&	\sepcomma{26} & 4 &  \secformat{0.002}	&	\secformat{0.011}\\ 
		\name{ltl2dba25} & \sepcomma{104}	&	\sepcomma{158} & 4 &  \secformat{0.004}	&	\secformat{0.074}\\ 
		\name{ltl2dba26} & \sepcomma{148}	&	\sepcomma{236} & 4 &  \secformat{0.003}	&	\secformat{0.19}\\ 
		\name{ltl2dba27} & \sepcomma{21}	&	\sepcomma{30} & 4 &  \secformat{0.005}	&	\secformat{0.059}\\ 
		\name{ltl2dba alpha} & \sepcomma{63}	&	\sepcomma{96} & 4 &  \secformat{0.003}	&	\secformat{0.128}\\ 
		\name{ltl2dba beta} & \sepcomma{907}	&	\sepcomma{1454} & 4 &  \secformat{0.022}	&	\secformat{341.418}\\ 
		\name{ltl2dba C2} & \sepcomma{19}	&	\sepcomma{26} & 4 &  \secformat{0.003}	&	\secformat{0.005}\\ 
		\name{ltl2dba E} & \sepcomma{19}	&	\sepcomma{26} & 4 &  \secformat{0.003}	&	\secformat{0.013}\\ 
		\name{ltl2dba Q} & \sepcomma{46}	&	\sepcomma{68} & 4 &  \secformat{0.003}	&	\secformat{0.077}\\ 
		\name{ltl2dba R} & \sepcomma{19}	&	\sepcomma{26} & 4 &  \secformat{0.002}	&	\secformat{0.003}\\ 
		\name{ltl2dba theta} & \sepcomma{275}	&	\sepcomma{440} & 5 &  \secformat{0.011}	&	\secformat{112.133}\\ 
		\name{ltl2dba U1} & \sepcomma{61}	&	\sepcomma{92} & 4 &  \secformat{0.009}	&	\secformat{0.142}\\ 
		\name{ltl2dpa01} & \sepcomma{102}	&	\sepcomma{164} & 5 &  \secformat{0.004}	&	\secformat{0.52}\\ 
		\name{ltl2dpa02} & \sepcomma{87}	&	\sepcomma{138} & 4 &  \secformat{0.007}	&	\secformat{0.073}\\ 
		\name{ltl2dpa03} & \sepcomma{8135}	&	\sepcomma{14280} & 6 &  \secformat{2.481}	&		\timeout\\ 
		\name{ltl2dpa04} & \sepcomma{40}	&	\sepcomma{56} & 4 &  \secformat{0.004}	&	\secformat{0.017}\\ 
		\name{ltl2dpa05} & \sepcomma{40}	&	\sepcomma{56} & 4 &  \secformat{0.003}	&	\secformat{0.008}\\ 
		\name{ltl2dpa06} & \sepcomma{70}	&	\sepcomma{100} & 4 &  \secformat{0.003}	&	\secformat{0.044}\hspace*{0.2cm}*\\ 
		\name{ltl2dpa07} & \sepcomma{137}	&	\sepcomma{204} & 4 &  \secformat{0.004}	&	\secformat{0.555}\\ 
		\name{ltl2dpa08} & \sepcomma{71}	&	\sepcomma{112} & 4 &  \secformat{0.003}	&	\secformat{0.089}\\ 
		\name{ltl2dpa09} & \sepcomma{137}	&	\sepcomma{204} & 4 &  \secformat{0.004}	&	\secformat{0.556}\\ 
		\name{ltl2dpa10} & \sepcomma{996}	&	\sepcomma{1704} & 6 &  \secformat{0.049}	&	\secformat{500.069}\\ 
		\name{ltl2dpa11} & \sepcomma{174}	&	\sepcomma{268} & 4 &  \secformat{0.015}	&	\secformat{0.612}\\ 
		\name{ltl2dpa12} & \sepcomma{2544}	&	\sepcomma{4422} & 7 &  \secformat{0.714}	&		\timeout\\ 
		\name{ltl2dpa13} & \sepcomma{661}	&	\sepcomma{1118} & 5 &  \secformat{0.066}	&	\secformat{317.768}\\ 
		\name{ltl2dpa14} & \sepcomma{114}	&	\sepcomma{184} & 5 &  \secformat{0.006}	&	\secformat{0.711}\\ 
		\name{ltl2dpa15} & \sepcomma{121}	&	\sepcomma{194} & 4 &  \secformat{0.004}	&	\secformat{0.696}\\ 
		\name{ltl2dpa16} & \sepcomma{118}	&	\sepcomma{188} & 4 &  \secformat{0.009}	&	\secformat{0.438}\\ 
		\name{ltl2dpa17} & \sepcomma{82}	&	\sepcomma{128} & 4 &  \secformat{0.008}	&	\secformat{0.115}\hspace*{0.2cm}*\\ 
		\name{ltl2dpa18} & \sepcomma{82}	&	\sepcomma{128} & 4 &  \secformat{0.008}	&	\secformat{0.111}\hspace*{0.2cm}*\\ 
		\name{ltl2dpa19} & \sepcomma{841}	&	\sepcomma{1392} & 5 &  \secformat{0.041}	&	\secformat{9.497}\hspace*{0.2cm}*\\ 
		\name{ltl2dpa20} & \sepcomma{340}	&	\sepcomma{560} & 4 &  \secformat{0.008}	&	\secformat{34.289}\\ 
		\name{ltl2dpa21} & \sepcomma{996}	&	\sepcomma{1704} & 6 &  \secformat{0.05}	&	\secformat{500.785}\\ 
		\name{ltl2dpa22} & \sepcomma{9351}	&	\sepcomma{16128} & 6 &  \secformat{4.603}	&		\timeout\\ 
		\name{ltl2dpa23} & \sepcomma{443}	&	\sepcomma{676} & 4 &  \secformat{0.019}	&	\secformat{22.212}\\ 
		\name{ltl2dpa24} & \sepcomma{443}	&	\sepcomma{676} & 4 &  \secformat{0.017}	&	\secformat{22.209}\\ 
		\name{MusicAppFeedback} & \sepcomma{1584}	&	\sepcomma{2590} & 3 &  \secformat{0.166}	&	\secformat{1159.62}\\ 
		\name{MusicAppMotivating} & \sepcomma{2501}	&	\sepcomma{4118} & 3 &  \secformat{0.261}	&	\secformat{3064.71}\\ 
		\name{MusicAppSimple} & \sepcomma{344}	&	\sepcomma{562} & 3 &  \secformat{0.018}	&	\secformat{11.77}\\ 
		\name{OneCounterGuiA0} & \sepcomma{12353}	&	\sepcomma{19052} & 3 &  \secformat{13.995}	&	\secformat{1720.65}\\ 
		\name{OneCounterGuiA1} & \sepcomma{21272}	&	\sepcomma{32778} & 3 &  \secformat{38.838}	&		\timeout\\ 
		\name{OneCounterGuiA2} & \sepcomma{35020}	&	\sepcomma{54106} & 3 &  \secformat{104.356}	&		\timeout\\ 
		\name{OneCounterGuiA3} & \sepcomma{62852}	&	\sepcomma{97434} & 3 &  \secformat{411.197}	&		\timeout\\ 
		\name{OneCounterGuiA4} & \sepcomma{63275}	&	\sepcomma{98280} & 3 &  \secformat{412.184}	&		\timeout\\ 
		\name{OneCounterGuiA5} & \sepcomma{63698}	&	\sepcomma{99126} & 3 &  \secformat{405.678}	&		\timeout\\ 
		\name{OneCounterGuiA6} & \sepcomma{69756}	&	\sepcomma{109186} & 3 &  \secformat{500.2}	&		\timeout\\ 
		\name{OneCounterGuiA7} & \sepcomma{71106}	&	\sepcomma{111886} & 3 &  \secformat{508.946}	&		\timeout\\ 
		\name{OneCounterGuiA8} & \sepcomma{84017}	&	\sepcomma{133596} & 3 &  \secformat{747.632}	&		\timeout\\ 
		\name{OneCounterGuiA9} & \sepcomma{87077}	&	\sepcomma{139716} & 3 &  \secformat{786.04}	&		\timeout\\ 
		\name{OneCounterGui} & \sepcomma{12353}	&	\sepcomma{19052} & 3 &  \secformat{13.862}	&	\secformat{1725.76}\\ 
		\name{OneCounterInRangeA0} & \sepcomma{429}	&	\sepcomma{688} & 3 &  \secformat{0.017}	&	\secformat{7.494}\\ 
		\name{OneCounterInRangeA1} & \sepcomma{518}	&	\sepcomma{832} & 3 &  \secformat{0.02}	&	\secformat{15.474}\\ 
		\name{OneCounterInRangeA2} & \sepcomma{534}	&	\sepcomma{864} & 3 &  \secformat{0.018}	&	\secformat{24.461}\\ 
		\name{OneCounterInRangeA3} & \sepcomma{447}	&	\sepcomma{724} & 3 &  \secformat{0.018}	&	\secformat{8.735}\\ 
		\name{OneCounterInRange} & \sepcomma{429}	&	\sepcomma{688} & 3 &  \secformat{0.033}	&	\secformat{7.41}\\ 
		\name{OneCounter} & \sepcomma{87077}	&	\sepcomma{139716} & 3 &  \secformat{805.637}	&		\timeout\\ 
		\name{prioritized arbiter} & \sepcomma{125}	&	\sepcomma{196} & 4 &  \secformat{0.009}	&	\secformat{0.38}\hspace*{0.2cm}*\\ 
		\name{prioritized arbiter unreal1} & \sepcomma{851}	&	\sepcomma{1412} & 4 &  \secformat{0.075}	&	\secformat{10.748}\hspace*{0.2cm}*\\ 
		\name{prioritized arbiter unreal2} & \sepcomma{851}	&	\sepcomma{1412} & 4 &  \secformat{0.085}	&	\secformat{10.739}\hspace*{0.2cm}*\\ 
		\name{prioritized arbiter unreal3} & \sepcomma{10588}	&	\sepcomma{18196} & 4 &  \secformat{14.389}	&	\secformat{2273.01}\hspace*{0.2cm}*\\ 
		\name{RegManager} & \sepcomma{20}	&	\sepcomma{28} & 3 &  \secformat{0.004}	&	\secformat{0.005}\\ 
		\name{robot grid} & \sepcomma{3136}	&	\sepcomma{5408} & 3 &  \secformat{0.091}	&		\timeout\\ 
		\name{RotationCalculator} & \sepcomma{560}	&	\sepcomma{856} & 3 &  \secformat{0.023}	&	\secformat{0.295}\\ 
		\name{round robin arbiter} & \sepcomma{125}	&	\sepcomma{190} & 3 &  \secformat{0.007}	&	\secformat{2.067}\\ 
		\name{round robin arbiter unreal1} & \sepcomma{157}	&	\sepcomma{242} & 3 &  \secformat{0.008}	&	\secformat{3.642}\\ 
		\name{round robin arbiter unreal2} & \sepcomma{157}	&	\sepcomma{242} & 3 &  \secformat{0.009}	&	\secformat{3.689}\\ 
		\name{round robin arbiter unreal3} & \sepcomma{2087}	&	\sepcomma{3388} & 3 &  \secformat{0.048}	&		\timeout\\ 
		\name{SensorRegister} & \sepcomma{12}	&	\sepcomma{16} & 3 &  \secformat{0.007}	&	\secformat{0.007}\\ 
		\name{Sensor} & \sepcomma{7905}	&	\sepcomma{13596} & 4 &  \secformat{8.325}	&		\timeout\\ 
		\name{simple arbiter} & \sepcomma{662}	&	\sepcomma{1184} & 3 &  \secformat{0.039}	&	\secformat{8.903}\\ 
		\name{simple arbiter unreal1} & \sepcomma{2178}	&	\sepcomma{3676} & 3 &  \secformat{0.17}	&	\secformat{1985.03}\\ 
		\name{simple arbiter unreal2} & \sepcomma{6703}	&	\sepcomma{11846} & 3 &  \secformat{1.408}	&		\timeout\\ 
		\name{simple arbiter unreal3} & \sepcomma{45260}	&	\sepcomma{80600} & 3 &  \secformat{201.107}	&	\secformat{3202.98}\\ 
		\name{SliderDefault} & \sepcomma{751}	&	\sepcomma{1250} & 3 &  \secformat{0.029}	&	\secformat{80.1}\\ 
		\name{SPIReadClk} & \sepcomma{37}	&	\sepcomma{54} & 3 &  \secformat{0.002}	&	\secformat{0.003}\\ 
		\name{SPIReadSdi} & \sepcomma{178}	&	\sepcomma{284} & 3 &  \secformat{0.008}	&	\secformat{0.477}\\ 
		\name{SPIWriteClk} & \sepcomma{144}	&	\sepcomma{220} & 3 &  \secformat{0.004}	&	\secformat{0.018}\\ 
		\name{SPIWriteSdi} & \sepcomma{547}	&	\sepcomma{830} & 3 &  \secformat{0.02}	&	\secformat{1.077}\\ 
		\name{starve.ehoa} & \sepcomma{13}	&	\sepcomma{18} & 3 &  \secformat{0.002}	&	\secformat{0.003}\\ 
		\name{starve-smart.ehoa} & \sepcomma{17}	&	\sepcomma{22} & 3 &  \secformat{0.002}	&	\secformat{0.02}\\ 
		\name{test2.ehoa} & \sepcomma{176}	&	\sepcomma{280} & 4 &  \secformat{0.008}	&	\secformat{7.772}\\ 
		\name{TorcsAccelerating} & \sepcomma{21}	&	\sepcomma{30} & 3 &  \secformat{0.007}	&	\secformat{0.008}\\ 
		\name{TorcsGearing} & \sepcomma{27}	&	\sepcomma{38} & 3 &  \secformat{0.002}	&	\secformat{0.006}\\ 
		\name{TorcsSimple} & \sepcomma{85}	&	\sepcomma{130} & 3 &  \secformat{0.003}	&	\secformat{0.017}\\ 
		\name{TorcsSteeringSimple} & \sepcomma{189}	&	\sepcomma{306} & 3 &  \secformat{0.005}	&	\secformat{0.263}\\ 
		\name{TorcsSteeringSmart} & \sepcomma{404}	&	\sepcomma{628} & 3 &  \secformat{0.013}	&	\secformat{2.859}\\ 
		\name{TwoCounters2} & \sepcomma{39747}	&	\sepcomma{61044} & 3 &  \secformat{140.334}	&		\timeout\\ 
		\name{TwoCounters3} & \sepcomma{30511}	&	\sepcomma{46672} & 3 &  \secformat{78.866}	&		\timeout\\ 
		\name{TwoCounters4} & \sepcomma{46702}	&	\sepcomma{72904} & 3 &  \secformat{195.105}	&		\timeout\\ 
		\name{TwoCountersInRangeA0} & \sepcomma{22231}	&	\sepcomma{34212} & 3 &  \secformat{42.399}	&		\timeout\\ 
		\name{TwoCountersInRangeA1} & \sepcomma{31271}	&	\sepcomma{48192} & 3 &  \secformat{84.137}	&		\timeout\\ 
		\name{TwoCountersInRangeA2} & \sepcomma{44750}	&	\sepcomma{69000} & 3 &  \secformat{179.063}	&		\timeout\\ 
		\name{TwoCountersInRangeA3} & \sepcomma{45122}	&	\sepcomma{69744} & 3 &  \secformat{185.008}	&		\timeout\\ 
		\name{TwoCountersInRangeA4} & \sepcomma{45494}	&	\sepcomma{70488} & 3 &  \secformat{183.387}	&		\timeout\\ 
		\name{TwoCountersInRangeA5} & \sepcomma{46098}	&	\sepcomma{71696} & 3 &  \secformat{189.761}	&		\timeout\\ 
		\name{TwoCountersInRangeA6} & \sepcomma{46702}	&	\sepcomma{72904} & 3 &  \secformat{192.715}	&		\timeout\\ 
		\name{TwoCountersInRangeM0} & \sepcomma{46702}	&	\sepcomma{72904} & 3 &  \secformat{196.336}	&		\timeout\\ 
		\name{TwoCountersInRangeM1} & \sepcomma{46702}	&	\sepcomma{72904} & 3 &  \secformat{194.149}	&		\timeout\\ 
		\name{TwoCountersInRangeM2} & \sepcomma{46098}	&	\sepcomma{71696} & 3 &  \secformat{189.046}	&		\timeout\\ 
		\name{TwoCountersInRangeM3} & \sepcomma{46098}	&	\sepcomma{71696} & 3 &  \secformat{186.365}	&		\timeout\\ 
		\name{TwoCountersInRangeM4} & \sepcomma{46098}	&	\sepcomma{71696} & 3 &  \secformat{187.458}	&		\timeout\\ 
		\name{TwoCountersInRangeM5} & \sepcomma{46098}	&	\sepcomma{71696} & 3 &  \secformat{189.278}	&		\timeout\\ 
		\name{TwoCountersInRange} & \sepcomma{22231}	&	\sepcomma{34212} & 3 &  \secformat{42.468}	&		\timeout\\ 
		\name{TwoCountersRefinedRefined} & \sepcomma{1933}	&	\sepcomma{3140} & 3 &  \secformat{0.209}	&	\secformat{43.752}\\ 
		\name{TwoCountersRefined} & \sepcomma{1961}	&	\sepcomma{3196} & 3 &  \secformat{0.217}	&	\secformat{45.308}\\ 
		\name{TwoCounters} & \sepcomma{1216}	&	\sepcomma{1970} & 3 &  \secformat{0.086}	&	\secformat{14.067}\\ 
		\name{UnderapproxDemo2} & \sepcomma{42}	&	\sepcomma{60} & 3 &  \secformat{0.003}	&	\secformat{0.012}\\ 
		\name{UnderapproxDemo} & \sepcomma{20}	&	\sepcomma{24} & 3 &  \secformat{0.003}	&	\secformat{0.007}\\ 
		\name{Zoo10} & \sepcomma{479}	&	\sepcomma{768} & 3 &  \secformat{0.02}	&	\secformat{2.868}\\ 
		\name{Zoo5} & \sepcomma{479}	&	\sepcomma{768} & 3 &  \secformat{0.021}	&	\secformat{2.872}\\ 
		\hline
	\end{longtable}

\end{document}